\documentclass[11 pt, onecolumn,showpacs,notitlepage,nofootinbib,floatfix,superscriptaddress]{revtex4-1} 
\usepackage{graphicx,xcolor,soul}
\usepackage{amssymb,amsmath,amsthm,amsfonts}
\usepackage[ruled]{algorithm2e}
\usepackage{hyperref} 
\usepackage{tikz-cd}

\newlength{\fighskip} \fighskip=2pt
\newlength{\figvskip} \figvskip=3pt
\newcommand*{\figbox}[2]{{\def\figscale{#1}\def\arraystretch{0.8}\arraycolsep=0pt\begin{array}{c}\vbox{\vskip\figscale\figvskip\hbox{\hskip\figscale\fighskip\includegraphics[scale=\figscale]{#2}}}\end{array}}}

\newtheorem{lemma}{Lemma}
\newtheorem{theorem}{Theorem}

\renewcommand{\phi}{\varphi}
\renewcommand{\hat}{\widehat}
\newcommand{\e}{\mathfrak{e}}
\newcommand{\cl}{\mathfrak{C}}
\newcommand{\cub}{\mathcal{L}_\mathrm{cub}^*}
\newcommand{\ket}[1]{|#1\rangle}
\newcommand{\fix}[1]{#1+\hat#1} 
\newcommand{\N}{\mathcal{N}}

\newcommand{\M}{\mathcal{M}}

\newcommand{\R}{\mathcal{R}}
\DeclareMathOperator*{\fail}{fail}
\DeclareMathOperator*{\supp}{supp}

\DeclareMathOperator*{\argmin}{arg\,min}
\DeclareMathOperator*{\im}{im}
\newcommand{\psus}{$p_{\mathrm{STC}} \approx 1.045\%$}

\begin{document}

\title{Single-shot quantum error correction \\with the three-dimensional subsystem toric code}
\author{Aleksander Kubica}
\affiliation{Perimeter Institute for Theoretical Physics, Waterloo, ON N2L 2Y5, Canada}
\affiliation{Institute for Quantum Computing, University of Waterloo, Waterloo, ON N2L 3G1, Canada}
\affiliation{AWS Center for Quantum Computing, Pasadena, CA 91125, USA}
\affiliation{California Institute of Technology, Pasadena, CA 91125, USA}
\author{Michael Vasmer}
\affiliation{Perimeter Institute for Theoretical Physics, Waterloo, ON N2L 2Y5, Canada}
\affiliation{Institute for Quantum Computing, University of Waterloo, Waterloo, ON N2L 3G1, Canada}

\begin{abstract}
We introduce a new topological quantum code, the three-dimensional subsystem toric code (3D~STC), which is a generalization of the stabilizer toric code.
The 3D~STC can be realized by measuring geometrically-local parity checks of weight at most three on the cubic lattice with open boundary conditions.
We prove that single-shot quantum error correction (QEC) is possible with the 3D~STC, i.e., one round of local parity-check measurements suffices to perform reliable QEC even in the presence of measurement errors.
We also propose an efficient single-shot QEC strategy for the 3D~STC and investigate its performance.
In particular, we numerically estimate the resulting storage threshold against independent bit-flip, phase-flip and measurement errors to be \psus.
Such a high threshold together with local parity-check measurements of small weight make the 3D~STC particularly appealing for realizing fault-tolerant quantum computing.
\end{abstract}

\maketitle

Building reliable and scalable universal quantum computers is a heroic endeavor~\cite{Nigg2014,Barends2014,Corcoles2015,Ofek2016,Arute2019}, which requires the implementation of fault-tolerant protocols~\cite{Shor1996,Steane1997,Knill2005,Knill2005a}.
Even the substantially simpler task of storing quantum information is very challenging and requires the usage of quantum error correction (QEC) techniques to detect and eliminate faults.
Due to unreliable physical components, QEC is itself a noisy process, which, if carried out haphazardly, can destroy encoded logical information.
Nevertheless, QEC together with fault-tolerant gadgets to implement logical gates on the encoded information allow one, in principle, to perform arbitrary long quantum computations provided the noise affecting the system is below some constant threshold value~\cite{Aharonov1997,Kitaev1997,Knill1998,Aliferis2005}.
However, questions about the practicality, noise tolerance and resource requirements for various realizations of universal quantum computation still remain topics of active research~\cite{Fowler2012,Karzig2017,Litinski2019,Chao2020,Chamberland2020arch,Guillaud2021,Beverland2021}.

Topological quantum error-correcting codes~\cite{Dennis2002,Bombin2013book} provide a realistic and resource-efficient approach to building scalable quantum computers.
Codes in this class have desirable features, such as efficient classical decoding algorithms with high storage thresholds and fault-tolerant logical gates with low overhead.
Importantly, topological quantum codes can be realized by placing qubits on geometrical lattices and measuring only geometrically-local parity checks.
We emphasize that locality is critical not only from the perspective of fault tolerance, but also from the fact that the physical interactions that we can engineer have a local nature.
To experimentally realize topological quantum codes, we are restricted to at most three spatial dimensions (unless we allow non-local connections between qubits, which can effectively boost the dimensionality of the system).

The archetypal topological quantum code, the toric code~\cite{Kitaev2003,Bravyi1998}, can be engineered in two spatial dimensions.
There has been a lot of effort devoted to realizing the toric code, both from the theory as well as the experimental side.
In the presence of measurement errors, one way to perform reliable QEC with the toric code is to use a simple fault-tolerant method to extract the syndrome---it suffices to repeat parity-check measurements to gain confidence in their outcomes~\cite{Dennis2002,Fowler2012,Fowler2012proof}
Unfortunately, the number of measurement rounds necessarily grows with the code size, and thus the penalty one pays is the increased time overhead and the need to store measurement outcomes.
Subsequently, QEC extends over time, the system effectively becomes (2+1)D, and the overall fault-tolerant protocols become more complicated.

\begin{figure*}[ht!]
\centering
\includegraphics[width=.31\textwidth]{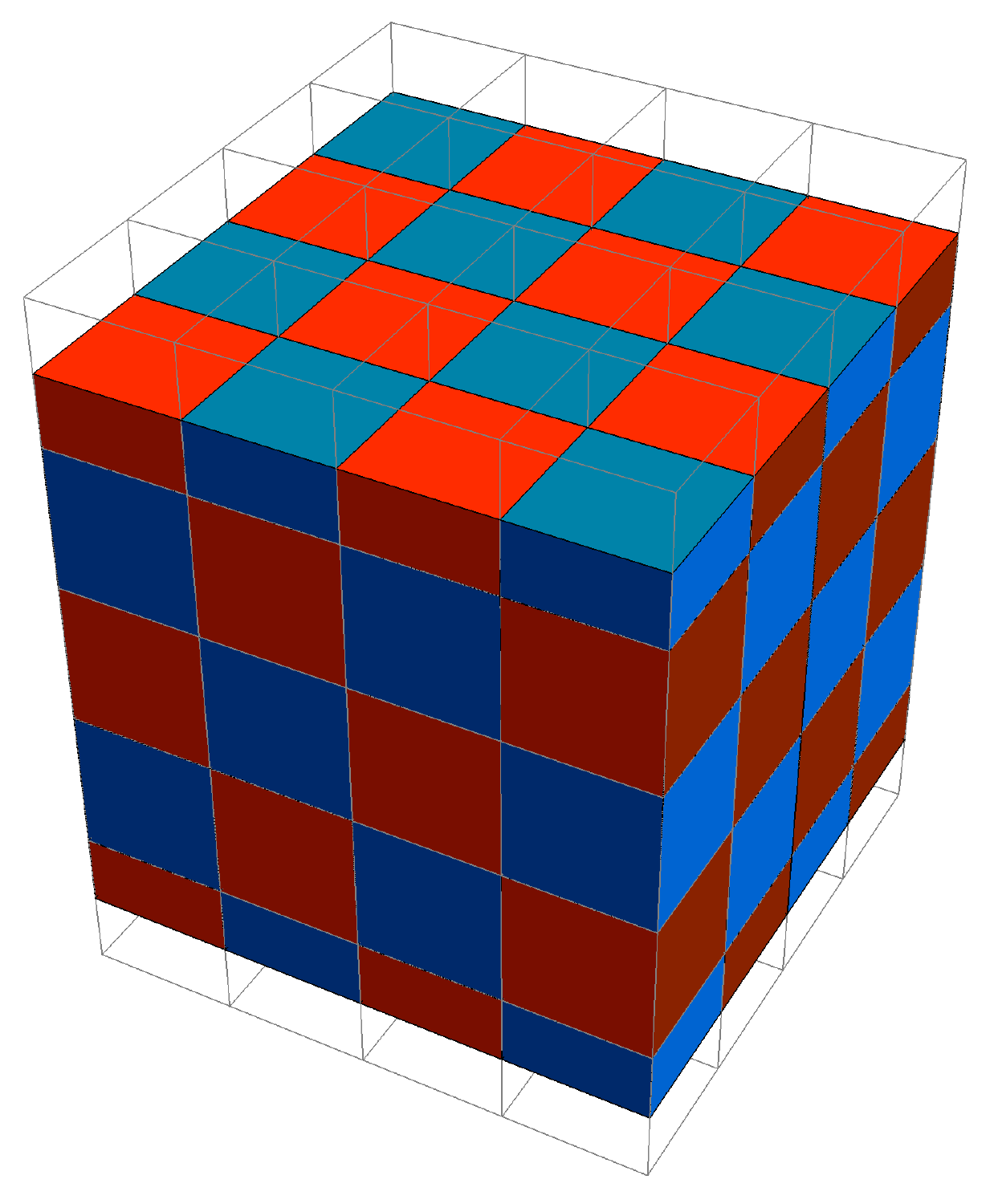}
\hspace*{-.28\textwidth}(a)\hspace*{.26\textwidth}
\includegraphics[width=.31\textwidth]{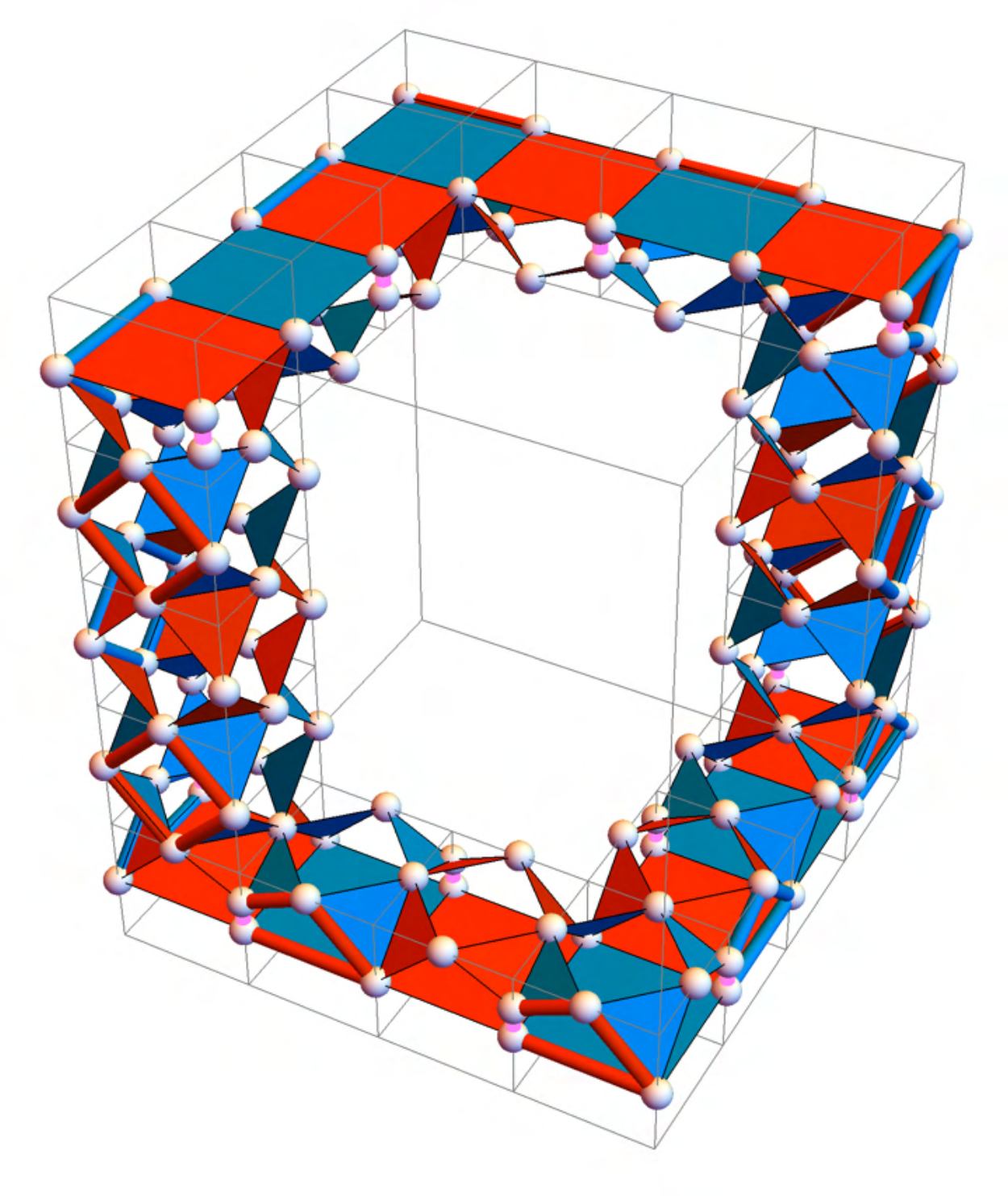}
\hspace*{-.28\textwidth}(b)\hspace*{.26\textwidth}
\includegraphics[width=.31\textwidth]{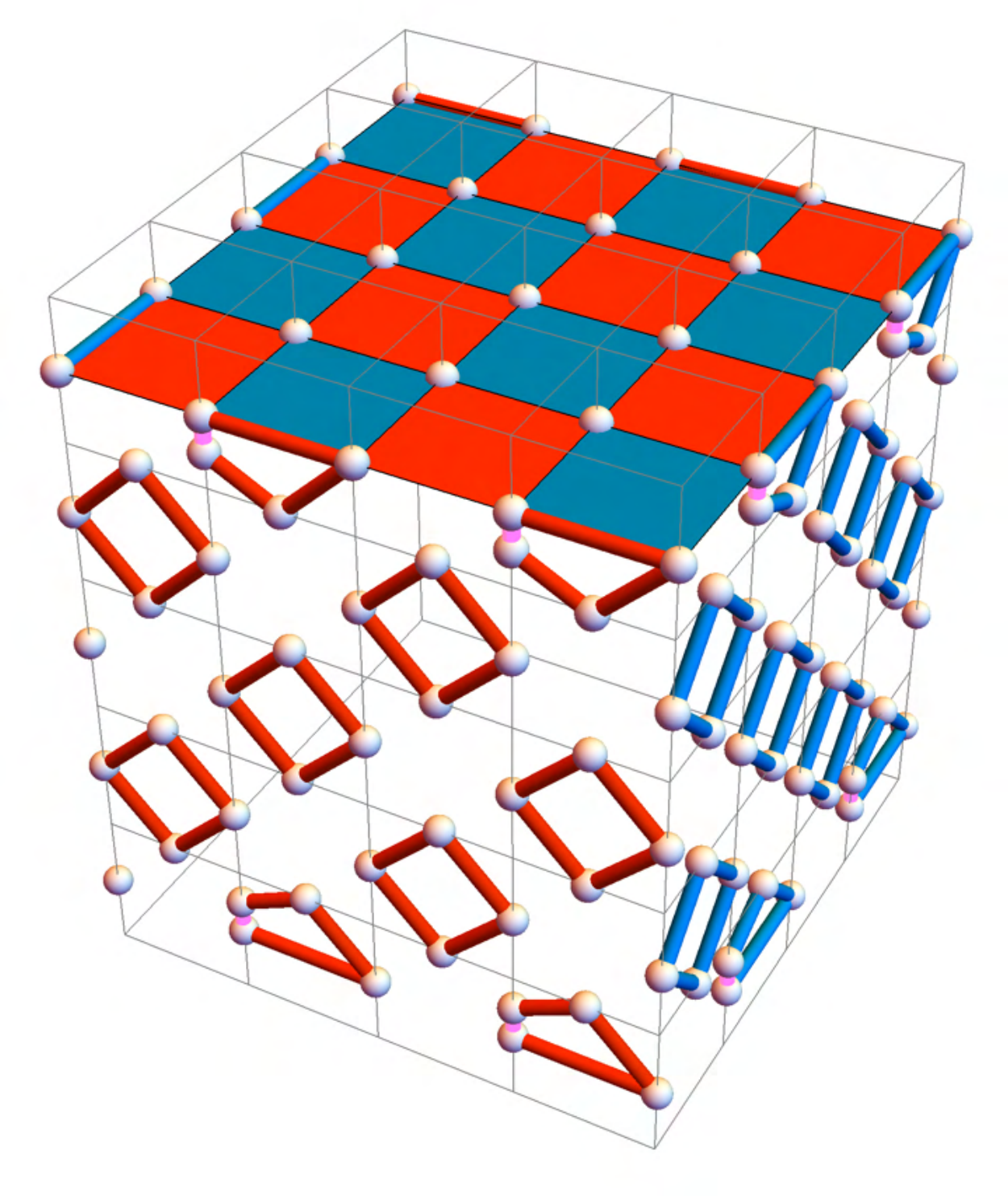}
\hspace*{-.285\textwidth}(c)\hspace*{.27\textwidth}
\caption{
(a) The cubic lattice with open boundary conditions and the linear size $L = 4$.
Its cubic volumes can be colored in red and blue in a checkerboard pattern.
(b)-(c) We define the 3D~STC on the lattice $\cub$ (described in Sec.~\ref{sec_glance}) according to Eqs.~\eqref{eq_3DSTC_gauge}~and~\eqref{eq_3DSTC_gauge_top}.
We illustrate the support of some of the $X$- and $Z$-type gauge operators (red and blue shapes).
White dots represent qubits, and pink edges near the top and bottom boundaries represent the weight-two $X$- and $Z$-type gauge operators.
}
\label{fig_3DSTC}
\end{figure*}

Recently, a lot of effort has been devoted into developing QEC techniques that do not require repeated rounds of parity-check measurements in the presence of measurement errors~\cite{Fujiwara2014,Campbell2019,Ashikhmin2020,Delfosse2020,Fawzi2021}.
Rather, these techniques rely on a careful choice which parity checks to measure, as well as on redundancies among the measurement outcomes due to the choice of an overcomplete set of parity checks.
Unfortunately, if these techniques are applied to the 2D toric code, then the geometric locality of the system is lost, as we would need to measure some high-weight non-local parity checks, which is a serious limitation.
Thus, an intriguing question is whether in the presence of measurement errors one can perform reliable QEC with the toric code with only local parity-check measurements, without repeated rounds of measurements.

In this article we propose a radically different realization of the toric code capable of handling measurement errors, which relies on single-shot QEC discovered by Bomb\'in~\cite{Bombin2015}.
Intuitively, single-shot QEC guarantees that one can perform reliable QEC without repeating (geometrically-local) parity-check measurements.
In order to achieve this, we introduce a subsystem version of the toric code, the three-dimensional subsystem toric code (3D~STC).
Remarkably, the 3D~STC is a topological quantum code that can be realized on the cubic lattice with open boundary conditions and low-weight parity-check measurements; see Fig.~\ref{fig_3DSTC}.
Due to its simplicity, the 3D~STC can be viewed as the quintessential topological code demonstrating single-shot QEC.

Our work comprises three parts.
First, in Sec.~\ref{sec_model} we introduce the 3D~STC.
Then, in Sec.~\ref{sec_ss_decoding} we develop a single-shot decoding algorithm for the 3D~STC and numerically estimate its performance.
Lastly, in Sec.~\ref{sec_proof} we prove that single-shot QEC is indeed possible with the 3D~STC. 
In what follows we briefly describe the context of each of our main contributions.

\vspace*{-2pt}
\subsection*{Model}
\vspace*{-2pt}

Topological quantum error-correcting codes are alluring not only from the perspective of QEC, but also from the perspective of quantum many-body physics.
Even the simplest topological quantum codes, which belong to the class of stabilizer codes~\cite{Gottesman1996} or their slight generalization, subsystem codes~\cite{Poulin2005}, illustrate a variety of physical concepts.
For instance, the the ground state of the 2D toric code is an epitome of a topologically-ordered state.
In three dimensions, Chamon's model~\cite{Chamon2005} and the cubic code~\cite{Haah2011} provide concrete realizations of exotic quantum phases of matter with fractal-like excitations.
In four dimensions, the 4D toric code gives rise to a local commuting Hamiltonian that exhibits the phenomenon of self-correction~\cite{Alicki2010,Brown2014}.
Analogously, the 3D~STC, which we introduce, demonstrates single-shot QEC.

The 3D~STC is similar to another topological quantum code, the 3D gauge color code~\cite{Bombin2013}.
The 3D~STC and 3D gauge color code constitute subsystem versions of the stabilizer toric code and the stabilizer color code~\cite{Bombin2006,Bombin2007,Kubicathesis}, respectively.
Both codes facilitate single-shot QEC, as well as a fault-tolerant universal gate set without magic state-distillation~\cite{Bombin2013,Kubica2015a,Iverson2021}.
However, the 3D~STC is more appealing due to its simplicity---it can be realized on the cubic lattice with open boundary conditions by measuring geometrically-local parity checks of weight at most three.
This should be contrasted with the known realizations of the 3D gauge color code, which require parity checks of weight at least six.
Subsequently, the 3D~STC provides significantly better protection from errors than the 3D gauge color code~\cite{Brown2015,Kubica2017}.

Despite the close connection between the toric code and the color code in $d\geq 2$ dimensions~\cite{Kubica2015}, a genuine subsystem generalization of the toric code was not known.
The 3D~STC provides such a generalization.
Although in the main text we focus on the three-dimensional case, our construction is more general and works in any dimension $d\geq 3$; see Appendix~\ref{app_ddim} for details.
Lastly, we remark that in the special case of two dimensions there exist realizations of the toric code as a subsystem code~\cite{Bravyi2013sub,Higgott2020}.
However, by applying constant-depth circuits composed of geometrically-local gates one can remove the gauge qubits and effectively map those models to the stabilizer toric code.
This, however, is not possible for the 3D~STC, and therefore the 3D~STC is the first bona fide subsystem version of the toric code.

\vspace*{-2pt}
\subsection*{Single-shot decoding}
\vspace*{-2pt}

Active QEC comprises the detection and correction of errors.
For stabilizer and subsystem codes, we first measure parity checks corresponding to some stabilizer and gauge operators, respectively.
Then, using the obtained classical information we infer the stabilizer syndrome, i.e., the set of all stabilizers returning $-1$ measurement outcome.
Note that in order to reliably infer the stabilizer syndrome in the presence of measurement errors one may have to repeat parity-check measurements, which in turn results in a significant time overhead, as exemplified by the 2D toric code.
Lastly, we use classical decoding algorithms and find an appropriate recovery for the given stabilizer syndrome.

Performing optimal QEC for generic stabilizer codes is a computationally hard task even in the absence of measurement errors~\cite{Iyer2015}.
However, for topological quantum codes there exist various decoding algorithms with good performance, many of which rely on solving the minimum-weight perfect-matching (MWPM) problem~\cite{Dennis2002,Delfosse2014,Nickerson2017,Kubica2019}.
The MWPM problem, roughly speaking, is the task of pairing some subset of the vertices of a given graph, which, importantly, can be solved efficiently~\cite{Edmonds1965}.

An efficient QEC strategy for the 3D~STC, which we propose, requires only one round of parity-check measurements and works reliably even in the presence of measurement errors.
Our QEC strategy, which we name the single-shot MWPM decoder, consists of two steps: (i) syndrome estimation and (ii) ideal MWPM decoding.
Remarkably, both steps can be reduced to the MWPM problem, which should be contrasted with alternative approaches to single-shot QEC~\cite{Brown2015,Duivenvoorden2017,Breuckmann2020,Quintavalle2020,Grospellier2021}.
We numerically benchmark the performance of the single-shot MWPM decoder against bit-flip and phase-flip noise in the presence of measurement errors and estimate the storage threshold to be \psus.
Since the parity checks are of weight at most three, we expect that for the circuit-level noise the storage threshold will not be substantially reduced and the overall performance of the 3D~STC will be on par with the toric code on the square lattice.

\vspace*{-2pt}
\subsection*{Proof of single-shot QEC}
\vspace*{-2pt}

Local operations play a central role in fault-tolerant quantum computation, as they preserve the local structure of noise, and thus are trivially fault-tolerant.
However, strictly local operations, such as transversal gates or cellular-automata decoders, are limited---the computational power of the former is restricted~\cite{Eastin2009,Zeng2011,Bravyi2013,Pastawski2014,Beverland2014,Jochym-OConnor2018,Webster2020,Faist2019,Woods2020,Kubica2020}, whereas the latter require that the syndrome has some underlying structure~\cite{Ahn2004,Breuckmann2017,Kubica2018toom,Vasmer2020}.

One way to avoid the aforementioned limitations is to consider quantum-local operations~\cite{Bombin2015}, i.e., local operations that depend on classical information stored only for a limited time.
Such processes are physically-motivated, as quantum operations are typically constrained by geometrically-local interactions, whereas classical information can be processed globally in a reliable way.
Examples of quantum-local operations include the procedure of gauge fixing~\cite{Paetznick2013,Anderson2014} and the single-shot MWPM decoder, which allow for, respectively, a fault-tolerant universal gate set without magic state distillation and single-shot QEC.
Unfortunately, quantum-local operations are not a priori fault-tolerant, as they are not guaranteed to preserve the local structure of noise.

In this work, we prove that the single-shot MWPM decoder for the 3D~STC is indeed fault-tolerant.
In other words, performing repeated rounds of error correction in the presence of measurement errors will not lead to the uncontrollable accumulation of errors and the logical information encoded in the 3D~STC will be protected for a long time.
This, in turn, rigorously establishes that the 3D~STC allows for QEC in a single-shot manner, which drastically reduces the time overhead associated with QEC.

\section{Model}
\label{sec_model}

We start this section by presenting a simple and concrete realization of the 3D~STC on the lattice $\cub$, which is based on the cubic lattice with open boundary conditions.
The resulting 3D~STC has one logical qubit and code distance proportional to the linear size of $\cub$.
Then, we discuss a systematic construction of the 3D~STC and show that that for any orientable closed 3-manifold it has zero logical qubits\footnote{Note that the 3D gauge color code behaves alike, i.e., for any orientable closed 3-manifold it has zero logical qubits.}.
Lastly, we explain how the 3D~STC can be viewed as a generalization of the 3D stabilizer toric code.

\subsection{First glance at the model}
\label{sec_glance}

A subsystem code is a generalization of a stabilizer code.
Intuitively, a subsystem code is like a stabilizer code, except we only encode quantum information into a subset of the qubits in the stabilizer subspace.
We refer to the encoded qubits in this subset as the logical qubits; the remaining qubits are called the gauge qubits.
A subsystem code is specified by its gauge group $\mathcal{G}$, which is a subgroup of the Pauli group $\mathcal{P}$ that may contain $-I$.
Let $\mathcal{Z}(\mathcal{G})$ denote the centralizer of $\mathcal{G}$ in the Pauli group $\mathcal{P}$, i.e., all the elements in $\mathcal{P}$ that commute with every element in $\mathcal{G}$.
Ignoring the phases, we refer to the center of the gauge group $\mathcal{G}$ as the stabilizer group $\mathcal{S}$, i.e.,
$\mathcal{S} = (\mathcal{Z}(\mathcal{G}) \cap \mathcal G )/ \langle i \rangle$.
Whenever the gauge group $\mathcal{G}$ does not contain $-I$, it can be viewed as the stabilizer group defining a stabilizer code.
Lastly, we say that a stabilizer or a subsystem code is a CSS code~\cite{Calderbank1996,Steane1996CSS}
if its generators can be chosen as either Pauli $X$ or $Z$ operators.

The 3D~STC is a topological quantum code, and is also a CSS subsystem code.
We start with a simple and concrete realization of the 3D~STC on the lattice $\cub$, which is based on the cubic lattice with open boundary conditions; see Fig.~\ref{fig_3DSTC}(a).
The cubic volumes of the cubic lattice can be colored in red and blue in a checkerboard pattern.
In the bulk of the lattice, we place one qubit on every edge.
For every red and blue volume we introduce eight weight-three $X$- and $Z$-type gauge operators, respectively, as follows
\begin{equation}
\label{eq_3DSTC_gauge}
\figbox{.3}{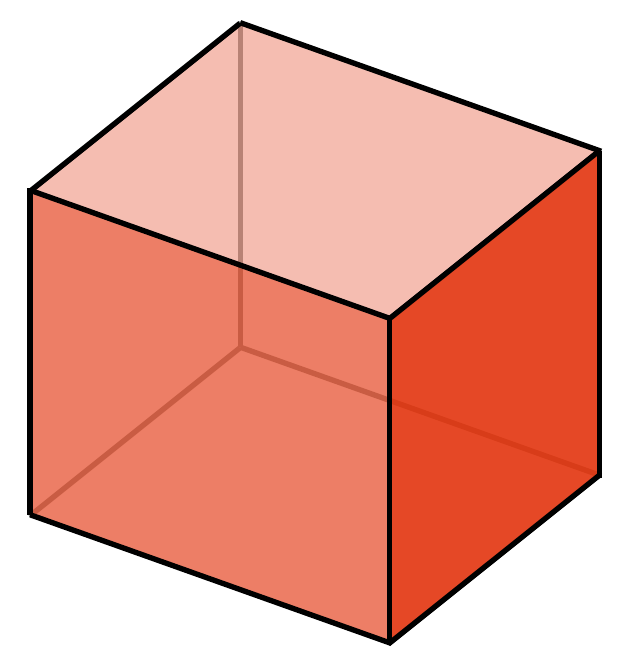}  \longmapsto \figbox{.3}{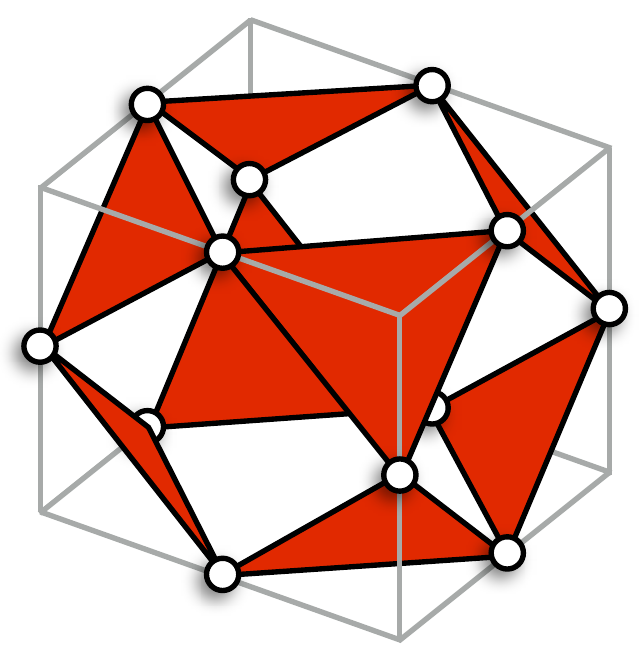},\quad\quad\quad
\figbox{.3}{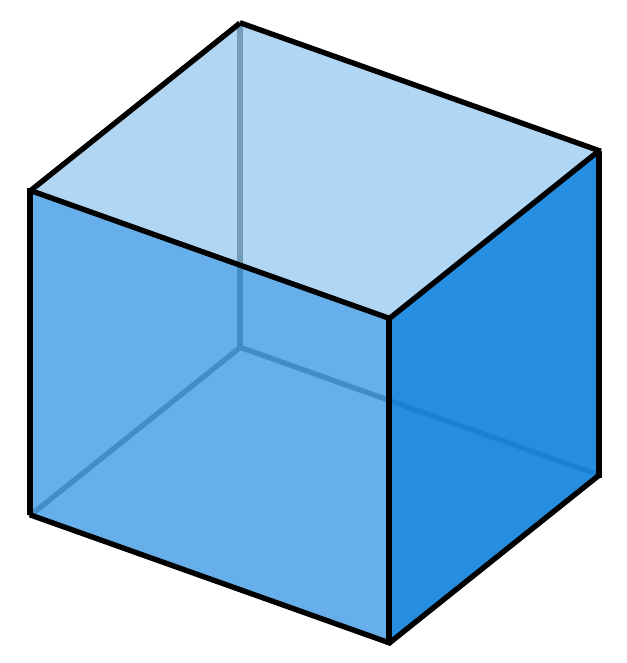}  \longmapsto \figbox{.3}{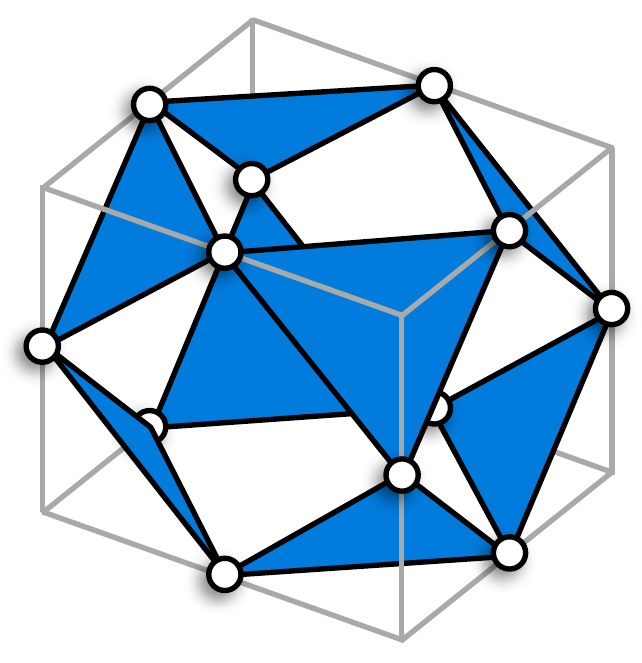},
\end{equation}
where we depict the support of operators as red and blue triangles, respectively.
Then, weight-twelve $X$- and $Z$-type stabilizer operators, which we associate with red and blue volumes, can be formed in two different ways by multiplying gauge operators from Eq.~\eqref{eq_3DSTC_gauge}, namely
\begin{equation}
\label{eq_3DSTC_stab}
\figbox{.3}{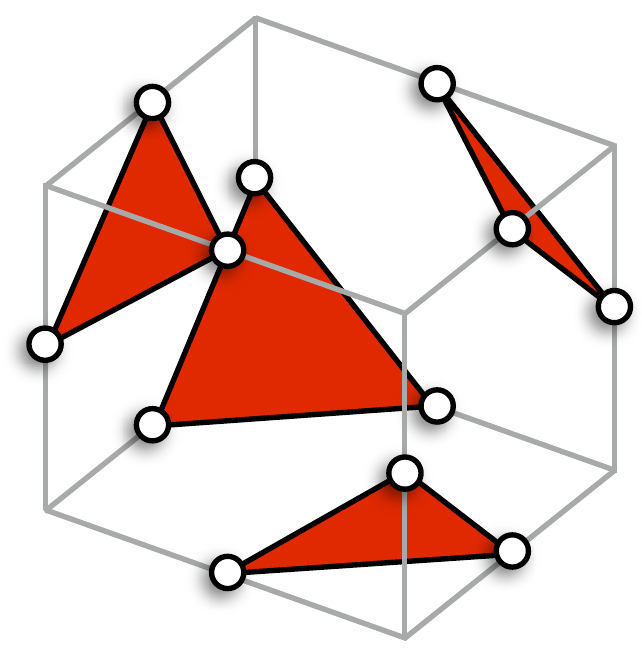}  = \figbox{.3}{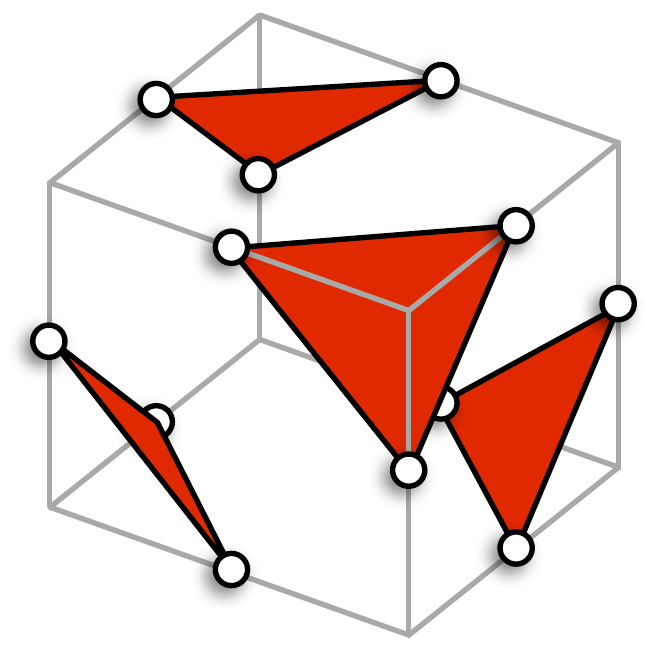},\quad\quad\quad
\figbox{.3}{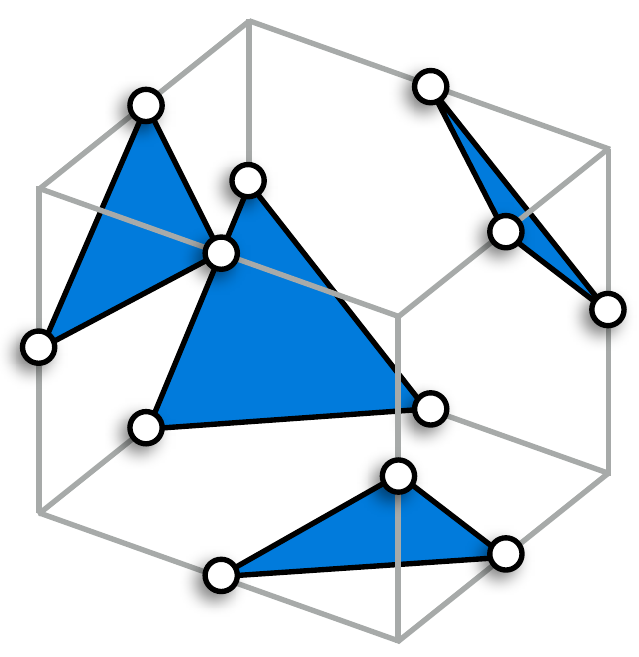}  = \figbox{.3}{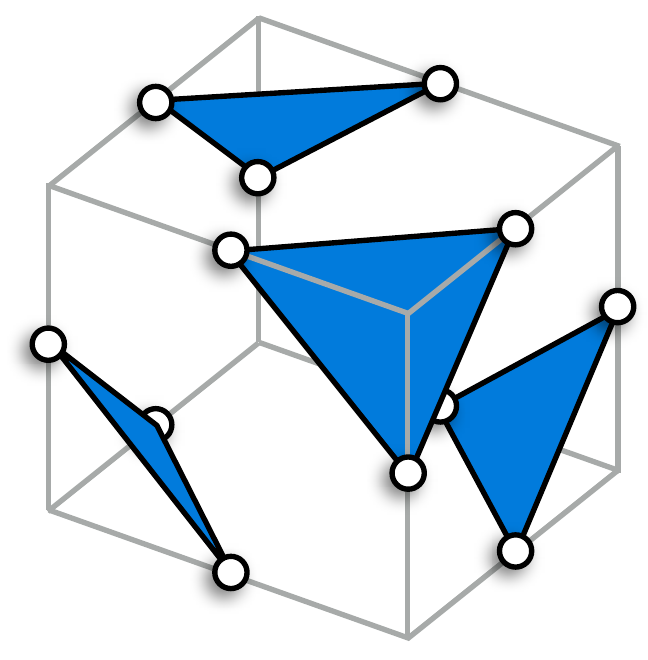}.
\end{equation}
Near the top boundary of the lattice, we place one additional qubit on every other vertical edge and introduce seven $X$- and $Z$-type gauge operators for every red and blue volume as follows
\begin{equation}
\label{eq_3DSTC_gauge_top}
\figbox{.3}{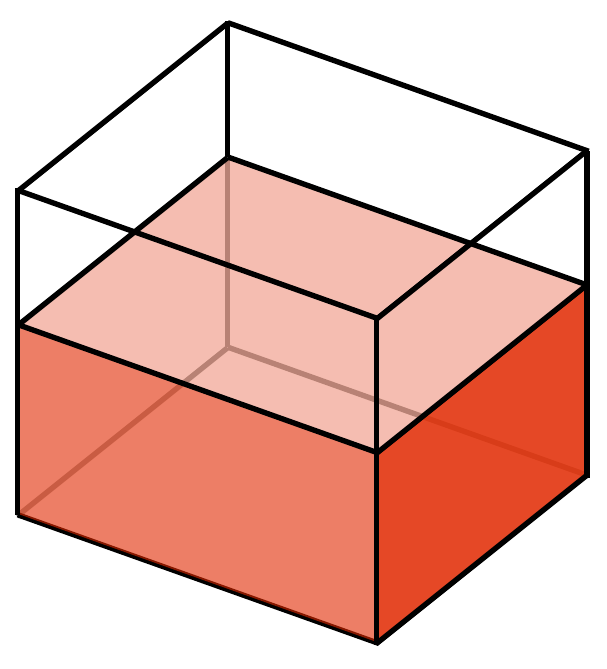}  \longmapsto \figbox{.3}{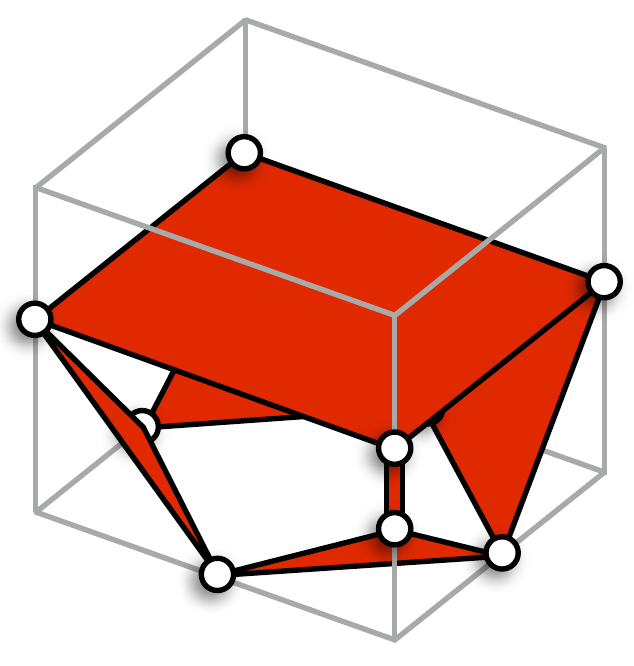},\quad\quad\quad
\figbox{.3}{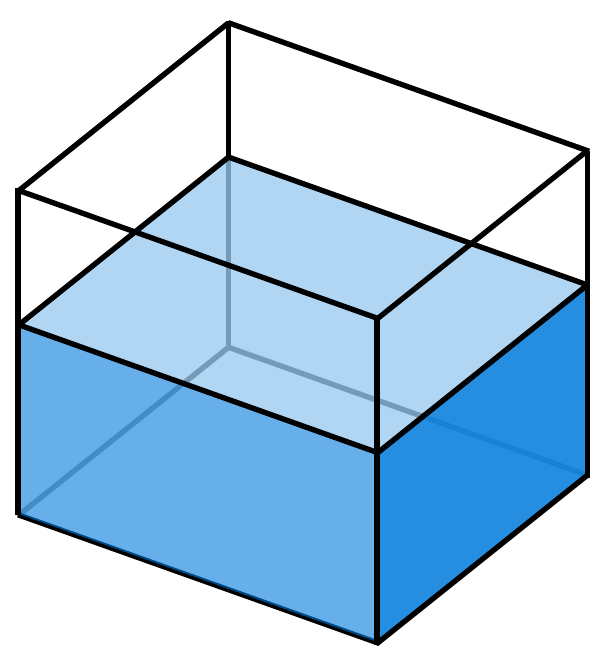}  \longmapsto \figbox{.3}{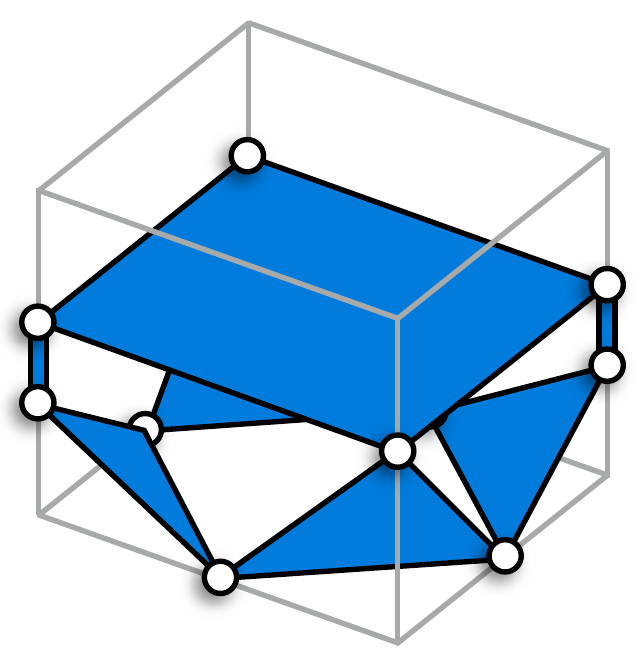}.
\end{equation}
Similarly, weight-ten $X$- and $Z$-type stabilizer operators can be formed in two different ways by multiplying gauge operators from Eq.~\eqref{eq_3DSTC_gauge_top}, namely
\begin{equation}
\label{eq_3DSTC_stab_top}
\figbox{.3}{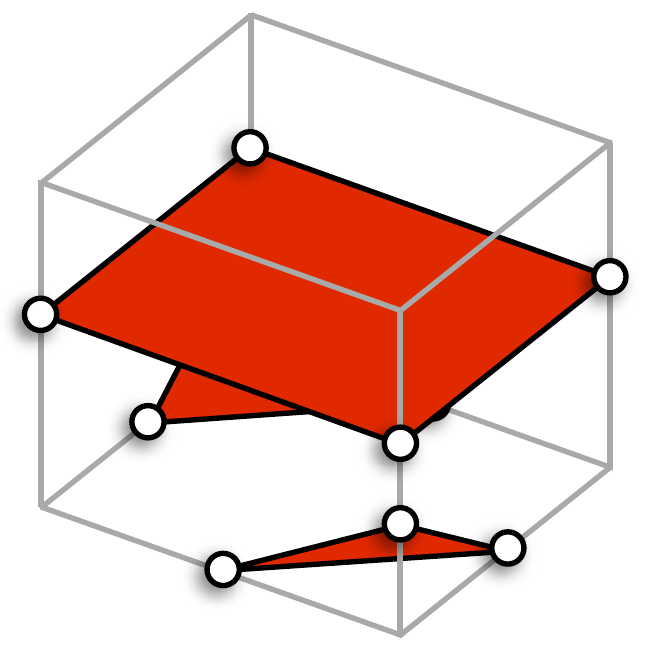}  = \figbox{.3}{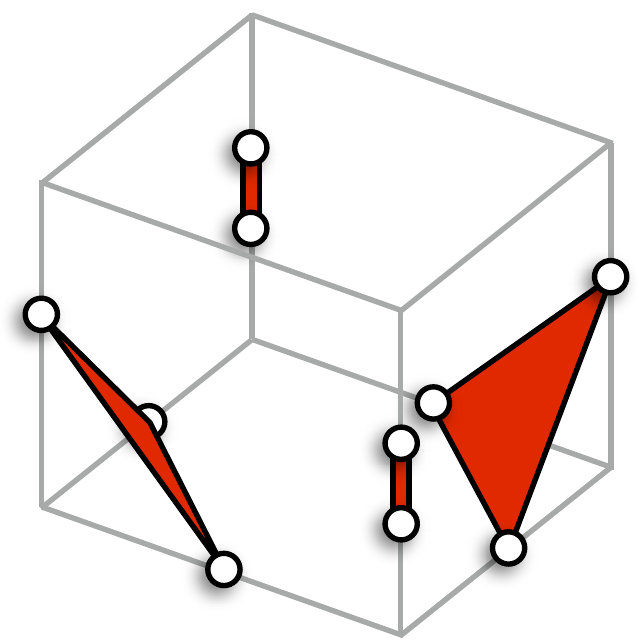},\quad\quad\quad
\figbox{.3}{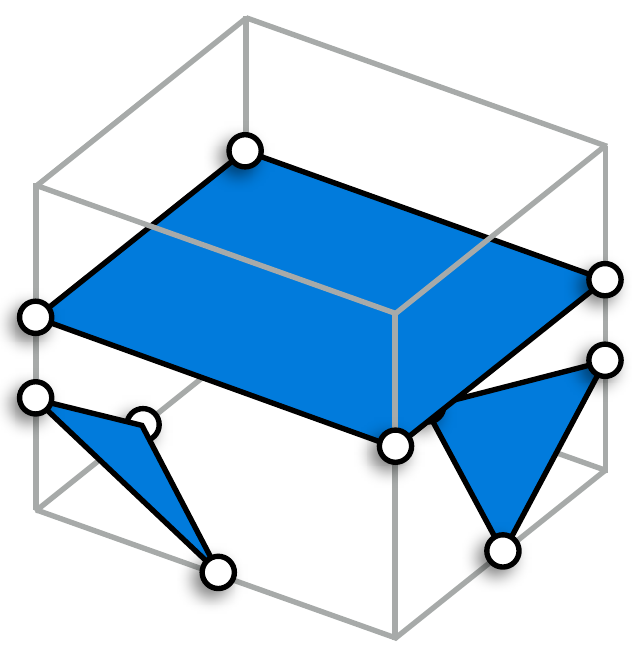}  = \figbox{.3}{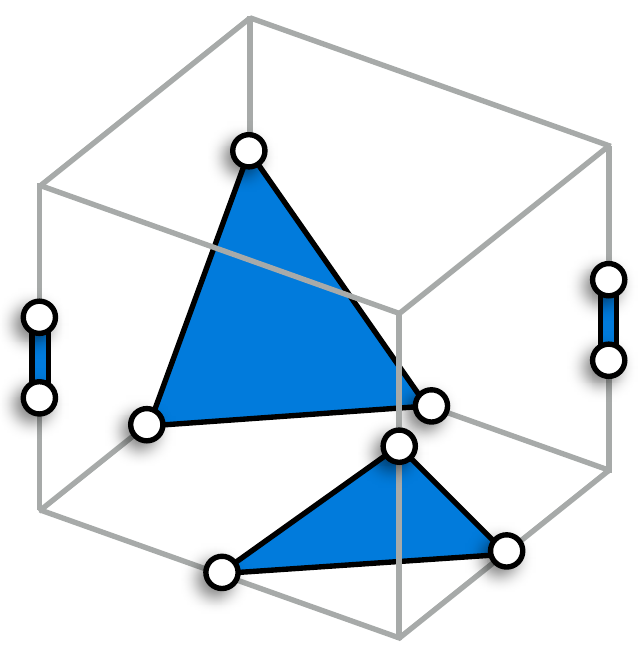}.
\end{equation}
Lastly, on every outward-facing side of a blue (respectively, red) volume along the front (right) boundary of the lattice we introduce four weight-two $X$-type ($Z$-type) gauge operators, which we can further multiply in two different ways to form a weight-four $X$-type ($Z$-type) stabilizer operator.
We remark that the bottom, rear and left boundaries are the same as the top, front and right boundaries, respectively, as the lattice is invariant under point reflection through its center.
We illustrate some of the gauge generators of the resulting 3D~STC on the lattice $\cub$ in Fig.~\ref{fig_3DSTC}(b)(c).
Throughout the article we use red and blue to depict the support of Pauli $X$ and $Z$ operators, respectively.

\begin{figure*}[ht!]
\centering
(a)\includegraphics[height=.18\textheight]{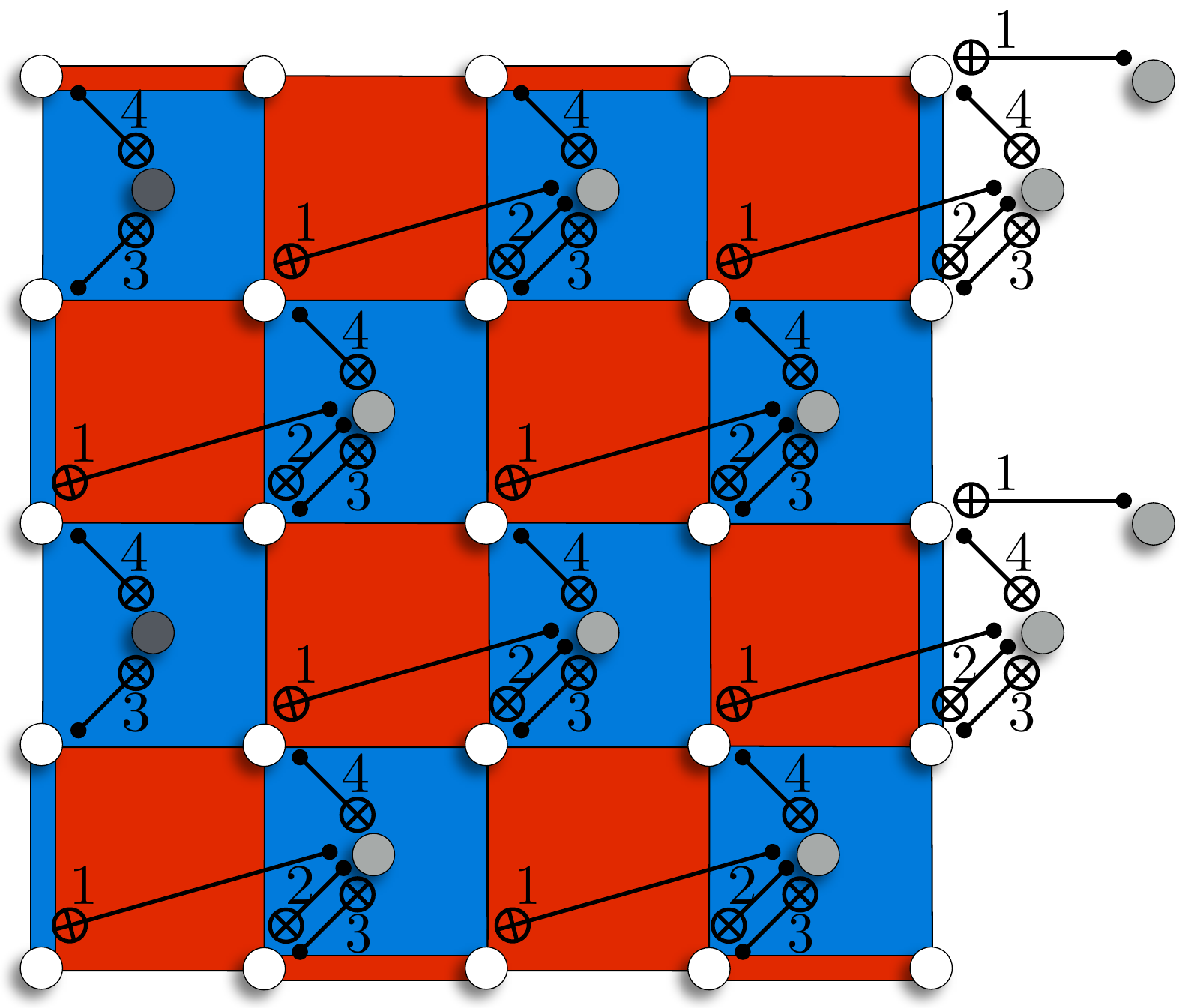}
\quad\quad
(b) \includegraphics[height=.18\textheight]{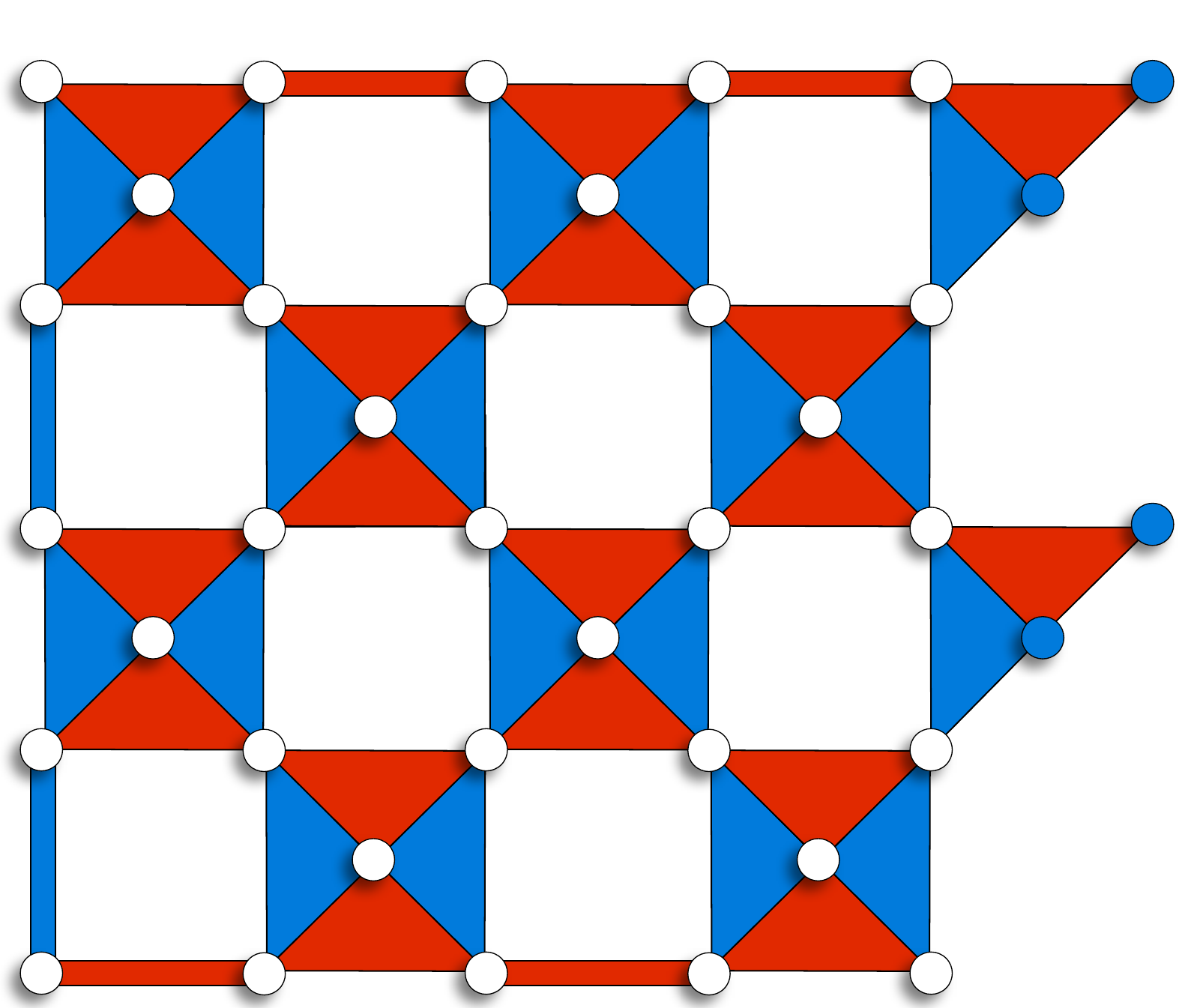}
\caption{
(a) To reduce the weight of weight-four gauge operators along the top and bottom boundaries of the lattice $\cub$ we introduce ancilla qubits $\mathcal A_S$, which are in the state $\ket{0}$ (dark gray dots), as well as gauge qubits $\mathcal A_G$ (light gray dots).
We then implement a unitary $U$ via a circuit comprising four rounds of CNOT gates.
The circuit has depth three as CNOT gates in rounds two and three can be combined.
(b) The choice of gauge generators of the subsystem code $U\mathcal G_\text{aug} U^\dag$ that are fully supported within the top and bottom boundaries.
Some qubits along the right side (blue dots) support single-qubit Pauli $Z$ operators.
}
\label{fig_2DSTC}
\end{figure*}

The bare and dressed logical Pauli operators of the subsystem code $\mathcal G$ are defined as the elements of $\mathcal{Z}(\mathcal{G})$ and $\mathcal{Z}(\mathcal{S})$, respectively.
Dressed logical operators, unlike bare logical operators, may change the state of gauge qubits, however no logical information is encoded into them.
Similarly to the logical operators of the 3D stabilizer toric code, bare and dressed logical Pauli operators of the 3D~STC can be expressed as 2D sheet-like and 1D string-like operators connecting opposite boundaries of the lattice $\cub$.
Moreover, bare logical Pauli $X$ and $Z$ operators can be fully supported within, respectively, the front or rear and left or right side boundaries.
If the linear size of the lattice $\cub$ is $L$, then the code parameters of the 3D~STC on $\cub$ are
\begin{equation}
\label{eq_code_params}
[\![ N = L^3+6L^2+5L+1, K = 1, D = L+1]\!],
\end{equation}
where $N$ and $K$ are the number of physical and logical qubits, and $D$ is the code distance (defined as the weight of the smallest dressed logical operator).
The number of logical qubits of the the 3D~STC defined on $\mathcal L^*_{\mathrm{cub}}$ can be derived from the general method for constructing the 3D~STC presented in Sec.~\ref{sec_octahedral}-\ref{sec_cubiclattice}.

\subsection{Reducing the weight of gauge operators}

The realization of the 3D~STC on the lattice $\cub$ is very simple, however the gauge group $\mathcal G$ of the 3D~STC cannot be generated by gauge operators of weight at most three.
Namely, along the top and bottom boundaries of $\cub$ there are gauge operators of weight four; see Fig.~\ref{fig_2DSTC}(c).
Ideally, we would like all gauge operators to be of weight at most three, as in such a case there would be no hook errors introduced in the noisy measurement circuits~\cite{Dennis2002}.
We reduce the weight of the gauge operators by using an idea from Ref.~\cite{Bravyi2013sub}, where a version of the 2D toric code with only weight-three operators was proposed.

\begin{figure*}[ht!]
\centering
\includegraphics[height=.18\textheight]{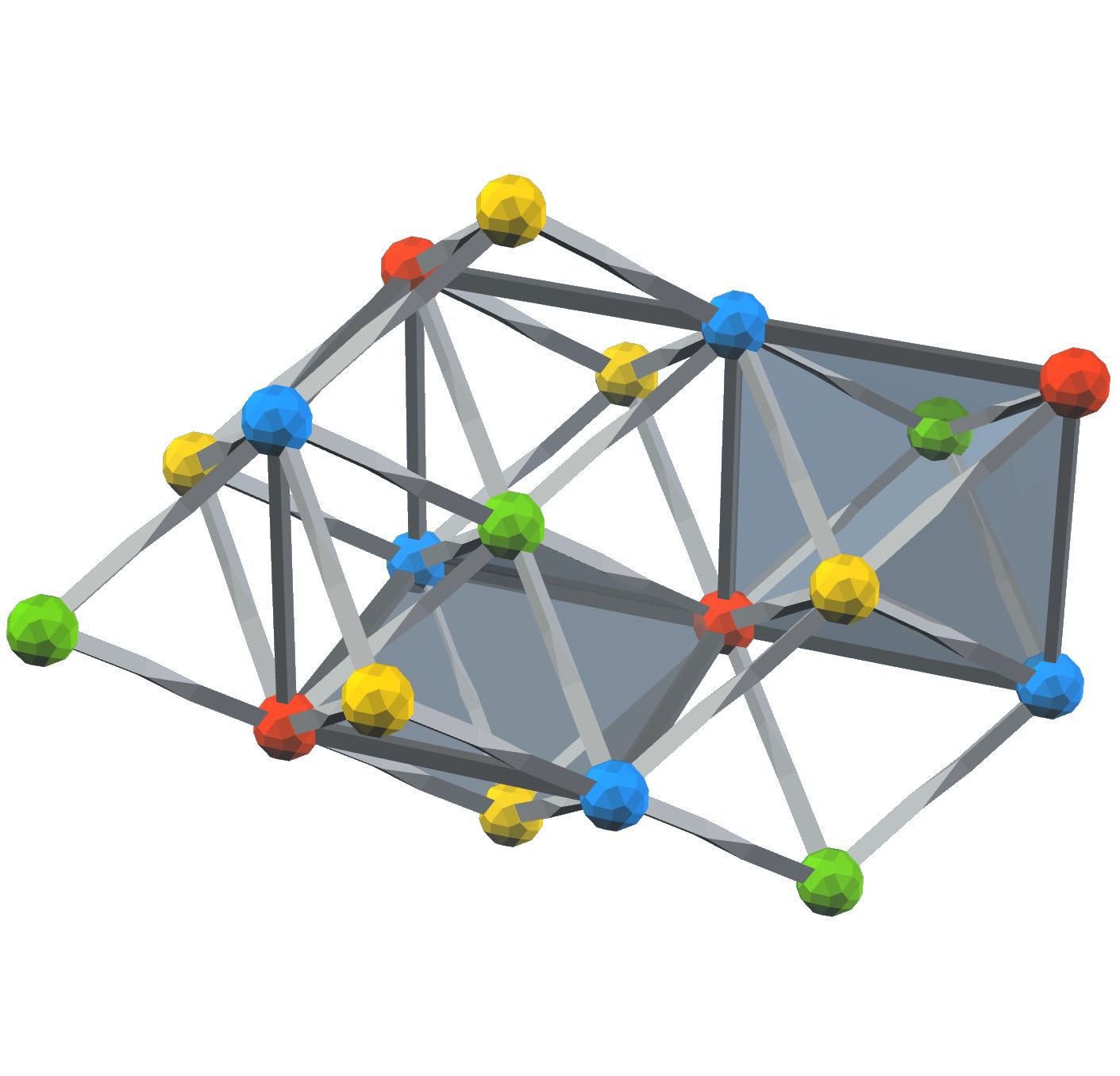}
\hspace*{-.27\textwidth}(a)\hspace*{.25\textwidth}
\quad\quad\quad
\includegraphics[height=.17\textheight]{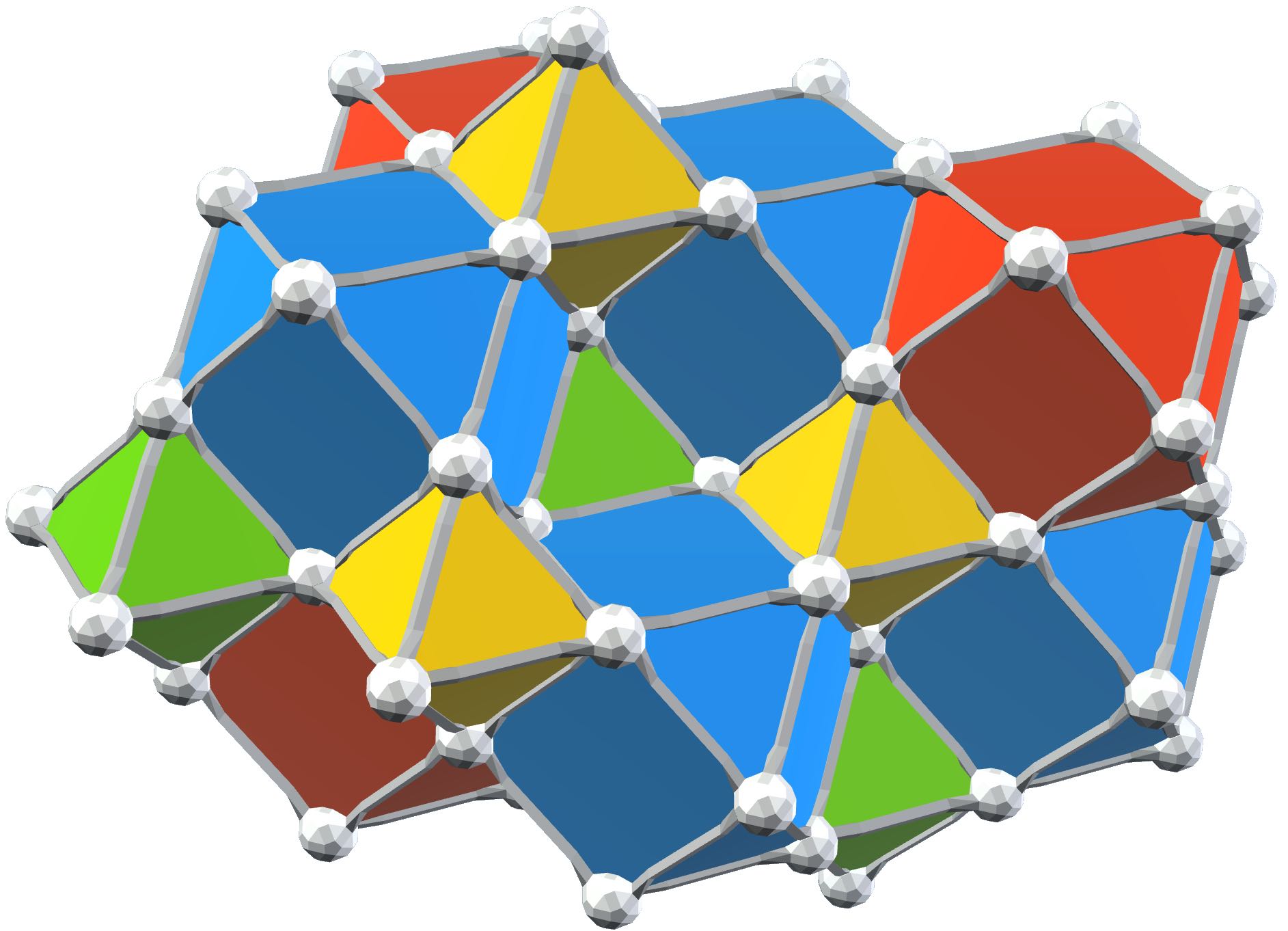}
\hspace*{-.33\textwidth}(b)\hspace*{.31\textwidth}
\caption{
(a) A colorable octahedral lattice $\mathcal{L}$.
We obtain $\mathcal L$ from the body-centered cubic lattice by first assigning colors $G$ or $Y$ to the vertices at the centers of the cubic volumes, and colors $R$ or $B$ to all other vertices, followed by filling in octahedra.
We only depict two octahedra (shaded in gray).
(b) The dual lattice $\mathcal{L}^*$ is obtained from $\mathcal L$ by replacing its volumes, faces, edges and vertices by vertices, edges, faces and volumes, respectively.
The lattice $\cub$ from Sec.~\ref{sec_glance} is a simplified version of $\mathcal{L}^*$, where the $G$ and $Y$ volumes are omitted.
Note that the 3D~STC can be defined on $\mathcal{L}^*$ by placing one qubit on every vertex, and introducing $X$- and $Z$-type gauge generators for every $RG$ or $RY$ face and $BG$ or $BY$ face, respectively.}
\label{fig_lattice}
\end{figure*}

To reduce weight of gauge operators at the top and bottom boundaries of the lattice $\cub$, we first place some unentangled ancilla qubits $\mathcal A_S \cup \mathcal A_G$ on those boundaries; see Fig.~\ref{fig_2DSTC}(a).
Every ancilla qubit $i\in\mathcal A_S$ is prepared in the state $\ket{0}$ and thus is stabilized by a single-qubit Pauli $Z$ operator; every other ancilla qubit $j\in\mathcal A_G$ is treated as a gauge qubit.
Hence, we arrive at an augmented subsystem code with the following gauge group
\begin{equation}
\mathcal G_\text{aug} = \langle G, Z_i, X_j, Z_j \mathrel{|} G\in \mathcal G, i \in \mathcal A_S, j \in \mathcal A_G \rangle.
\end{equation}
Then, we implement a unitary $U$ via a constant-depth circuit composed of four rounds of geometrically-local CNOT gates.
One can straightforwardly verify that the generators of the gauge group $U\mathcal G_\text{aug} U^\dag$ can be chosen in such a way that each of them has weight at most three.
In Fig.~\ref{fig_2DSTC}(b) we depict the choice of generators fully supported within the top and bottom boundaries.
We remark that the subsystem code $U\mathcal G_\text{aug} U^\dag$ and the 3D~STC on $\cub$ from the previous subsection are equivalent in the sense of a local unitary transformation and adding or removing ancilla qubits~\cite{Hastings2005,Bravyi2006,Chen2010,Bravyi2013sub}.

\subsection{Octahedral lattices and colorability}
\label{sec_octahedral}

Before we provide a general construction, which allows us to define the 3D~STC on lattices other than the cubic lattice, we need to define bipyramidal lattices and a notion of colorability.
A bipyramidal lattice $\mathcal{L}$ is obtained by gluing bipyramidal volumes\footnote{
An $n$-gonal bipyramid is a polyhedron formed by glueing together two $n$-gonal pyramids along their bases.}
along their proper faces of matching dimensions; this procedure is analogous to a construction of a homogeneous simplicial complex~\cite{Hatcher2002}.
If a bipyramidal lattice contains only octahedral volumes, we call it an octahedral lattice. 
We restrict our attention to bipyramidal lattices containing finitely many constituents.
Although the lattice $\mathcal{L}$ is a collection of vertices, edges, faces and volumes, it can be also viewed as a topological space.
In particular, $\mathcal{L}$ is a manifold, possibly with boundary.

We say that an octahedral volume is antipodally colored if the three pairs of its opposite vertices are two $R$ vertices, two $B$ vertices, and $G$ and $Y$ vertices. 
Then, we say that an octahedral lattice is colorable if we can assign four colors $R$, $B$, $G$ and $Y$ to its vertices in such a way that every octahedral volume is antipodally colored.
Note that other cells of the octahedral lattice inherit colors in a natural way.
For example, an edge between $R$ and $G$ vertices has color $RG$.
In Fig.~\ref{fig_lattice}, we illustrate an example of a colorable octahedral lattice $\mathcal{L}$, as well as its dual lattice $\mathcal L^*$.
We also remark that the lattice dual to the lattice $\cub$ from Sec.~\ref{sec_glance} forms a colorable octahedral lattice.

Let $\mathcal{L}$ be a three-dimensional lattice, which satisfies the following two conditions
\begin{itemize}
\item[(i)] $\mathcal{L}$ is an octahedral lattice,
\item[(ii)] $\mathcal{L}$ is colorable.
\end{itemize}
We use a notation $\mathcal{L}_i$ to denote the set of all $i$-dimensional cells of $\mathcal{L}$, and write $\mathcal{L}_i^C$ to further restrict our attention to cells of color $C$.
The following lemma proves useful in calculating the number of constituents of the lattice $\mathcal{L}$.
\begin{lemma}
\label{lemma_euler}
Let $\mathcal{L}$ be a colorable octahedral lattice without boundary.
Then, the following conditions on the number of constituents of $\mathcal{L}$ hold
\begin{enumerate}
\item[(i)] $|\mathcal{L}_2| = 4|\mathcal{L}_3|$,
\item[(ii)] $|\mathcal{L}_1| =
|\mathcal{L}^{RG}_1| +  |\mathcal{L}^{RY}_1| + |\mathcal{L}^{BG}_1| + |\mathcal{L}^{BY}_1| + |\mathcal{L}^{RB}_1|$,
\item[(iii)] $2|\mathcal{L}^{G}_0| + 2|\mathcal{L}^{Y}_0| - |\mathcal{L}_1| + |\mathcal{L}^{RB}_1| + 2|\mathcal{L}_3| = 0$.
\end{enumerate}
\end{lemma}

\begin{proof}
To show (i), we observe that each face of $\mathcal{L}$ is shared between two octahedral volumes and every octahedral volume has eight faces.
To show (ii), we observe that there are no $GY$ edges in $\mathcal{L}$, as every octahedral volume in $\mathcal{L}$ is antipodally colored.
To show (iii), we construct a new tessellation of the manifold corresponding to $\mathcal{L}$.
We first split every octahedral volume $\omega\in\mathcal{L}_3$ into two pyramids by inserting into $\omega$ a face $f(\omega)$ glued along the cycle consisting of four $RB$ edges.
Then, for every $G$ or $Y$ vertex $v\in\mathcal{L}^{G}_0\cup\mathcal{L}^{Y}_0$ we find all the pyramids containing $v$ and merge them into a single three-dimensional cell $c(v)$, whose boundary corresponds to the collection of the bases of these pyramids.
Note that this step removes all the edges and vertices of color different than $RB$ and $R$ or $B$, respectively.
The resulting tessellation $\hat{\mathcal{L}}$ consists of the vertices $\mathcal{L}^{R}_0\cup\mathcal{L}^{B}_0$, edges $\mathcal{L}^{RB}_1$, faces $\{ f(\omega) | \omega\in\mathcal{L}_3\}$, and volumes $\{ c(v) | v\in\mathcal{L}^{G}_0\cup\mathcal{L}^{Y}_0\}$.
Since for $\mathcal L$ and $\hat{\mathcal{L}}$ the Euler characteristic is the same (and equal to zero), we obtain
\begin{equation}
|\mathcal{L}_0| - |\mathcal{L}_1| + |\mathcal{L}_2| - |\mathcal{L}_3| =
|\mathcal{L}^{R}_0| + |\mathcal{L}^{B}_0| - |\mathcal{L}^{RB}_1| + |\mathcal{L}_3| - |\mathcal{L}^{G}_0| -| \mathcal{L}^{Y}_0|,
\end{equation}
and by using (i) and rearranging the terms we finally obtain (iii).
\end{proof}

Lastly, we explain how one can transform any three-dimensional lattice $\mathcal{L}$ without boundary into a colorable octahedral lattice $\mathcal{L}^{\mathrm{oct}}$.
In the first step, we convert $\mathcal{L}$ into a simplicial $d$-complex $\mathcal{L}^{\mathrm{sim}}$, where $d$ denotes the dimensionality of $\mathcal{L}$.
This step is general and does not require that $d=3$.
For each flag of $\mathcal L$, we include a corresponding $d$-simplex in $\mathcal L^{\mathrm{sim}}$.
We recall that a (geometrical) flag is a sequence of cells of $\mathcal L$, where each cell is contained in the next and there is exactly one $i$-cell of each dimension $i \in \{0,\ldots,d\}$.
We note that the faces of a given $d$-simplex correspond to the non-empty subsets of its corresponding flag.
In particular, the vertices of a given $d$-simplex correspond to the individual $i$-cells contained in its corresponding flag. 
In the second step, since the vertices of $\mathcal{L}^{\mathrm{sim}}$ correspond to the vertices, edges, faces and volumes of $\mathcal{L}$, we can color them in $R$, $G$, $Y$ and $B$, respectively\footnote{
We remark that for our transformation to work the vertices could also be colored in, respectively, 
 $G$, $R$, $Y$ and $B$, or $R$, $G$, $B$ and $Y$.}.
Then, for every $GY$ edge in $\mathcal{L}^{\mathrm{sim}}_1$ we find four tetrahedra in $\mathcal{L}^{\mathrm{sim}}_3$ containing it, and merge them into a single octahedral volume.
Note that this step removes all the edges and faces of color $GY$ and $RGY$ or $BGY$, respectively.
It is easy to see that the resulting lattice $\mathcal{L}^{\mathrm{oct}}$ forms a colorable octahedral lattice.
In Fig.~\ref{fig_octahedral} we illustrate the process of converting $\mathcal{L}$ into $\mathcal{L}^{\mathrm{oct}}$.
We remark that the first step is equivalent to the inflation procedure~\cite{Bombin2007} presented from the perspective of the dual lattice $\mathcal{L}^*$.

\subsection{3D subsystem toric code with no logical qubits}
\label{sec_dual_lattice}

\begin{figure*}[ht!]
\centering
\includegraphics[width=.25\textwidth]{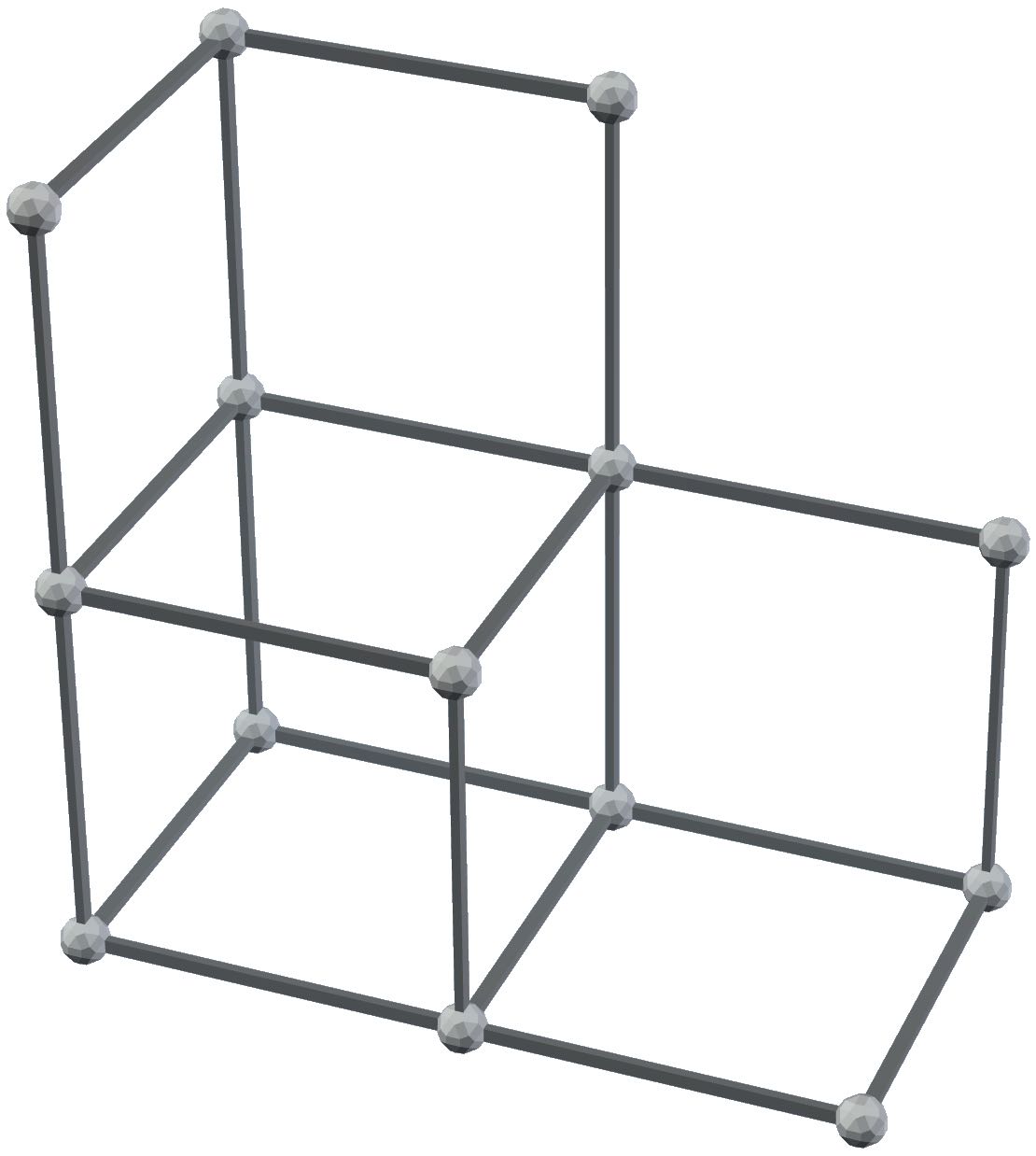}
\hspace*{-.25\textwidth}(a)\hspace*{.30\textwidth}
\includegraphics[width=.25\textwidth]{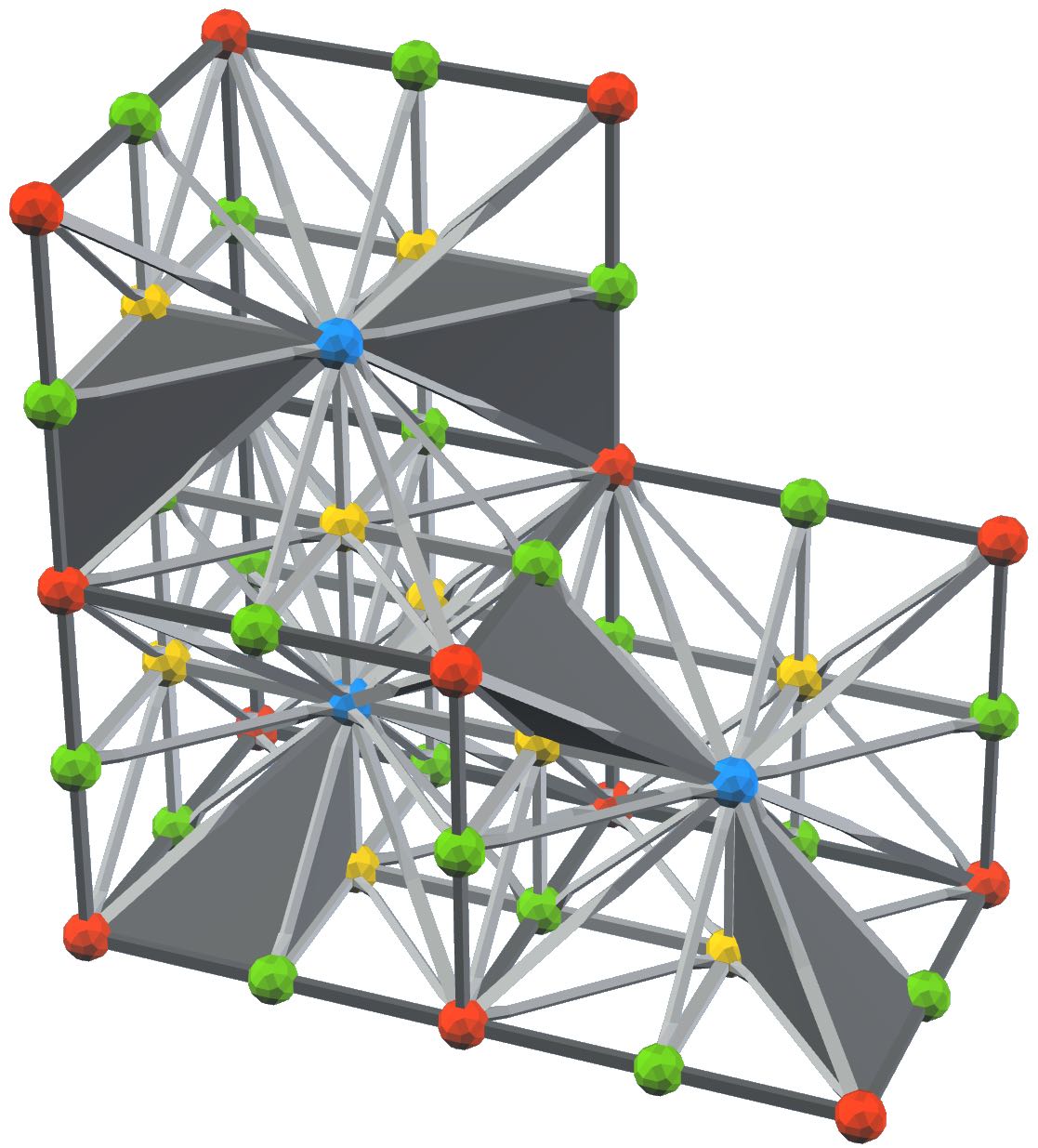}
\hspace*{-.25\textwidth}(b)\hspace*{.30\textwidth}
\includegraphics[width=.25\textwidth]{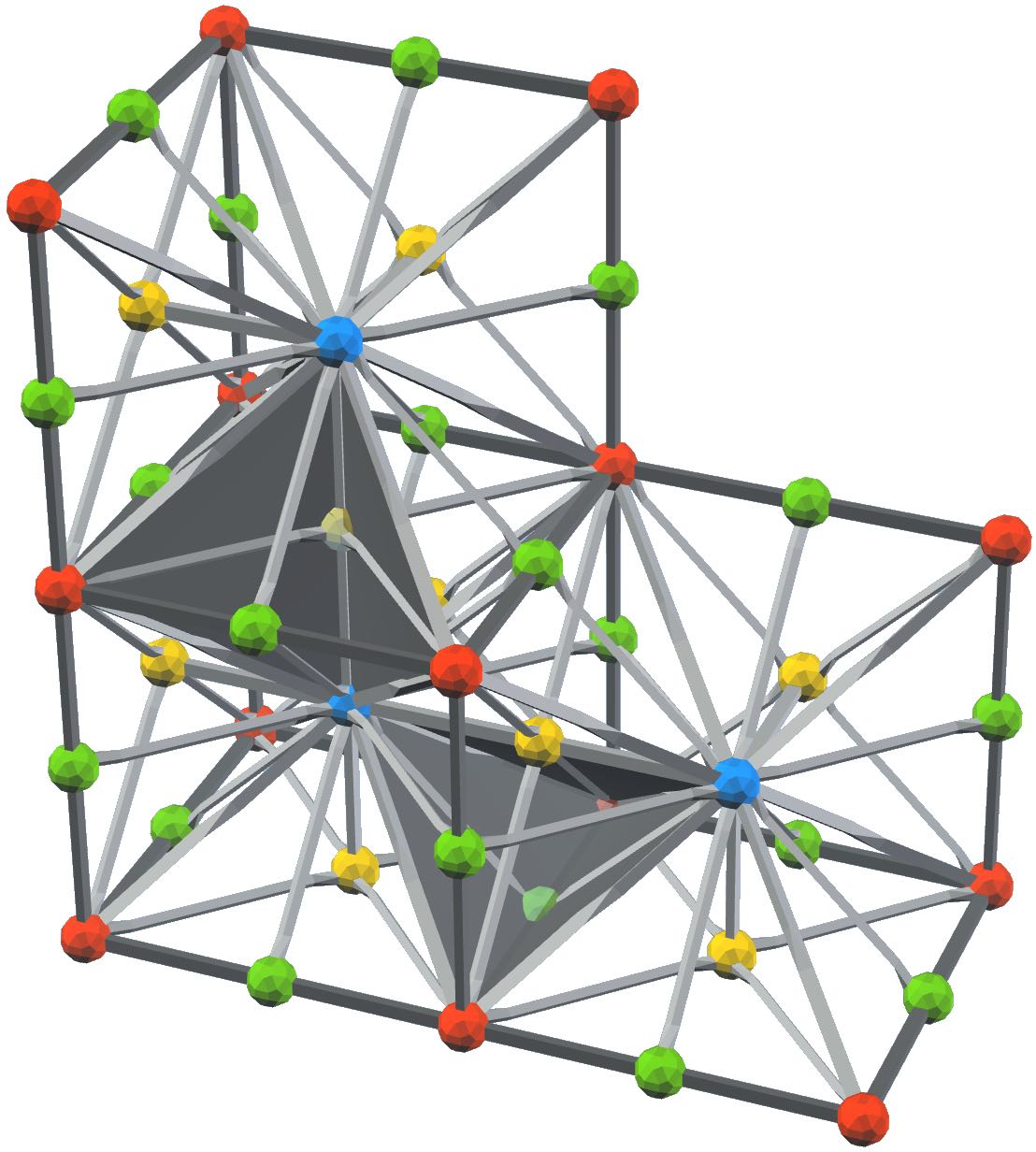}
\hspace*{-.25\textwidth}(c)\hspace*{.22\textwidth}
\caption{
A procedure of constructing a colorable octahedral lattice $\mathcal{L}^{\mathrm{oct}}$.
(a) We start with an arbitrary three-dimensional lattice $\mathcal{L}$ without boundary.
(b) We convert $\mathcal{L}$ into a simplicial complex $\mathcal{L}^{\mathrm{sim}}$.
The vertices of $\mathcal{L}^{\mathrm{sim}}$ can be colored in $R$, $G$, $Y$ and $B$, since they correspond to the vertices, edges, faces and volumes of $\mathcal{L}$.
We shade some tetrahedra in grey.
(c) We obtain $\mathcal{L}^{\mathrm{oct}}$ by finding for every $GY$ edge in $\mathcal{L}^{\mathrm{sim}}_3$ four tetrahedra containing it and merging them into a single cell.
We shade some octahedral volumes in grey.
}
\label{fig_octahedral}
\end{figure*}

Now we are ready to present a general construction of the 3D~STC, which, contrasted with the initial construction from Sec.~\ref{sec_glance}, allows us to define the 3D~STC on lattices other than the cubic lattice.
Moreover, this construction is not only useful in calculating the number of logical qubits of the 3D~STC, but also leads to a succinct description of the boundaries of the 3D~STC lattice.

Let $\mathcal{L}$ be a colorable octahedral lattice.
For simplicity, we assume that the lattice $\mathcal{L}$ is obtained by tessellating a 3-sphere.
We identify each octahedral volume of $\mathcal{L}$ with a qubit and thus the number of physical qubits is 
\begin{equation}
N = |\mathcal{L}_3|.
\end{equation}
For any $i$-dimensional cell $\delta\in\mathcal{L}_i$ we denote by $\mathcal{Q}(\delta)$ the set of all the qubits on octahedra containing $\delta$, i.e.,
\begin{equation}
\mathcal{Q}(\delta) = \{ \omega \in\mathcal{L}_3 \mathrel{|} \omega \supseteq \delta \}.
\end{equation}
By saying that an operator is supported on $\delta$ we mean that it is supported on the set of qubits $\mathcal{Q}(\delta)$ and, for instance, write $X(\delta) = \prod_{\omega\in\mathcal{Q}(\delta)} X_\omega$, where $X_\omega$ denotes Pauli $X$ operator acting on qubit $\omega$.

The gauge group $\mathcal{G}$ of the 3D~STC is generated by $X$- and $Z$-type gauge generators supported on, respectively, $RG$ or $RY$ edges, and $BG$ or $BY$ edges, namely
\begin{equation}
\mathcal{G} = \big\langle X(\mu), Z(\nu) \mathrel{\big |} \mu \in \mathcal{L}_1^{RG}\cup\mathcal{L}_1^{RY},
\nu \in \mathcal{L}_1^{BG}\cup\mathcal{L}_1^{BY} \big\rangle.
\end{equation}
The stabilizer group of the 3D~STC is generated by $X$- and $Z$-type stabilizer generators supported on $R$ and $B$ vertices, i.e.,
\begin{equation}
\label{eq_stabilizers}
\mathcal{S} = \big\langle X(u), Z(v) \mathrel{\big |}
u \in \mathcal{L}_0^{R},v \in \mathcal{L}_0^{B} \big\rangle.
\end{equation}
Indeed, for each vertex in $\mathcal{L}_0^{R}\cup\mathcal{L}_0^{B}$ a corresponding stabilizer generator can be formed in two ways by multiplying gauge generators supported on edges incident to that vertex, namely
\begin{eqnarray}
X(u) &=& \prod_{\substack{\mu\in\mathcal{L}^{RG}_1 : \mu\supset u}} X(\mu)
= \prod_{\substack{\mu\in\mathcal{L}^{RY}_1 : \mu\supset u}} X(\mu),\\
\label{eq_Zstab}
Z(v) &=& \prod_{\substack{\nu\in\mathcal{L}^{BG}_1 : \nu\supset v}} Z(\nu)
= \prod_{\substack{\nu\in\mathcal{L}^{BY}_1 : \nu\supset v}} Z(\nu).
\end{eqnarray}
To see that stabilizer operators commute with each other and with gauge operators, it suffices to show that a stabilizer generator $X(u)$ and a gauge generator $Z(\nu)$ commute; the argument for a stabilizer generator $Z(v)$ and a gauge generator $X(\mu)$ is the same.
If the intersection $\mathcal{Q}(u)\cap \mathcal{Q}(\nu)$ is non-empty, then there exists an octahedral volume $\omega\in\mathcal{L}_3$ containing both the vertex $u$ and the edge $\nu$, i.e., $\omega\supset u,\nu$.
Since the lattice $\mathcal{L}$ is colorable, the octahedral volume $\omega$ is antipodally colored, and 
by definition of the stabilizer and gauge groups, $u$ does not belong to $\nu$, i.e., $u \cap \nu = \emptyset$.
Thus, there is a triangular face $f$ of $\omega$ spanned by $u$ and $\nu$, i.e., $\omega\supset f\supset u,\nu$, and we have $\mathcal{Q}(u)\cap\mathcal{Q}(\nu) = \mathcal{Q}(f)$.
Since the set $\mathcal{Q}(f)$ contains two elements, we conclude that $X(u)$ and $Z(\nu)$ commute.

We can verify that the 3D~STC defined on the lattice $\mathcal{L}$ has no logical qubits.
Since stabilizer generators correspond to $R$ and $B$ vertices of $\mathcal{L}$ and there are two relations between them, namely
\begin{equation}
\label{eq_relations_stab}
\prod_{u\in\mathcal{L}_0^R} X(u) = \prod_{v\in\mathcal{L}_0^R} Z(u) = I,
\end{equation}
thus the number of independent generators of the stabilizer group $\mathcal{S}$ is
\begin{equation}
\log_2 |\mathcal{S}| = |\mathcal{L}^{R}_0|+|\mathcal{L}^{B}_0| - 2.
\end{equation}
We also find the following four types of relations for gauge generators
\begin{eqnarray}
\label{eq_relations_gauge1}
\forall u\in\mathcal{L}^R_0&:& 
\prod_{\substack{\mu\in\mathcal{L}^{RG}_1\cup \mathcal{L}^{RY}_1 : \mu\supset u}} X(\mu) = I,\\
\label{eq_relations_gauge2}
\forall v\in\mathcal{L}^B_0&:& 
\prod_{\substack{\mu\in\mathcal{L}^{BG}_1\cup \mathcal{L}^{BY}_1 : \mu\supset v}} Z(\mu) = I,\\
\label{eq_relations_gauge3}
\forall w\in\mathcal{L}^G_0&:&
\prod_{\substack{\mu\in\mathcal{L}^{RG}_1 : \mu\supset w}} X(\mu)
= \prod_{\substack{\mu\in\mathcal{L}^{BG}_1 : \mu\supset w}} Z(\mu) = I,\\
\label{eq_relations_gauge4}
\forall w\in\mathcal{L}^Y_0&:&
\prod_{\substack{\mu\in\mathcal{L}^{RY}_1 : \mu\supset w}} X(\mu)
= \prod_{\substack{\mu\in\mathcal{L}^{BY}_1 : \mu\supset w}} Z(\mu) = I.
\end{eqnarray}
However, the above relations are not all independent; rather, we overcount them by 2.
Since gauge generators correspond to $RG$, $RY$, $BG$ and $BY$ edges of $\mathcal{L}$,
thus the number of independent generators of the gauge group $\mathcal{G}$ is
\begin{eqnarray}
\log_2 |\mathcal{G}| 
&=& |\mathcal{L}^{RG}_1|+|\mathcal{L}^{RY}_1|+|\mathcal{L}^{BG}_1|+|\mathcal{L}^{BY}_1|
- |\mathcal{L}^{R}_0| - |\mathcal{L}^{B}_0| - 2 |\mathcal{L}^{G}_0| - 2 |\mathcal{L}^{Y}_0| + 2\\
&=& 2|\mathcal{L}_3| - |\mathcal{L}^{R}_0| - |\mathcal{L}^{B}_0| + 2,
\end{eqnarray}
where we use Lemma~\ref{lemma_euler}.
We finally obtain that the number of logical qubits encoded into the 3D~STC on the lattice $\mathcal{L}$ is 
\begin{eqnarray}
\label{eq_no_logical}
K &=& N - \frac{1}{2}(\log_2 {|\mathcal{G}|} + \log_2 {|\mathcal{S}|}) = 0.
\end{eqnarray}

We remark that we could also consider the 3D~STC on a colorable octahedral lattice, which is a tessellation of any orientable closed 3-manifold, not necessarily a 3-sphere.
In such a case, however, the stabilizer group would not only be generated by geometrically-local generators in Eq.~\eqref{eq_stabilizers}, but also by non-local ones corresponding to non-trivial elements of the second homology group of the manifold.
At the same time, for each non-trivial element of the second homology group we would find an independent relation for gauge generators.
Thus, we would conclude that the 3D~STC has zero logical qubits.
This is why in Sec.~\ref{sec_glance} we consider the 3D~STC defined on the lattice $\cub$ with open boundary conditions instead of a simpler translationally-invariant lattice, such as the cubic lattice with periodic boundary conditions.

\subsection{3D~STC with one logical qubit}
\label{sec_cubiclattice}

We have just seen that the 3D~STC defined on a colorable octahedral lattice $\mathcal{L}$, which is a tessellation of a 3-sphere, has no logical qubits.
To obtain the 3D~STC with one logical qubit, we first construct a new lattice $\mathcal{L}'$ by removing one octahedral volume from $\mathcal{L}$.
This procedure is analogous to the construction of the color code with one logical qubit~\cite{Kubica2015a}.
Note that the lattice $\mathcal{L}'$ is a tessellation of the three-dimensional ball, and its boundary $\partial\mathcal{L}'$ corresponds to the boundary of the removed octahedral volume.
The gauge group $\mathcal{G}'$ is generated by $X$- and $Z$-type operators supported on $RG$ or $RY$ and $BG$ or $BY$ edges in the interior of $\mathcal{L}'$, namely
\begin{eqnarray}
\mathcal{G}' = \big\langle X(\mu), Z(\nu) &\mathrel{\big |}&
\mu \in {\mathcal{L}'}_1^{RG}\cup{\mathcal{L}'}_1^{RY}\setminus \partial{\mathcal{L}'}_1, \nu \in {\mathcal{L}'}_1^{BG}\cup{\mathcal{L}'}_1^{BY}\setminus \partial{\mathcal{L}'}_1 \big\rangle.
\end{eqnarray}
One can verify that the stabilizer group is generated by $X$- and $Z$-type operators supported on $R$ and $B$ vertices in the interior of $\mathcal{L}'$, namely 
\begin{equation}
\mathcal{S}' = \big\langle X(u), Z(v) \mathrel{\big |}
u \in {\mathcal{L}'}_0^{R}\setminus \partial{\mathcal{L}'}_0,
v \in {\mathcal{L}'}_0^{B}\setminus \partial{\mathcal{L}'}_0 \big\rangle.
\end{equation}

Compared to the 3D~STC defined on $\mathcal{L}$, the number of physical qubits is reduced by one.
Also, we discard four stabilizer generators supported on vertices of the boundary $\partial\mathcal{L}$.
Since the stabilizer generators no longer satisfy the relations in Eq.~\eqref{eq_relations_stab}, we thus obtain that the number of independent stabilizer generators is reduced by two.
Similarly, we discard eight gauge generators supported on edges belonging to $\partial\mathcal{L}$ and the remaining gauge generators no longer satisfy eight relations out of the ones in Eqs.~\eqref{eq_relations_gauge1}-\eqref{eq_relations_gauge4}.
The remaining relations are, however, independent, and thus the number of independent gauge generators is reduced by two.
Combining the above and using Eq.~\eqref{eq_no_logical}, we obtain that the 3D~STC defined on the lattice $\mathcal{L}'$ with boundary has one logical qubit.
We then immediately conclude that the 3D~STC defined on the lattice $\cub$ from Sec.~\ref{sec_glance} also has one logical qubit, as it is an example of this construction (except in the dual lattice).
Lastly, we remark that by removing more octahedral volumes we could encode more logical qubits.

\subsection{Relation to the toric code}

To understand how the 3D~STC is related to the 3D stabilizer toric code, it is useful to recast the description of the latter.
A standard way of defining the 3D toric code on a three-dimensional lattice $\mathcal{K}$ is to place a qubit on every edge of $\mathcal{K}$ and to introduce $X$- and $Z$-type stabilizer generators for every vertex $v$ and face $f$ of $\mathcal{K}$ as the product of Pauli $X$ and $Z$ operators on qubits on edges adjacent to $v$ and $f$; see Fig.~\ref{fig_toric}(a).
There is, however, an equivalent way to describe the 3D toric code in the (dual of the) rectified lattice picture~\cite{Vasmer2019}.
This approach is similar in spirit to the definition of the color code.
Roughly speaking, to obtain a rectified lattice $\mathcal{K}^{\mathrm{rec}}$ from the lattice $\mathcal{K}$ we inflate every vertex of $\mathcal{K}$ by introducing an extra volume there; see Fig.~\ref{fig_toric}(b).
Let $\widetilde{\mathcal{K}}$ be a lattice dual to the rectified lattice $\mathcal{K}^{\mathrm{rec}}$; see Fig.~\ref{fig_toric}(c).
By construction, the lattice $\widetilde{\mathcal{K}}$ is a bipyramidal lattice.
We color the vertices of $\widetilde{\mathcal{K}}$ corresponding to the inflated vertices of $\mathcal{K}$ in red and all other vertices in silver.
Then, every bipyramidal volume of $\widetilde{\mathcal{K}}$ has a pair of antipodal vertices colored in red.
Subsequently, we can define the 3D toric code by placing a qubit at every bipyramidal volume, and introducing $X$- and $Z$-type stabilizer generators for every red vertex and every edge connecting two silver vertices.

\begin{figure*}[ht!]
\centering
\includegraphics[height=.23\textheight]{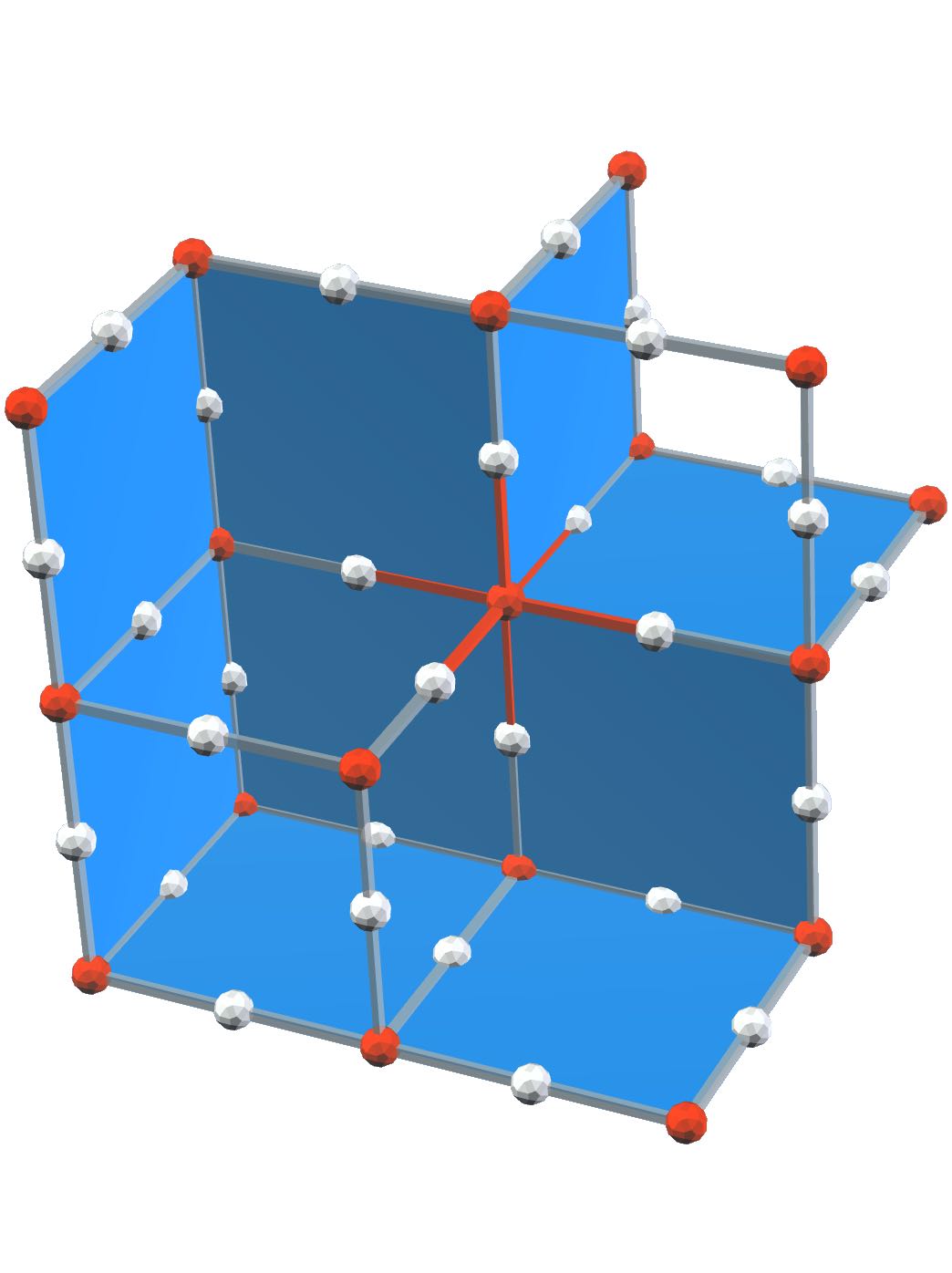}
\hspace*{-.25\textwidth}(a)\hspace*{.27\textwidth}
\includegraphics[height=.17\textheight]{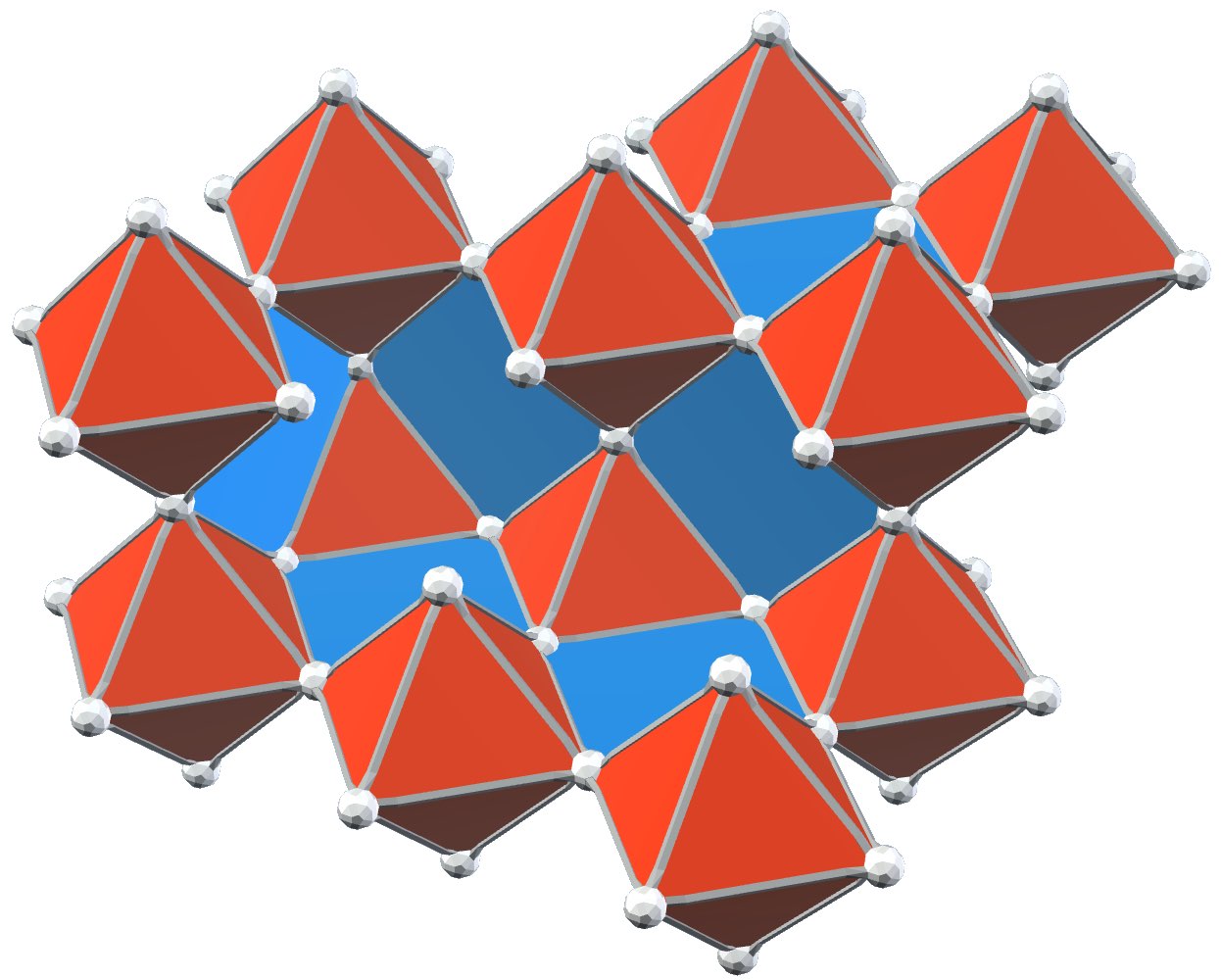}
\hspace*{-.31\textwidth}(b)\hspace*{.33\textwidth}
\includegraphics[height=.18\textheight]{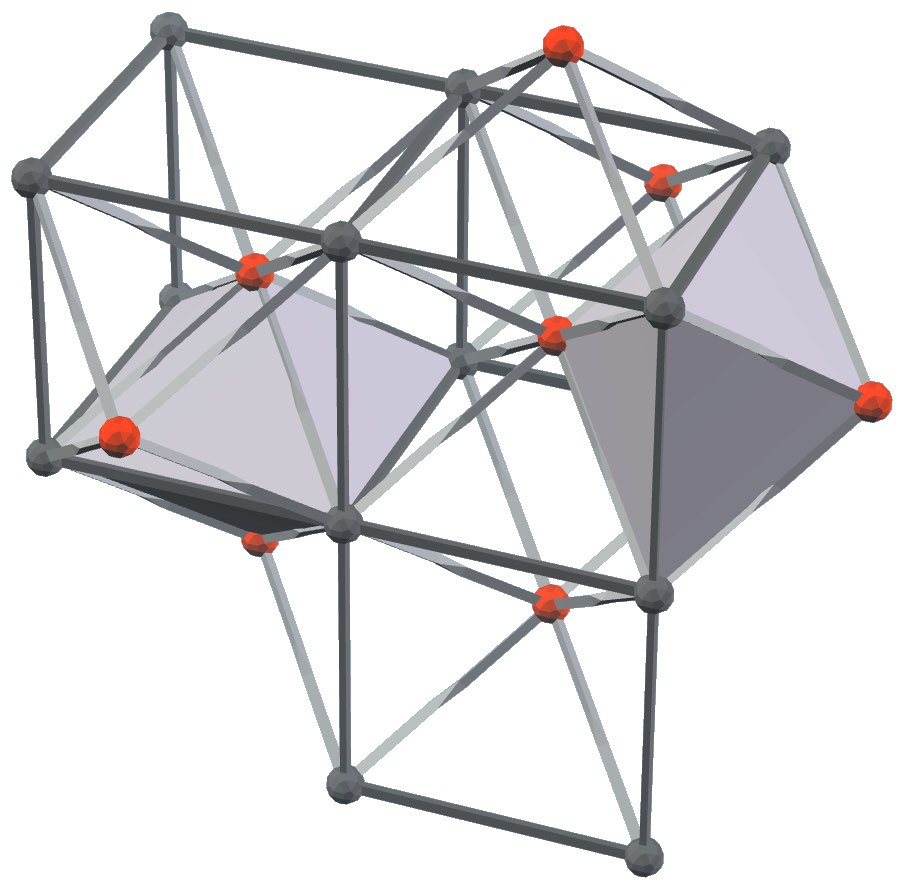}
\hspace*{-.27\textwidth}(c)\hspace*{.25\textwidth}
\caption{
Three equivalent definitions of the 3D stabilizer toric code.
(a) The standard construction on a lattice $\mathcal{K}$.
Qubits (white balls) are placed on edges, and $X$- and $Z$-stabilizers are identified with vertices (red) and faces (blue).
(b) A rectified lattice $\mathcal{K}^{\mathrm{rec}}$ obtained from $\mathcal{K}$.
Qubits are placed at vertices (white balls), and $X$- and $Z$-stabilizers are identified with volumes (red) and faces (blue).
(c) A lattice $\widetilde{\mathcal{K}}$ dual to $\mathcal{K}^{\mathrm{rec}}$ forms a bipyramidal lattice, whose vertices are colored in red and silver. 
Qubits are identified with bipyramidal volumes (we shade two of them), and $X$- and $Z$-stabilizers are identified with red vertices and edges linking silver vertices.
}
\label{fig_toric}
\end{figure*}

Equipped with the alternative definition of the 3D~toric code, we establish a relation between the 3D~STC and the 3D~toric code.
For simplicity, let $\mathcal L$ be a colorable octahedral lattice, which is a tessellation of a 3-sphere, and let $\mathcal G$ and $\mathcal S$ denote the gauge and stabilizer groups of the 3D~STC defined on $\mathcal L$.
Note that if we identify all the colors except red with silver, then $\mathcal L$ is a bipyramidal lattice, whose bipyramidal volumes each have a pair of antipodal vertices colored in red.
Therefore, the stabilizer group of the 3D~toric code on $\mathcal L$ is
\begin{equation}
\mathcal{S}_{\mathrm{3DST}} = \big\langle X(v), Z(e) \mathrel{\big |}
v \in \mathcal{L}_0^{R}, e \in\mathcal{L}_1^{BG}\cup \mathcal{L}_1^{BY} \big\rangle.
\end{equation}
One can then easily check that $\mathcal S$ is a subgroup of $\mathcal{S}_{\mathrm{3DST}}$, which, in turn, is a subgroup of $\mathcal G$, i.e.,
\begin{equation}
\mathcal S\leq \mathcal{S}_{\mathrm{3DST}} \leq \mathcal G.
\end{equation}
Thus, a state in the code space of the 3D toric code is also in the code space of the 3D~STC, with its gauge qubits in a certain determined state.
Furthermore, using the procedure of gauge fixing we can ensure that all $Z$-type gauge generators in $\mathcal G$ are satisfied, thereby fixing the gauge qubits of a state in the code space of the 3D~STC.
Subsequently, we can map that state to the code space of the 3D stabilizer toric code.

\section{Single-shot decoding}
\label{sec_ss_decoding}

We start this section by describing the MWPM problem and introduce a simple formalism capturing the decoding problem for CSS subsystem codes in the presence of measurement errors.
We focus our discussion on Pauli $X$ errors, as Pauli $Z$ errors can be handled analogously.
We then introduce the notion of the gauge flux in the 3D~STC.
Using this physical intuition, we propose a decoding strategy for the 3D~STC, the single-shot MWPM decoder, which consist of two steps: (i) syndrome estimation and (ii) ideal MWPM decoding.
Lastly, we numerically estimate the performance of the single-shot MWPM decoder.

\subsection{The Minimum-Weight Perfect Matching problem}

Let $G = (V,E)$ be a graph with the sets of vertices $V$ and edges $E$, and $V'\subseteq V$ be some subset of vertices.
Let $C_V$, $C_{V'}$ and $C_E$ be the $\mathbb{F}_2$-vector spaces with the sets $V$, $V'$ and $E$ as their bases, respectively.
In what follows we identify the elements of $\mathbb{F}_2$-vector spaces with the subsets of the corresponding basis sets.
We define a linear map $\partial_G: C_E \rightarrow C_V$, which we call the boundary operator, by specifying it for every basis element $e\in E$ as follows
\begin{equation}
\partial_G e = u + v,
\end{equation}
where $u$ and $v$ are the endpoints of $e$.
Subsequently, we define the relative boundary operator $\partial'_G: C_E \rightarrow C_V/C_{V'}$ with respect to $V'$, where $C_V/C_{V'}$ is the quotient space, as a composition of the boundary operator $\partial_G$ with the natural homomorphism $\pi: C_V \rightarrow C_V/C_{V'}$, i.e.,
$\partial'_G = \pi \circ \partial_G$.
Then, for the given graph $G = (V, E)$ and the subset of vertices $V'\subseteq V$, we can consider the following problem, which we call the Minimum-Weight Perfect Matching problem:
\textit{for any $c_V\in C_V/C_{V'}$ find the minimum-weight $c_E\in C_E$, such that $\partial'_G c_E = c_V$, i.e.,}
\begin{equation}
c_E = \argmin_{c_E'\in C_E: \partial'_G c_E' = c_V} |c_E'|.
\end{equation}
We remark that the MWPM problem might have no solution.
However, if a solution exists, then it can be efficiently found in a time polynomial in the number of constituents of the graph by, e.g., the blossom algorithm~\cite{Kolmogorov2009}.

\subsection{Decoding of CSS subsystem codes with measurement errors}

In order to diagnose Pauli $X$ errors affecting a CSS subsystem code, we need to measure some of the $Z$-type gauge operators.
Note that if we choose to measure an overcomplete set of $Z$-type gauge operators, then there will be some consistency checks on the measurement outcomes.
The obtained measurement outcomes are then classically processed to find an appropriate $X$-type recovery operator.
In the absence of measurement errors, we can find the stabilizer syndrome of the Pauli $X$ errors exactly and subsequently guarantee that the recovery operator returns the state to the code space.
We then say that decoding has succeeded if the error and the recovery operator are equivalent up to some $X$-type gauge operator.
However, in the presence of measurement errors, the recovery operator will most likely not return the state to the code space and there will be some residual Pauli $X$ errors left in the system.

To capture the decoding problem for any CSS subsystem code we introduce a commutative diagram consisting of $\mathbb{F}_2$-vector spaces (whose elements are column vectors) and linear maps between them (treated as binary matrices) as follows
\begin{equation}
\label{eq_diagram}
\begin{tikzcd}[ampersand replacement=\&, column sep = 0, row sep = 0]
\&[6em] \& \text{measurements} \& \\
\&\& C_M \ar[dr, "{\left(\begin{array}{c} \delta_S\\\hline \delta_R\end{array}\right)}"] \&\\[10ex] 
C_G \ar[r, "{\left(\begin{array}{c} \partial_Q\\\hline 0\end{array}\right)}"] \&
C_Q \oplus C_M \ar[ur, "{\left(\begin{array}{c|c} \delta_M &I\end{array}\right)}"]
\ar[rr, "{\left(\begin{array}{c|c} \partial_S& \delta_S \\\hline 0& \delta_R \end{array}\right)}"]\&\& C_S \oplus C_R\\
\begin{array}{c} \text{gauge} \\ \text{generators} \end{array} \&
\begin{array}{c} \text{qubits and} \\ \text{measurements} \end{array} \& \&
\begin{array}{c} \text{stabilizers} \\ \text{and relations} \end{array}
\end{tikzcd}
\end{equation}
The bases of $C_G$, $C_Q$, $C_M$, $C_S$ and $C_R$ correspond to the sets of, respectively, $X$-type gauge generators, qubits, measured $Z$-type gauge operators, independent $Z$-type stabilizer generators and independent relations between the measured $Z$-type gauge operators.
The relations amount to consistency checks on the measurement outcomes of the $Z$-type gauge operators.
We treat the elements of the vector spaces and the subsets of the corresponding basis sets interchangeably. 
We refer to the elements of the vector spaces $C_G$, $C_Q\oplus C_M$, $C_M$ and $C_S\oplus C_R$ as gauge operators, errors, measurement outcomes\footnote{
The measurement outcome is defined as the set of all measured $Z$-type gauge operators returning a $-1$ outcome.}
and syndromes, respectively.
Every error $\epsilon\oplus\mu \in C_Q\oplus C_M $ consists of the Pauli $X$ error $\epsilon$ and the measurement error $\mu$.
Similarly, every syndrome $\sigma\oplus\omega \in C_S \oplus C_R$ consists of the stabilizer syndrome $\sigma$, which is the set of violated stabilizers, and the relation syndrome $\omega$, which is the set of violated relations.
We choose the linear maps in Eq.~\eqref{eq_diagram} in such a way that:
(i) the support of any $X$-type gauge operator $\gamma \in C_G$ is $\partial_Q \gamma \in C_Q$,
(ii) the set of $Z$-type gauge operators anticommuting with any Pauli $X$ error $\epsilon \in C_Q$ is $\delta_M \epsilon\in C_M$,
(iii) the stabilizer syndrome and the relation syndrome corresponding to the measurement outcome $\zeta \in C_M$ are $\delta_S \zeta \in C_S$ and $\delta_R \zeta \in C_R$.
We define $\partial_S = \delta_S \delta_M$ and require that 
\begin{equation}
\label{eq_bnd1}
\partial_S \partial_Q = 0,
\end{equation}
as the stabilizer syndrome of any $X$-type gauge operator has to be trivial.
We also require that
\begin{equation}
\label{eq_bnd2}
\delta_R\delta_M = 0,
\end{equation}
as, by definition, the relation syndrome of any Pauli $X$ error is trivial.
Note that the sequence $C_G \rightarrow C_Q \oplus C_M \rightarrow C_S \oplus C_R$ in Eq.~\eqref{eq_diagram} forms a chain complex, which incorporates the standard chain complex $C_G \xrightarrow{\partial_Q} C_Q \xrightarrow{\partial_S} C_S$ describing CSS stabilizer codes~\cite{Bravyi2013b,Bombin2013book}.
We also remark that chain complexes are not only useful for describing the decoding problem---they can also give insights into the procedures of gauging and ungauging stabilizer symmetries~\cite{Vijay2016,Williamson2016,Kubica2018}.

Let $\epsilon\oplus\mu \in C_Q \oplus C_M$ be an error and $\gamma\in C_G$ be a gauge operator, which describes the state of the gauge qubits.
Then, the corresponding measurement outcome $\zeta\in C_M$ does not only depend on $\epsilon\oplus\mu$, but also on $\gamma$, namely
\begin{equation}
\label{eq_phi}
\zeta = \delta_M \epsilon + \mu + \delta_M\partial_Q \gamma.
\end{equation}
We remark that in a generic setting, when we perform interleaved measurement rounds of $X$- and $Z$-type gauge operators, the gauge qubits will be in a random state.
Thus, we should treat two measurement outcomes $\zeta,\zeta' \in C_M$ as equivalent iff they differ by an element from $\im(\delta_M\partial_Q)$.

We can phrase the decoding problem for CSS subsystem codes as follows:
given the measurement outcome $\zeta\in C_M$, find some Pauli $X$ recovery operator $\chi\in C_Q$, which attempts to correct the Pauli $X$ error $\epsilon$.
In the decoding problem, we do not use $\zeta$ explicitly;
rather, we use the stabilizer syndrome $\delta_S\zeta\in C_S$ and the relation syndrome $\delta_R\zeta\in C_R$.
Note that we can infer $\delta_S \zeta$ from measurement outcomes of some set of independent $Z$-type gauge operators.
On the other hand, to infer $\delta_R \zeta$, we need measurement outcomes for an overcomplete set of $Z$-type gauge operators.
Using Eqs.~\eqref{eq_bnd1}-\eqref{eq_phi}, we obtain that $\delta_S\zeta$ does not depend on the gauge operator $\gamma$, i.e., 
\begin{equation}
\delta_S\zeta = \partial_S\epsilon + \delta_S\mu,
\end{equation}
and $\delta_R \zeta$ only depends on the measurement error $\mu$, i.e.,
\begin{equation}
\delta_R\zeta = \delta_R\mu.
\end{equation}

\begin{figure*}[ht!]
\centering
\includegraphics[width=.3\textwidth]{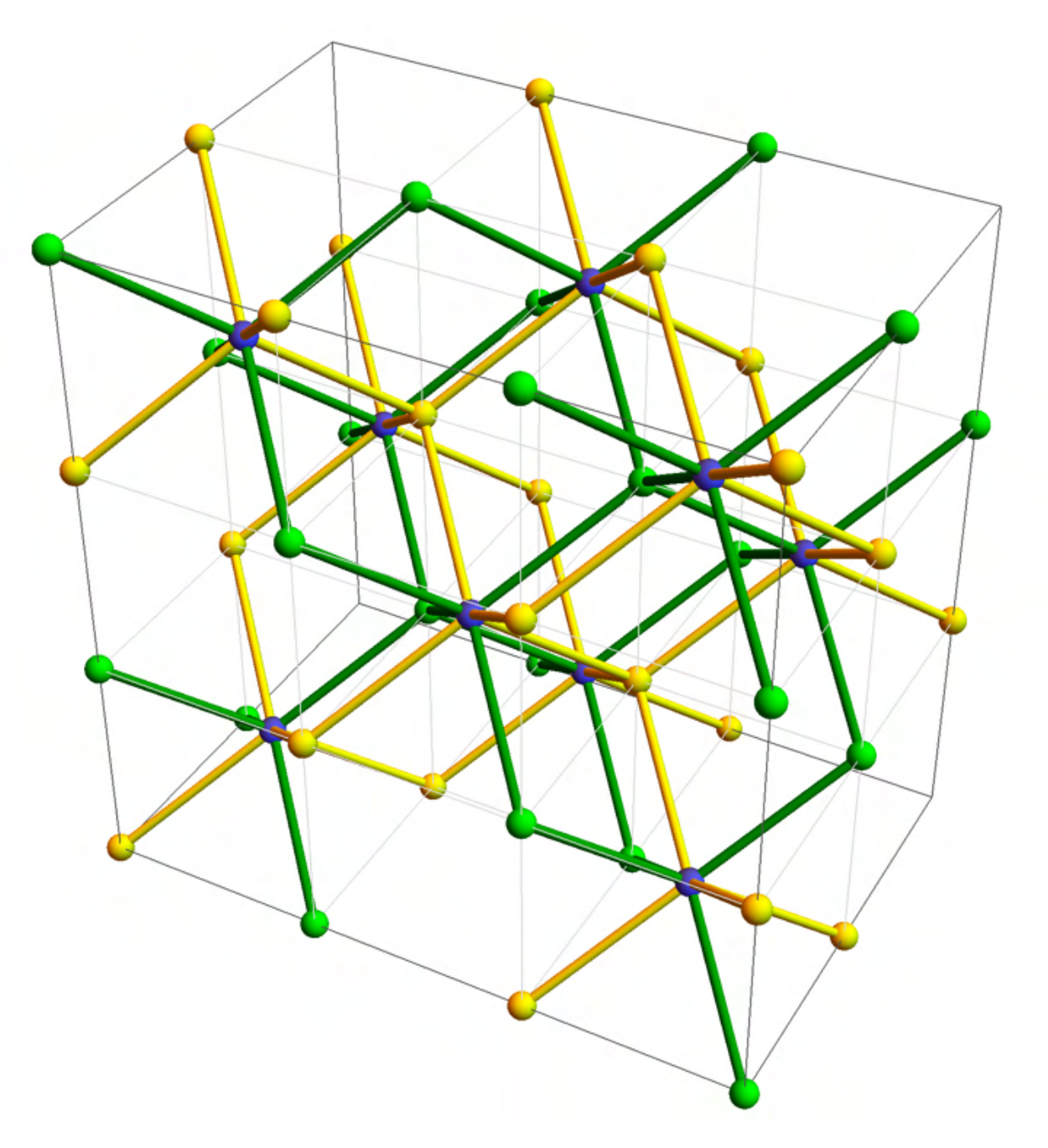}
\hspace*{-.28\textwidth}(a)\hspace*{.28\textwidth}
\includegraphics[width=.3\textwidth]{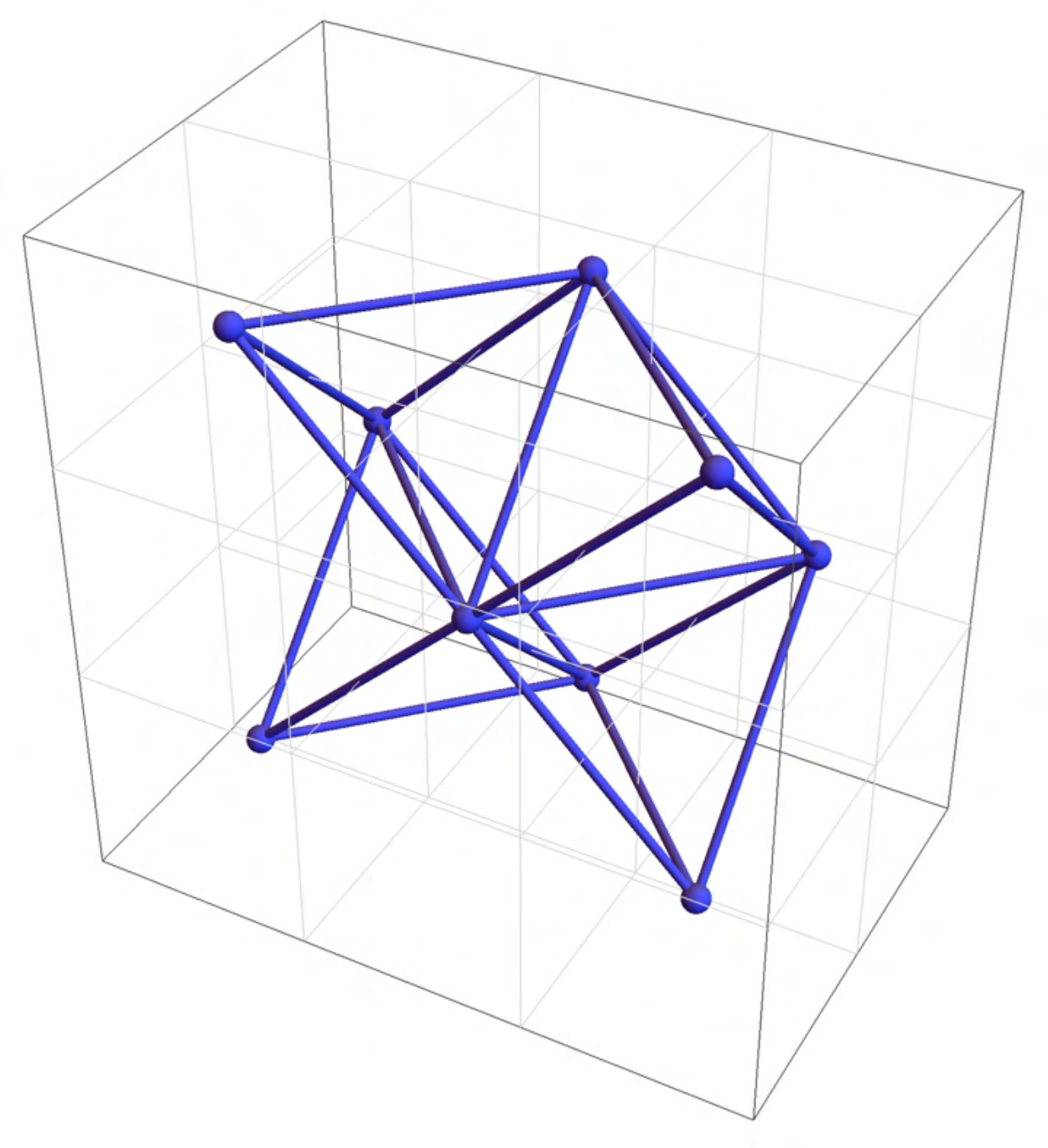}
\hspace*{-.285\textwidth}(b)\hspace*{.28\textwidth}
\includegraphics[width=.3\textwidth]{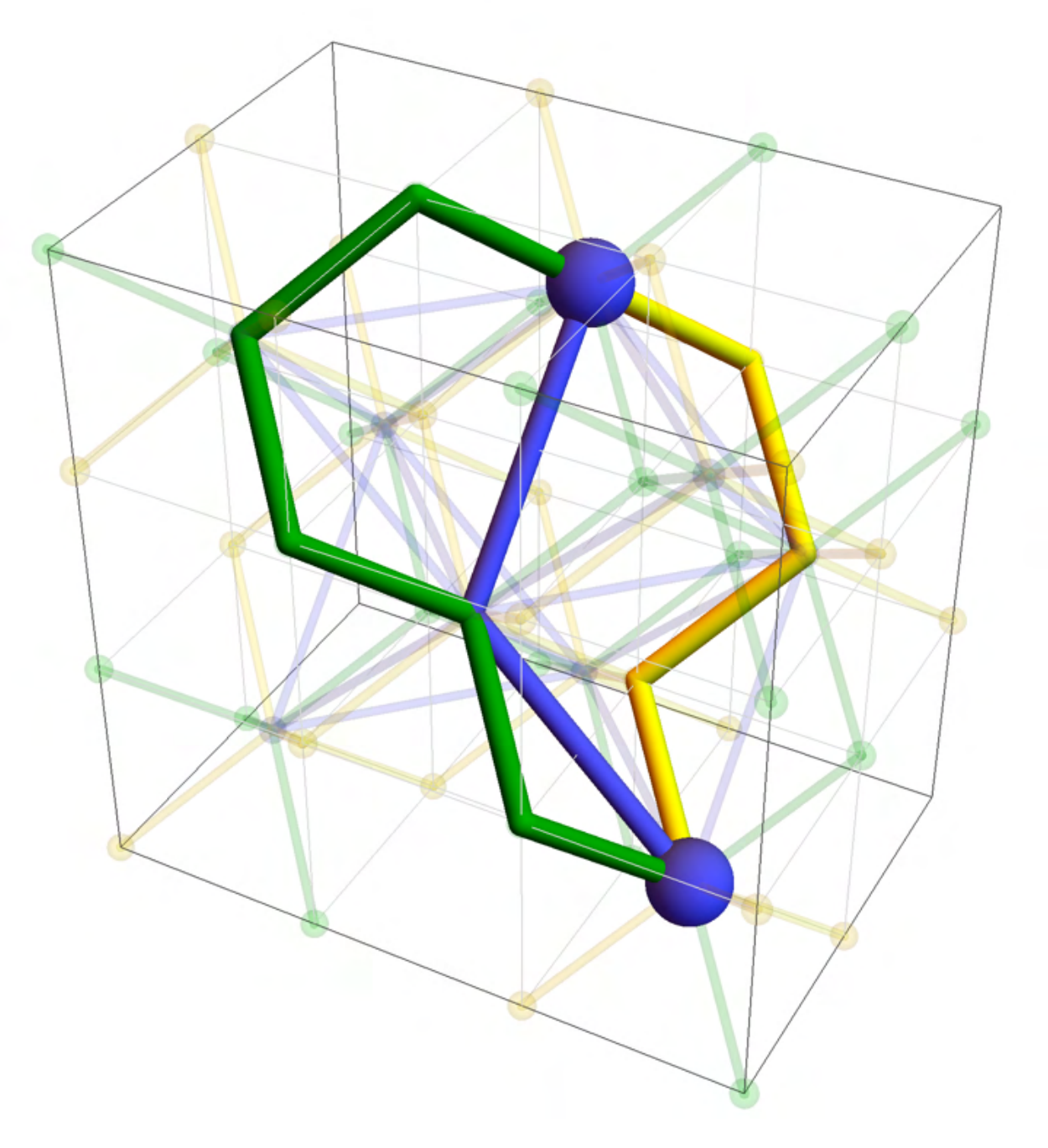}
\hspace*{-.28\textwidth}(c)\hspace*{.26\textwidth}
\caption{
For the 3D~STC on the lattice $\mathcal L$ we construct (a) the measurement graph $G_\text{mea}$ and (b) the qubit graph $G_\text{qub}$.
We use $G_\text{mea}$ and $G_\text{qub}$ in the single-shot MWPM decoder for syndrome estimation and ideal MWPM decoding, respectively.
(c) An example of a flux $\phi$ (green and yellow edges), which may arise for a Pauli $X$ error $\epsilon$ and its stabilizer syndrome $\partial_S \epsilon$ (blue edges and balls).
Note that $\phi$ satisfies the Gauss law and is not uniquely specified by $\epsilon$.
}
\label{fig_graph_flux}
\end{figure*}

We now consider the following decoding strategy.
First, using the relation syndrome $\delta_R \zeta$ of the measurement outcome $\zeta$, we find an estimate $\hat\mu$ of the measurement error $\mu$, such that 
$\delta_R\hat\mu = \delta_R \zeta$.
Then, we compute an estimate $\hat\sigma$ of the stabilizer syndrome $\partial_S \epsilon$ of the error $\epsilon$ as follows
$\hat\sigma = \delta_S(\zeta +\hat\mu) = \partial_S\epsilon + \delta_S(\mu+\hat\mu)$.
Lastly, we find a recovery operator $\chi$, such that $\partial_S \chi = \hat\sigma$.
In the rest of this section we describe how to find the measurement error estimate $\hat\mu$ and the recovery operator $\chi$.

\subsection{Physics of the gauge flux in the 3D~STC}
\label{sec_flux}

We start this subsection by illustrating the decoding problem for the 3D~STC using the measurement and qubit graphs; see Fig.~\ref{fig_graph_flux}.
The measurement graph $G_\text{mea} = (V_\text{mea},E_\text{mea})$ is a sublattice of $\mathcal L$, where $V_\text{mea}$ are the $B$, $G$ and $Y$ vertices of $\mathcal{L}$, i.e.,
\begin{equation}
V_\text{mea} = \mathcal{L}^B_0\cup\mathcal{L}^G_0\cup\mathcal{L}^Y_0,
\end{equation}
and $E_\text{mea}$ are the $BG$ and $BY$ edges of $\mathcal{L}$, i.e.,
\begin{equation}
E_\text{mea} = \mathcal{L}^{BG}_1\cup\mathcal{L}^{BY}_1.
\end{equation}
The qubit graph $G_\text{qub} = (V_\text{qub},E_\text{qub})$ is constructed by taking the $B$ vertices of the lattice $\mathcal{L}$, i.e.,
\begin{equation}
V_\text{qub} = \mathcal{L}^B_0,
\end{equation}
and adding edges between any two different $B$ vertices that belong to the same octahedron in $\mathcal L$.
Note that those edges are not present in $\mathcal{L}$, however they are in the one-to-one correspondence with the octahedral volumes of $\mathcal{L}$, and thus we can make an identification
\begin{equation}
E_\text{qub} = \mathcal{L}_3.
\end{equation}

Let $V'_\text{mea}\subseteq V_\text{mea}$ and $V'_\text{qub}\subseteq V_\text{qub}$ denote the sets of all vertices belonging to the boundary $\partial\mathcal L$ of the lattice $\mathcal{L}$.
We refer to the vertices in $V'_\text{mea}$ and $V'_\text{qub}$ as the boundary vertices.
The maps $\partial_S: C_Q\rightarrow C_S$, $\delta_R: C_M\rightarrow C_R$ and $\delta_S:C_M\rightarrow C_S$ in Eq.~\eqref{eq_diagram} can then be interpreted as the relative boundary operators with respect to the boundary vertices for the qubit graph $G_\text{qub}$, the measurement graph $G_\text{mea}$ and a subgraph of $G_\text{mea}$ containing only the $BG$ edges.

It will be convenient to introduce a notion of the $Z$-type gauge flux $\phi \in C_M$, which is defined to be the $Z$-type gauge measurement outcome in the absence of measurement errors, namely
\begin{equation}
\phi = \delta_M \epsilon + \delta_M\partial_Q \gamma,
\end{equation}
where $\epsilon\in C_Q$ is a Pauli $X$ error and $\gamma\in C_G$ is an $X$-type gauge operator.
From Eq.~\eqref{eq_bnd2} we obtain that the relation syndrome of the flux $\phi$ is zero, i.e.,
\begin{equation}
\label{eq_gauss}
\delta_R \phi = 0.
\end{equation}

\begin{figure*}[ht!]
\centering
\includegraphics[width=0.3\textwidth]{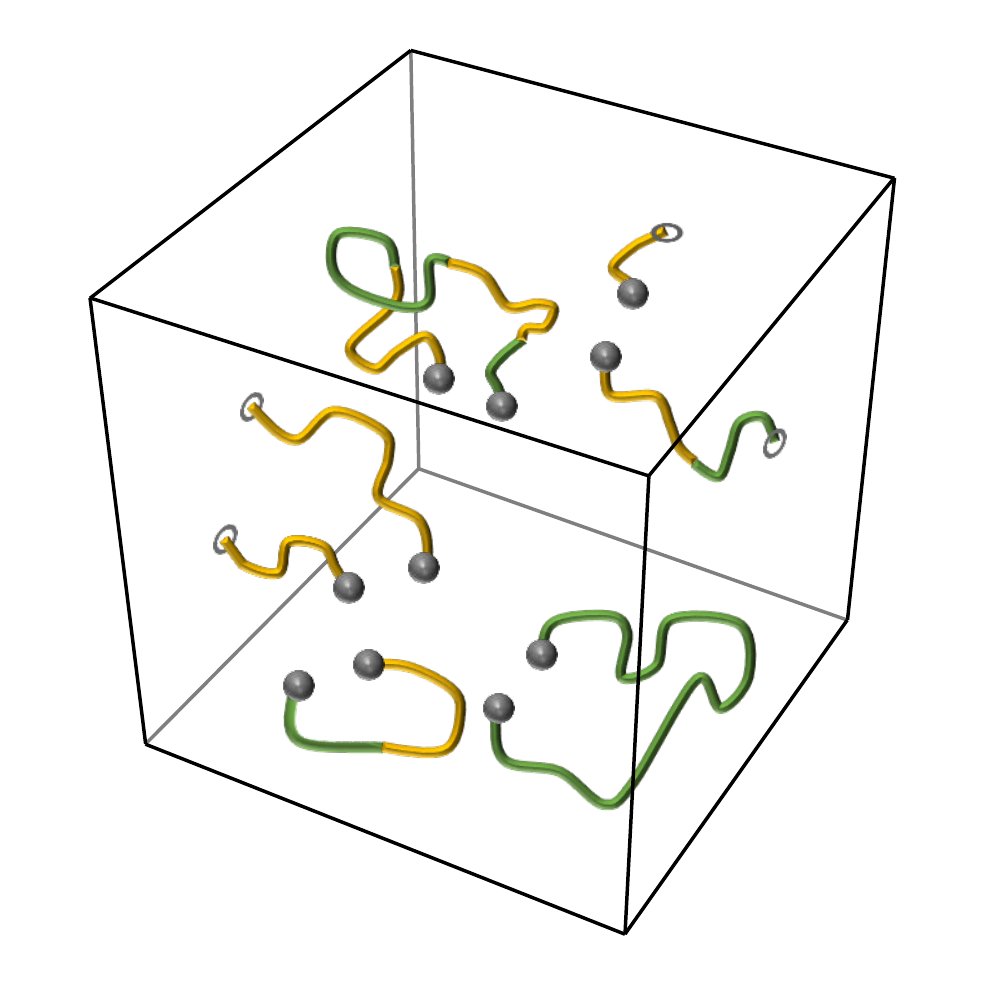}
\hspace*{-.28\textwidth}(a)\hspace*{.28\textwidth}
\includegraphics[width=0.3\textwidth]{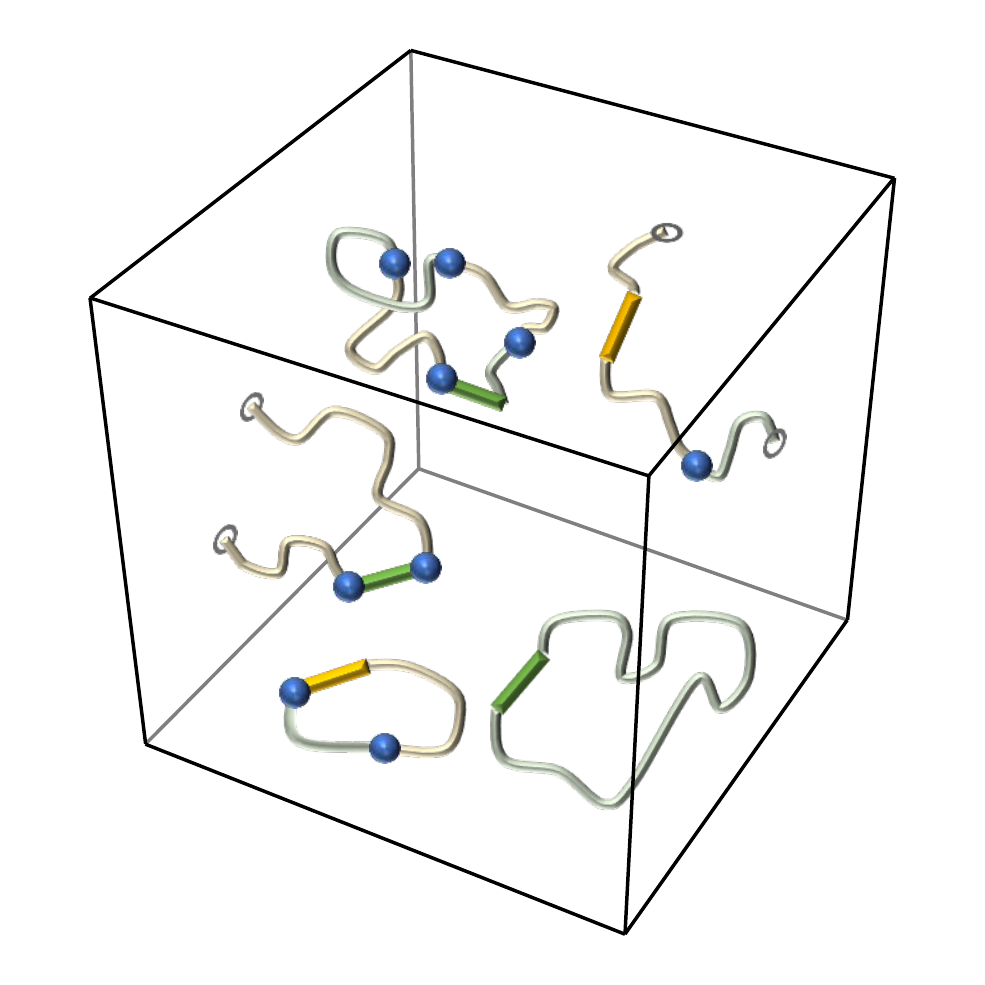}
\hspace*{-.285\textwidth}(b)\hspace*{.28\textwidth}
\includegraphics[width=0.3\textwidth]{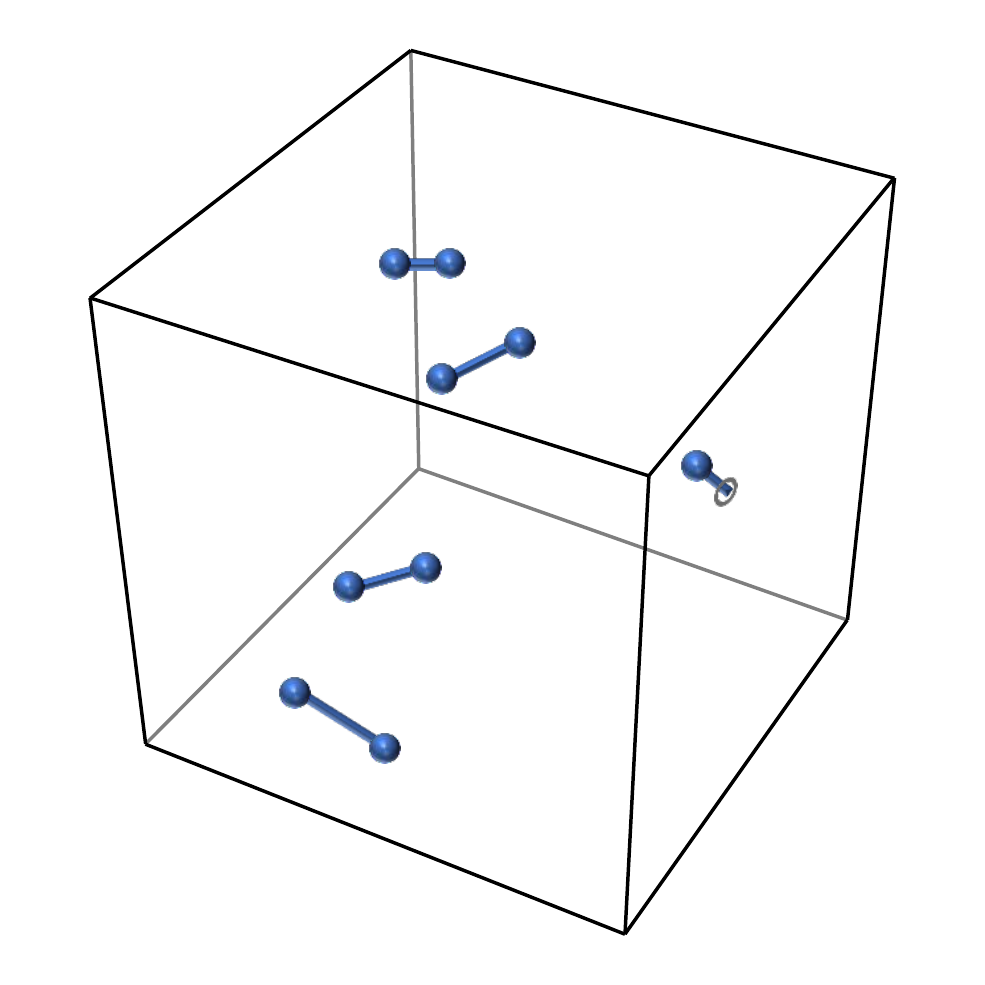}
\hspace*{-.28\textwidth}(c)\hspace*{.27\textwidth}
\caption{
An illustration of the single-shot MWPM decoder.
(a) Due to the measurement error $\mu$, the measurement outcome $\zeta$ (yellow and green lines) has non-zero relation syndrome $\delta_R \zeta$ (gray dots), i.e., $\zeta$ violates the Gauss law.
(b) In step (i), we first find the minimum-weight estimate $\hat\mu$ (yellow and green lines) of $\mu$, such that $\zeta+\hat\mu$ has trivial relation syndrome, i.e., $\delta_R(\zeta+\hat\mu) = 0$.
Then, we compute a syndrome estimate $\hat\sigma$ (blue dots).
(c) In step (ii), we find the minimum-weight recovery operator $\chi$ (blue lines) for $\hat\sigma$.
}
\label{fig_SSdecoders}
\end{figure*}

The Pauli $X$ error $\epsilon$ and the flux $\phi$ correspond to the subsets of edges of the qubit and measurement graphs $G_\text{qub}$ and $G_\text{mea}$, respectively.
In Fig.~\ref{fig_graph_flux}(c) we illustrate a flux $\phi$, which may arise from some Pauli $X$ error $\epsilon$.
A local constraint captured by Eq.~\eqref{eq_gauss}, which we refer to as the Gauss law, arises from the redundancies among gauge generators specified in Eqs.~\eqref{eq_relations_gauge2}-\eqref{eq_relations_gauge4}.
We can equivalently interpret this constraint as follows---although the flux $\phi$ can be random, it forms a collection of strings within $G_\text{mea}$, where each string can only terminate at the boundary vertices $V'_\text{mea}$.
By multiplying all gauge generators supported on edges incident to $v$ we obtain the identity operator, and thus the number of operators returning $-1$ measurement outcome has to be even.
Moreover, from Eq.~\eqref{eq_Zstab} we conclude that whenever the stabilizer $Z(v)$ is violated for some 
vertex $v \in V_\text{qub}\setminus V'_\text{qub}$, then the number of gauge generators supported on $BG$ edges (or $BY$ edges) incident to $v$ and returning $-1$ measurement outcomes, has to be odd.
Thus, for the given Pauli $X$ error $\epsilon$ and the flux $\phi$ the corresponding stabilizer syndrome $\sigma$ can be either found as the endpoints of strings in $\epsilon$ or, equivalently, as the vertices incident to an odd number of $BG$ edges (or $BY$ edges) in $\phi$.

We remark that the physics of the gauge flux is qualitatively similar in the 3D~STC and the 3D gauge color code~\cite{Bombin2015,Bombin2016dim,Bombin2018,Beverland2021}.
However, the behavior of the 3D~STC gauge flux is substantially simpler. 
For instance, in the 3D~STC the $Z$-type gauge flux has two types\footnote{
The flux type is given by the color of the corresponding edge in the measurement graph $G_\text{mea}$.},
compared to six types for the 3D gauge color code.
Also, the flux in the 3D gauge color code may have branching points where fluxes of three certain types meet; this phenomenon is not present in the 3D~STC.

\subsection{Single-shot MWPM decoder for the 3D~STC}
\label{sec_ssmwpm}

\begin{figure*}[ht!]
\centering
\includegraphics[width=.98\textwidth]{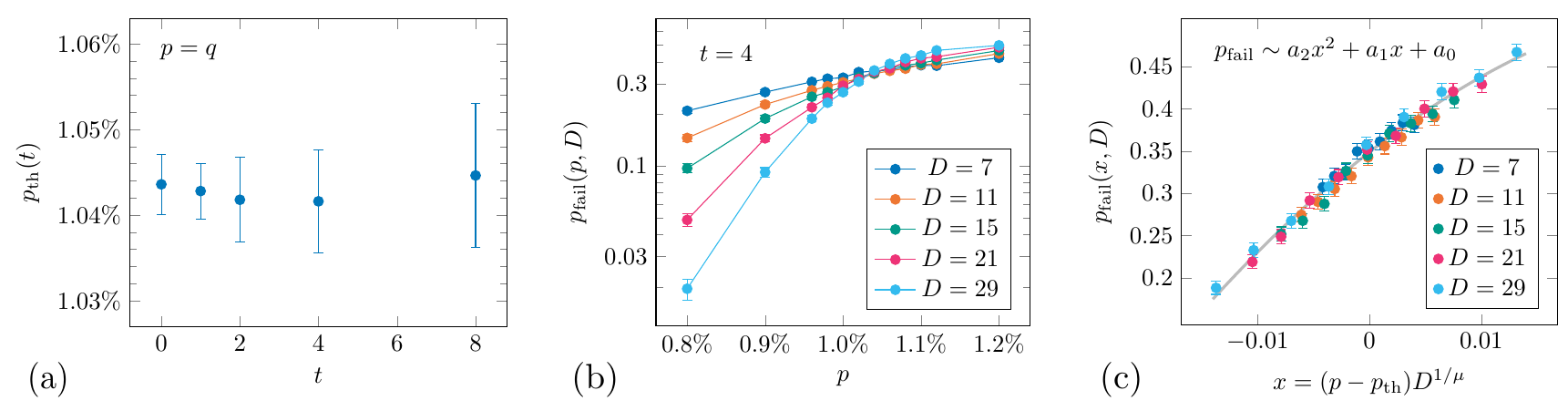}
\caption{
Numerical estimates of the single-shot MWPM decoder threshold for the 3D~STC on the lattice $\cub$ based on the cubic lattice with open boundary conditions.
(a) The threshold $p_\text{th}(t)$ does not change noticeably with the number of correction cycles $t$.
We estimate the storage threshold \psus.
(b) The failure probability $p_\text{fail}(p,D)$ after $t = 4$ correction cycles, where $p$ is the bit-flip error rate, $D$ is the code distance, and we set the measurement error rate $q=p$.
In (c), we show the same data using the rescaled variable $x = (p-p_\text{th}(t))D^{1/\mu}$, where the fitting parameters are $p_\text{th}(t) = 0.01042(6)$ and $\mu = 1.2(1)$.
}
\label{fig_threshold_all}
\end{figure*}

Now, we are ready to introduce the single-shot MWPM decoder for the 3D~STC defined on the lattice $\mathcal L$.
It consists of the following two steps.
\begin{itemize}
\item[(i)] {\bf Syndrome estimation.---}We exploit the consistency checks on the measurement outcomes of $Z$-type gauge operators.
Namely, for the given measurement outcome $\zeta$ we find the minimum-weight estimate $\hat\mu$ of the measurement error $\mu$, i.e.,
\begin{equation}
\hat\mu = \argmin_{\mu'\in C_M: \delta_R \mu'= \delta_R \zeta} |\mu'|.
\end{equation}
Finding $\hat\mu$ is an instance of the MWPM problem for the measurement graph $G_\text{mea}$. 
Then, we compute an estimate $\hat\sigma$ of the stabilizer syndrome $\partial_S\epsilon$ of the error $\epsilon$ as follows
\begin{equation}
\hat\sigma = \delta_S (\zeta+\hat\mu).
\end{equation}
\item[(ii)] {\bf Ideal MWPM decoding.---}For the given syndrome estimate $\hat\sigma$ we find the corresponding minimum-weight recovery operator $\chi$, i.e., 
\begin{equation}
\chi = \argmin_{\chi'\in C_Q: \partial_S \chi'= \hat\sigma} |\chi'|.
\end{equation}
Finding $\chi$ is an instance of the MWPM problem for the qubit graph $G_\text{qub}$.
\end{itemize}
We illustrate how the the single-shot MWPM decoder works in Fig.~\ref{fig_SSdecoders}(a)-(c).
We emphasize that in the presence of measurement errors it is likely that there will be some residual Pauli $X$ error left in the system after applying the single-shot MWPM decoder.

\subsection{Numerical simulations of the performance}
\label{sec_numericalsimu}

We numerically estimate the performance of the single-shot MWPM decoder for the 3D~STC on the lattice $\cub$ from Sec.~\ref{sec_model}, which is based on the cubic lattice with open boundary conditions.
We assume the independent and identically distributed bit-flip noise with probability $p$ and set the measurement error rate to match the bit-flip error rate, i.e., $q = p$.

\begin{figure*}[ht!]
\centering
\includegraphics[width=.98\textwidth]{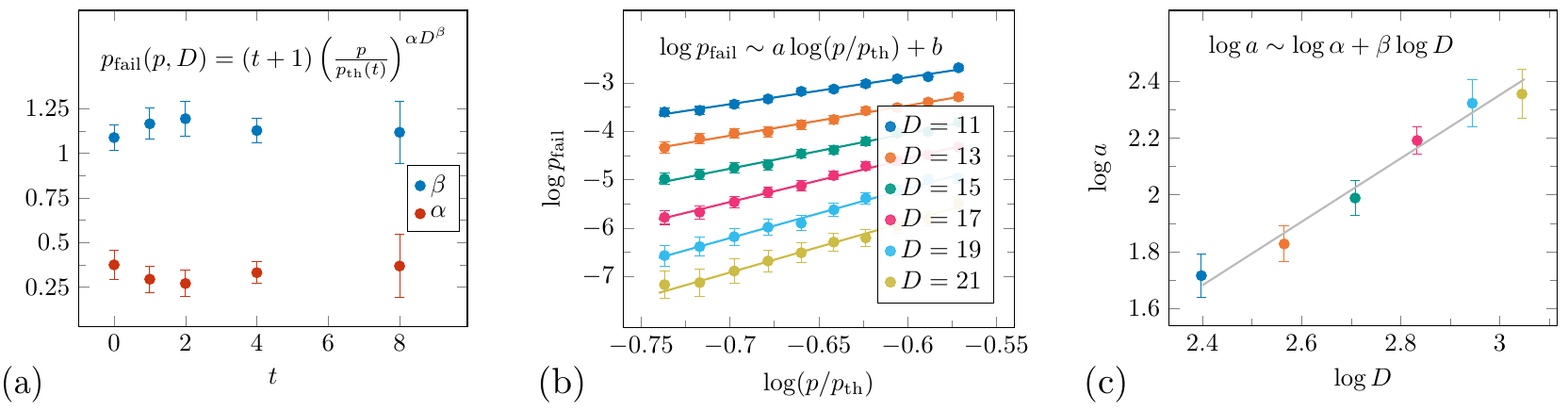}
\caption{
Subthreshold behavior of the single-shot MWPM decoder.
(a) We estimate the parameters $\alpha(t)$ and $\beta(t)$ of the ansatz in Eq.~\eqref{eq:pfail_fit} for the logical failure probability.
We find that both $\alpha(t)$ and $\beta(t)$ are stable across the range of simulated values of $t$.
(b)-(c) An illustration of the fitting procedure for $t=8$.
(b) We plot $\log p_{\rm fail}$ as a function of $\log p/p_{\rm th}$ for different values of the code distance $D$.
We fit the data to straight lines to obtain estimates of $a$ as a function of $D$.
(c) We plot $\log a = \log \alpha(t) + \beta(t) \log D$ as a function of $\log D$ and fit to a straight line to estimate $\alpha(t)$ and $\beta(t)$.
}
\label{fig_subthreshold_all}
\end{figure*}

In our Monte Carlo simulations, we first perform a fixed number of correction cycles $t$.
We start by initializing the residual error to zero, i.e., $\rho = 0$.
Then, in each correction cycle we:
(i) update the existing residual error $\rho$ by adding a randomly chosen error $\epsilon$ to it,
(ii) select a gauge operator $\gamma$ uniformly at random,
(iii) choose a random measurement error $\mu$,
(iv) find the measurement outcome $\zeta = \delta_M \rho + \mu + \delta_M\partial_Q \gamma$,
(v) use the single-shot MWPM decoder to find a recovery operator $\chi$,
(vi) update the residual error $\rho$ by adding the recovery operator $\chi$ to it.
After $t$ error correction cycles, we add a randomly chosen error to the residual error, extract the measurement outcome with no measurement error, use the single-shot MWPM decoder to find a recovery operator, which returns the state to the code space, and, finally, check for a logical error.
This, in turn, allows us to estimate the threshold $p_{\textrm{th}}(t)$; see Fig.~\ref{fig_threshold_all}(b)(c).
Note that $p_{\textrm{th}}(0)$ corresponds to the code capacity threshold.
Moreover, we are interested in the storage threshold of the 3D~STC, which we define as the limit of the threshold $p_{\textrm{th}}(t)$ as the number of correction cycles $t$ goes to infinity, i.e.,
\begin{equation}
p_{\textrm{STC}} = \lim_{t\rightarrow\infty} p_{\textrm{th}}(t).
\end{equation}
We observe that the threshold $p_{\textrm{th}}(t)$ does not change noticeably with $t$; see Fig.~\ref{fig_threshold_all}(a).
We estimate the storage threshold to be \psus.

We also investigate the behavior of the single-shot MWPM decoder in the subthreshold regime.
We use the following ansatz 
\begin{equation}
\label{eq:pfail_fit}
p_{\rm fail}(p, D, t) = (t + 1) \left(\frac{p}{p_{\rm th}(t)} \right)^{\alpha(t) D^{\beta(t)}}
\end{equation}
for the logical failure probability $p_{\rm fail}(p, D, t)$ as a function of the error rate $p$, the code distance $D$, and the number of correction cycles $t$.
In Eq.~\eqref{eq:pfail_fit}, $\alpha(t)$ and $\beta(t)$ are fitting parameters, which may depend on $t$, and 
$p_{\rm th}(t)$ is obtained from threshold fits similar to the ones depicted in Fig.~\ref{fig_threshold_all}(b)(c).
To estimate $\alpha(t)$ and $\beta(t)$, we first fix $t$ and take the logarithm of both sides of Eq.~\eqref{eq:pfail_fit} to obtain
\begin{equation}
\label{eq:lin_fit_1}
\log p_{\mathrm{fail}}(p, D, t) = \log (t+1) + a \log (p/p_{\mathrm{th}}(t)),\quad\quad a = \alpha(t) D^{\beta(t)}.
\end{equation}
Then, for different values of $D$, we plot $\log p_{\mathrm{fail}}(p, D, t)$ as a function of $\log (p/p_{\mathrm{th}}(t))$ and fit to a straight line to estimate $a$; see Fig.~\ref{fig_subthreshold_all}(b).
We then take the logarithm again and obtain
\begin{equation}
\label{eq:lin_fit_2}
\log a = \log \alpha(t) + \beta(t) \log D.
\end{equation}
Finally, we plot $\log a$ as a function of $\log D$ and fit to a straight line to get $\alpha(t)$ and $\beta(t)$;
see Fig.~\ref{fig_subthreshold_all}(c).
As the weight of the smallest logical operator is equal to $D$, we expect $\beta(t) \sim 1$, which is indeed what we observe.
Furthermore, the values of $\alpha(t)$ and $\beta(t)$ are stable for the values of $t$ that we simulated; see Fig.~\ref{fig_subthreshold_all}(a).

\section{Proof of single-shot QEC}
\label{sec_proof}

Here we prove that single-shot QEC is possible with the 3D~STC.
For concreteness, we assume that decoding is performed with the single-shot MWPM decoder from Sec.~\ref{sec_ssmwpm}.
Our proof is inspired by the seminal work by Bomb\'in~\cite{Bombin2015} and uses somewhat similar notation.
Although our presentation is centered around the 3D~STC and the single-shot MWPM decoder, we hope that it provides valuable insights into single-shot QEC in general.

\subsection{Setting the stage}

Let $p:2^A\rightarrow [0,1]$ be a discrete probability distribution over a collection $2^A$, which is the power set of some finite set $A$.
We are particularly interested in discrete probability distributions that describe Pauli and measurement errors happening in the system.
We say that $p$ is $\tau$-bounded with prefactor $c$ iff for any set $B\subseteq A$ the probability that a set $B'\subseteq A$ drawn according to the probability distribution $p$ contains $B$ is at most $c\tau^{|B|}$, i.e.,
\begin{equation}
\forall B\subseteq A: \sum_{B'\supseteq B} p(B') \leq c\tau^{|B|}.
\end{equation}
Unless we specify the prefactor explicitly, we assume $c=1$ by default.

As discussed before, since 3D~STC is a CSS code, we can focus our discussion on correcting Pauli $X$ errors, because Pauli $Z$ errors can be corrected analogously.
We can describe any stochastic Pauli $X$ noise model as a Pauli $X$ channel, i.e., a completely positive trace-preserving map admitting a Kraus representation
\begin{equation}
\N = \left\{ \sqrt{p_\N(N)} N \right\}_N,
\end{equation}
where $p_\N$ is a discrete probability distribution over Pauli $X$ operators, i.e., $\sum_{N} p_\N(N)  = 1$.
We denote the collection of all Pauli $X$ channels by
\begin{equation}
\mathbb P^X = \{ \N \mathrel{|} \text{$\N$ is a Pauli $X$ channel}\}.
\end{equation}
We say that $\N\in\mathbb P^X$ is $\tau$-bounded with prefactor $c$ iff the corresponding probability distribution $p_\N$ is $\tau$-bounded with prefactor $c$, i.e.,
\begin{equation}
\label{eq_taubounded}
\forall N: \sum_{N'\supseteq N} p_\N(N') \leq c\tau^{|N|}.
\end{equation}

We remark that in Eq.~\eqref{eq_taubounded} as well as in the rest of the article we write $N$ and $M$ to denote both the Pauli $X$ operators themselves, as well as their support, however it should be clear from the context what we mean.
For instance, $\partial N$ denotes the stabilizer syndrome of $N$, $N \cup M$ and $N\cap M$ denote, respectively, the union and intersection of the supports of $N$ and $M$, $NM$ denotes either the Pauli $X$ operator, which is the product of $N$ and $M$, or its support.

For any two channels $\N = \{ \sqrt{p_\N(N)} N \}_{N}$ and $\M = \{ \sqrt{p_\M(M)} M \}_{M}$ we define their composition as
\begin{equation}
\N \circ \M = \left\{ \sqrt{p_\N(N)p_\M(M)} N M \right\}_{N,M}.
\end{equation}
Lastly, for any two collections of channels $\mathbb{N}$ and $\mathbb{M}$, we define their composition as the following collection of channels
\begin{equation}
\mathbb{N}\circ \mathbb{M} = \{ \N \circ \M \mathrel{|} \N\in\mathbb{N},\M\in\mathbb{M} \}.
\end{equation}

Let $G = (V,E)$ be a hypergraph consisting of vertices $V$ and hyperedges $E$.
Note that a hypergraph is a generalization of a graph, in which a hyperedge is a non-empty subset of vertices.
For any subset of vertices $V'\subseteq V$ we can construct a new hypergraph $G' = (V\setminus V', \{e\setminus V' | e\in E\})$ by removing all the vertices $V'$ from the hypergraph $G$.
We say that a subset $A$ of hyperedges $E$ is connected if a subhypergraph of $G$ that includes all the hyperedges in $A$ is connected.
For any $A\subseteq E$ we can find a unique decomposition $A = \bigcup_i A_i$ in terms of its connected components, where each $A_i$ is a maximal connected subset of $A$.
We say that a subset $B$ of hyperedges $E$ is a connected cover of $A$ iff $B$ contains $A$ and every connected component $B_i$ of $B$ intersects $A$, i.e., $B\supseteq A$ and $B_i \cap A \neq \emptyset$.
We denote by $\cl(s,A)$ the collection of connected covers of $A$ that contain $s$ hyperedges.
Lastly, for any $A,B\subseteq E$ we define $A^{\cap B}$ to be the union of all connected components $A_i$ of $A$ that intersect $B$, namely
\begin{equation}
A^{\cap B} = \bigcup_{i: A_i \cap B \neq \emptyset} A_i.
\end{equation}
Note that if $A\supseteq B$, then $A^{\cap B}$ is a connected cover of $B$ of size $s=|A^{\cap B}|$, i.e., $A^{\cap B}\in \cl(s,A)$.

The following lemma will be useful in bounding the number of connected covers.
\begin{lemma}[connected covers]
\label{lemma_connected_cover}
Let $A\subseteq E$ be a subset of hyperedges in a finite hypergraph \mbox{$G = (V,E)$}.
Then, the number of connected covers of $A$ that contain $s$ hyperedges satisfies
\begin{equation}
|\cl(s,A)| \leq \frac{(\e z)^s}{\e z^{|A|}},\quad\quad
z = \max_{e \in E} \sum_{v\in e}(\deg v - 1),
\end{equation}
where $\e$ is Euler's number and $\deg v$ denotes the degree of a vertex $v$ of $G$.
\end{lemma}
\begin{proof}
Let $G' = (V',E')$ be a graph, whose vertices correspond to the hyperedges of $G$, i.e., $V' = E$.
We connect two vertices in $G'$ with an edge iff the corresponding hyperedges in $G$ are incident to the same vertex.
Note that maximum degree of the graph $G'$ is then $z$.
Let $A'\subseteq V'$ be the subset of vertices corresponding to the subset of hyperedges $A\subseteq E$.
Then, the task of counting the number of connected covers of $A$ in $G$ is equivalent to the cluster counting problem in Lemma~5 in Ref.~\cite{Aliferis2007} for $A'$ in $G'$.
We thus obtain the upper bound $\e^{|A|-1}(\e z)^{s-|A|}$.
\end{proof}

\subsection{3D~STC lattice}
\label{sec_stclattice}

In the rest of this section we focus our attention on the 3D~STC defined on some lattice $\mathcal{L}$.
Using the notions from Sec.~\ref{sec_flux}, let $G_\text{mea}$ and $G_\text{qub}$ be the measurement and qubit graphs associated with $\mathcal{L}$, and let $V'_\text{mea}$ and $V'_\text{qub}$ be the sets of boundary vertices of $G_\text{mea}$ and $G_\text{qub}$, respectively.
Recall that $V'_\text{mea}\supseteq V'_\text{qub}$ as $V_\text{mea}\supseteq V_\text{qub}$.
When we talk about connected components and connected covers we always consider them within a hypergraph $G_\mathcal{L} = (V_\mathcal{L},E_\mathcal{L})$, which is obtained by taking the union of $G_\text{mea}$ and $G_\text{qub}$, and then removing all the boundary vertices, i.e.,
\begin{equation}
G_{\mathcal L} = (V_\text{mea}\setminus V'_\text{mea}, \{e\setminus V'_\text{mea} {\mathrel{|}} e\in E_\text{mea}\cup E_\text{qub}\}).
\end{equation}
Note that $G_{\mathcal L}$ is a hypergraph since every edge in $E_\text{mea}\cup E_\text{qub}$ that is incident to exactly one boundary vertex in $V'_\text{mea}$ is replaced by the other vertex it is incident to.
The hyperedges $E_{\mathcal L}$ correspond to the measurement outcomes and Pauli $X$ errors.
The vertices $V_\mathcal{L}$ and $V_\text{qub}\setminus V'_\text{qub}$ correspond to, respectively, locations where the Gauss law has to be satisfied and the $Z$-type stabilizers.
Let $\Delta_\mathcal{L}$ denote the maximum vertex degree of $G_{\mathcal L}$, i.e., 
\begin{equation}
\Delta_\mathcal{L} = \max_{v\in V_\mathcal{L}} \deg v.
\end{equation}
Let $\Delta_\text{qub}$ and $\Sigma_\text{qub}$ be the maximum vertex degree for interior vertices of $G_\text{qub}$ and the sum of the degrees of all the boundary vertices in $V'_\text{qub}$, namely
\begin{eqnarray}
\Delta_\text{qub} &=& \max_{v\in V_\text{qub}\setminus V'_\text{qub}} \deg v,\\
\Sigma_\text{qub} &=& \sum_{v\in V_\text{qub}'} \deg v.
\end{eqnarray}

In order to prove the existence of a non-zero threshold for the 3D~STC, we need to assume that we have a family of lattices parametrized by a positive integer $L$.
The integer $L$ provides a lower bound on the shortest distance between any two different boundary vertices in the qubit graph $G_\text{qub}$ associated with the lattice $\mathcal L$.
For such a lattice family, we require that $\Delta_\mathcal{L}$ and $\Delta_\text{qub}$ are bounded by some constants, and that $\Sigma_\text{qub}$ scales polynomially in $L$.
Moreover, in order to prove the main result, (single-shot QEC) Theorem~\ref{thm_main}, we need the following assumption about the lattice $\mathcal L$ and the associated hypergraph $G_\mathcal{L}$ to hold.
\begin{itemize}
\item {\bf Confinement of the flux.---}For any $Z$-type gauge flux $\phi$ there exists a Pauli $X$ operator $M$, such that the stabilizer syndromes of $M$ and $\phi$ are the same, i.e., $\partial_S M = \delta_S \phi$, and $M$ is of comparable size to $\phi$, i.e., 
\begin{equation}
|M| \leq c_\text{flu} |\phi|.
\end{equation}
\end{itemize}
We remark that the assumption about confinement of the flux is analogous to the confining property specified in Definition~16 in Ref.~\cite{Bombin2015}.
Moreover, for the lattice $\cub$ of linear length $L$, which we describe in Sec.~\ref{sec_glance}, the constants are
\begin{equation}
\Delta_{\mathcal L}=20, \quad \Delta_\text{qub} = 12, \quad \Sigma_\text{qub} = 4 L^2 + 6 L + 2, \quad c_\text{flu} = 1/4. 
\end{equation}

\subsection{Error correction with and without measurement errors}

We now revisit the question of how to perform QEC, which we already discussed at length in Sec.~\ref{sec_ss_decoding}.
Let us start with the case of ideal error correction, when there are no measurement errors.
To simplify the notation, in the rest of this section we write $\partial$ to denote both $\delta_S$ and $\partial_S$, however it should be clear from the context what we mean.
First, to diagnose Pauli $X$ errors we measure the set of $Z$-type gauge operators in the gauge group $\mathcal G$ of the 3D~STC, which correspond to the edges of the measurement graph $G_\text{mea}$.
Let $\phi$ be the flux, i.e., the set of gauge operators returning $-1$ measurement outcomes.
Knowing the flux $\phi$ we can infer the measurement outcomes for the set of $Z$-type stabilizer operators in the stabilizer group $\mathcal S$ of the 3D~STC, which correspond to the vertices of the qubit graph $G_\text{qub}$.
Namely, if $\sigma$ denotes the stabilizer syndrome, i.e., the set of stabilizer operators returning $-1$ measurement outcomes, then we have $\sigma = \partial\phi$.
Lastly, knowing the stabilizer syndrome $\sigma$ and using the ideal MWPM decoder we can find the minimum-weight Pauli $X$ recovery operator $R_\sigma$, such that $\partial R_\sigma =  \sigma$.

We can succinctly describe ideal error correction by the following channel
\begin{equation}
\label{eq_decoder}
\R_0 = \left\{ R_{\partial\phi} \Pi^\mathcal{G}_\phi \right\}_{\phi}
= \bigg\{ R_{\sigma} \sum_{\phi: \partial\phi = \sigma} \Pi^\mathcal{G}_\phi \bigg\}_{\sigma}
= \left\{ R_{\sigma} \Pi^\mathcal{S}_\sigma \right\}_{\sigma},
\end{equation}
where $\Pi^\mathcal{G}_\phi$ and $\Pi^\mathcal{S}_\sigma = \sum_{\phi : \partial\phi = \sigma} \Pi^\mathcal{G}_\phi$ are the projection operators onto the subspaces with the given flux $\phi$ and stabilizer syndrome $\sigma$, respectively.
We refer to $\R_0$ as the ideal MWPM decoding channel.

Let $\N = \{ \sqrt{p_\N(N)} N \}_N\in\mathbb P^X$ and consider the following composite channel
\begin{equation}
\label{eq_composite}
\R_0 \circ \N \circ \Pi^\mathcal{S}_0 =
\left\{ \sqrt{p_\N(N)} R_\sigma \Pi^\mathcal{S}_\sigma N \Pi^\mathcal{S}_0\right\}_{\sigma, N} =
\left\{ \sqrt{p_\N(N)} R_{\partial N} N \Pi^\mathcal{S}_0\right\}_{N},
\end{equation}
where $\Pi^\mathcal{S}_0$ is the projection operator onto the 3D~STC code subspace, i.e., the subspace with the trivial stabilizer syndrome $\sigma = 0$.
Note that in Eq.~\eqref{eq_composite} we use the fact that for any two stabilizer syndromes $\sigma$ and $\sigma'$, and any Pauli $X$ operator $N$ we have
 $\Pi^\mathcal{S}_{\sigma} \Pi^\mathcal{S}_{\sigma'} = \delta_{\sigma,\sigma'} \Pi^\mathcal{S}_{\sigma'}$
 and $\Pi^\mathcal{S}_\sigma N = N\Pi^\mathcal{S}_{\sigma+\partial N}$, where $\delta_{\sigma,\sigma'}$ denotes the Kronecker delta.
Then, we define 
\begin{equation}
\fail(\N) = \sum_{N : R_{\partial N} N \not\in \mathcal{G}} p_\N (N)
\end{equation}
to be the probability that the channel $\R_0 \circ \N \circ \Pi^\mathcal{S}_0$ implements any non-trivial logical Pauli operator.
We refer to $\fail(\N)$ as the failure probability of the ideal MWPM decoding channel for the Pauli $X$ channel $\N$ or, in short, the failure probability for $\N$.
Moreover, we show the following lemma about the ideal MWPM decoding channel $\R_0$.

\begin{lemma}(decoding failure)
\label{lemma_threshold}
Let $\N = \{ \sqrt{p_\N (N)} N \}_N$ be a $\tau$-bounded Pauli $X$ channel.
If $\tau < \tau^*$, where $\tau^* = (2(\Delta_\mathrm{qub}-1))^{-2}$, then the failure probability of the ideal MWPM decoding channel $\R_0$ for $\N$ satisfies
\begin{equation}
\label{eq_decoderfailure}
\fail(\N) \leq f(\tau),\quad\quad
f(\tau) = \frac{\Sigma_\mathrm{qub} \tau^{*}} {(\tau^{*})^{1/2}-\tau^{1/2}}\left(\frac{\tau}{\tau^*}\right)^{L/2}.
\end{equation}
\end{lemma}

We remark that for the family of lattices considered in Sec.~\ref{sec_stclattice} (decoding failure) Lemma~\ref{lemma_threshold} immediately implies the existence of $\tau^*>0$, such that $\lim_{L\rightarrow\infty} \fail(\N) = 0$ for $\tau < \tau^*$, i.e.,
that there exists a non-zero threshold $\tau^*$ for the 3D~STC with the ideal MWPM decoder and the $\tau$-bounded Pauli $X$ channel $\N$.
Thus, (decoding failure) Lemma~\ref{lemma_threshold} can be viewed as a special case of more general results establishing non-zero thresholds for certain families of quantum low-density parity-check codes in Refs.~\cite{Kovalev2013,Gottesman2013}.
We also note that for the 3D~STC defined on the lattice $\cub$, we have $\tau^* = 1/484 \approx 0.21\%$, which is approximately one-fifth of the numerical value (the $t=0$ data-point in Fig.~\ref{fig_threshold_all}(a)) observed in simulations.

\begin{proof}
Let $\Lambda(i)$ denote the set of simple paths\footnote{
A simple path is a path in a graph that does not have repeating vertices.}
of length $i$ in the qubit graph $G_\text{qub}$, such that each path contains exactly two different vertices in $V'_\text{qub}$, where it starts and ends.
Roughly speaking, $\Lambda(i)$ contains paths connecting different boundaries of $\mathcal L$.
By definition, $\Lambda(i) = \emptyset$ for any $i < L$.
Let us now assume that $i\geq L$.
Since any path $\lambda\in\Lambda(i)$ starts at some boundary vertex in $G_\text{qub}$, the first edge of $\lambda$ can be chosen in $\Sigma_\text{qub}$ ways.
Moreover, every following edge of $\lambda$ can be chosen in at most $\Delta_\text{qub}-1$ ways, as the path $\lambda$ cannot backtrack.
Thus, we arrive at the following upper bound
\begin{equation}
\label{eq_ncp}
|\Lambda(i)| \leq \Sigma_\text{qub} (\Delta_\text{qub}-1)^{i-1}.
\end{equation}

We proceed by first proving the following inequality
\begin{equation}
\label{eq_fail_uneq1}
\fail(\N) = \sum_{N: R_{\partial N} N \not\in\mathcal{G}} p_\N (N)
\leq \sum_{i \geq L} \sum_{\lambda\in \Lambda(i)} 
\sum_{N' \subseteq \lambda: |N'|\geq \frac{i}{2}} \sum_{N\supseteq N'} p_\N (N),
\end{equation}
where in the summation on the right-hand side we treat $N$ and $N'$ as the subsets of qubits supporting the corresponding Pauli $X$ operators rather than those operators themselves.
Let $N$ be a Pauli $X$ error, for which $p_\N (N)$ appears on the left-hand side of the inequality.
Since $R_{\partial N} N \not\in\mathcal{G}$, the support of $R_{\partial N} N$ is guaranteed to contain some path $\lambda\in\Lambda(i)$ for some $i\geq L$, i.e., $\lambda\subseteq R_{\partial N} N$.
We now show that the term $p_\N (N)$ appears on the right-hand side for $N' = N \cap \lambda$.
Let $R' = R_{\partial N} \cap \lambda$.
Then, $\partial N' = \partial R'$ and $|\lambda| = |N'| + |R'|$.
Since $R_{\partial N}$ is the minimum-weight recovery operator with the stabilizer syndrome $\partial N$, we also have $|R'| \leq |N'|$; otherwise, $N'R'R_{\partial N}$ would be a recovery operator with the stabilizer syndrome $\partial N$ of weight smaller than $R_{\partial N}$.
Thus, $|N'| \geq |\lambda|/2 = i/2$.
This, in turn, implies that the term $p_\N(N)$ appears on the right-hand side and, subsequently, establishes the inequality in Eq.~\eqref{eq_fail_uneq1}, as all the terms in the inequality are non-negative.

Using the fact that $p_\N$ is $\tau$-bounded and Eq.~\eqref{eq_ncp} we obtain
\begin{eqnarray}
\sum_{i \geq L} \sum_{\lambda\in \Lambda(i)} 
\sum_{N' \subseteq \lambda: |N'|\geq \frac{i}{2}} \sum_{N\supseteq N'} p_\N(N)
&\leq& \sum_{i \geq L} \sum_{\lambda\in \Lambda(i)} \sum_{N' \subseteq \lambda: |N'|\geq \frac{i}{2}} \tau^{|N'|}
\leq \sum_{i \geq L} \sum_{\lambda\in \Lambda(i)} 2^{i-1} \tau^{i/2}\\
&\leq& \sum_{i \geq L} \Sigma_\text{qub} (\Delta_\text{qub}-1)^{i-1} 2^{i-1} \tau^{i/2}\\
&\leq& \frac{\Sigma_\text{qub} \left(2(\Delta_\text{qub}-1)\tau^{1/2}\right)^L}
{2(\Delta_\text{qub}-1)\left(1- 2(\Delta_\text{qub}-1)\tau^{1/2}\right)},
\end{eqnarray} 
leading to $\fail(\N) \leq f(\tau)$.
\end{proof}

Now, we consider the case of error correction with imperfect measurements.
As before, we diagnose Pauli $X$ errors by measuring $Z$-type gauge operators.
This time, however, we do not learn the flux $\phi$; rather, we register the measurement outcome $\zeta = \phi + \mu$, where $\mu$ is the measurement error, i.e., the set of gauge operators whose corresponding measurement outcomes have flipped.
We assume that the measurement error $\mu$ is independent of the flux $\phi$.
We proceed by finding the minimum-weight measurement error estimate $\hat\mu$, such that $\partial_R \hat\mu = \partial_R \zeta$.
Knowing the measurement outcome $\zeta$ and the measurement error estimate $\hat\mu$ we can infer an estimate $\hat\sigma$ of the stabilizer syndrome $\sigma = \partial\phi$, namely $\hat\sigma = \partial(\phi + \mu + \hat\mu)$.
Lastly, using the ideal MWPM decoder we can find the minimum-weight Pauli $X$ recovery operator $R_{\hat\sigma}$, such that $\partial R_{\hat\sigma} = \hat\sigma$.

We can capture error correction with imperfect measurements with the following channel
 \begin{eqnarray}
\R &=& \left\{ \sqrt{p_\R(\mu)}R_{\partial(\phi+\fix\mu)} \Pi^\mathcal{G}_\phi \right\}_{\mu,\phi}\\
&=& \bigg\{ \sqrt{p_\R(\mu)} R_{\sigma+\partial(\fix\mu)} \sum_{\phi: \partial\phi = \sigma} \Pi^\mathcal{G}_\phi \bigg\}_{\mu,\sigma}
= \left\{ \sqrt{p_\R(\mu)} R_{\sigma+\partial(\fix\mu)} \Pi^\mathcal{S}_\sigma \right\}_{\mu,\sigma},
\end{eqnarray}
where $p_\R(\mu)$ is the probability that the measurement error $\mu$ occurs and $\sum_\mu p_\R(\mu) = 1$.
For brevity, we also refer to $\R$ as the single-shot MWPM decoding channel.

Lastly, we introduce $\mathbb{R}_\eta$ to be a class of channels consisting of all the single-shot MWPM decoding channels, for which the measurement error probability distributions are $\eta$-bounded, i.e.,
\begin{equation}
\mathbb{R}_\eta = \{ \text{$\R$ is a single-shot MWPM decoding channel}\mathrel{|} \text{$p_\R$ is $\eta$-bounded} \}.
\end{equation}
Note that the ideal MWPM decoding channel $\R_0$ can be viewed as the single-shot MWPM decoding channel with $p_{\R_0} (\mu) = \delta_{\mu,0}$.
In that case, the measurement error probability distribution $p_{\R_0}$ is $0$-bounded and thus $\R_0 \in \mathbb{R}_0$.

\subsection{More on channels}

Let $\N = \{ \sqrt{p_\N(N)} N \}_N\in\mathbb P^X$ and $\R_0$ be the ideal MWPM decoding channel defined in Eq.~\eqref{eq_decoder}.
We equivalently express $\N$ in the following decoder-dependent form
\begin{equation}
\N = \{ \sqrt{p_\N(N)} R_{\partial N} (R_{\partial N} N) \}_N.
\end{equation}
Note that the recovery operator $R_{\partial N}$ has the same stabilizer syndrome as $N$, i.e., 
$\partial R_{\partial N} = \partial N$, whereas $R_{\partial N} N$ forms a logical (possibly trivial) Pauli operator, i.e., $R_{\partial N} N \in \mathcal{Z}(\mathcal S)$.
Then, we define the $\R_0$-dependent channel $\overline\N$ for the Pauli $X$ channel $\N$ as follows
\begin{equation}
\overline\N = \{ \sqrt{p_\N(N)} R_{\partial N} \}_N
=\left \{ \sqrt{p_{\overline\N}(N')} N' \right\}_{N'},
\end{equation}
where we group the same Pauli $X$ terms and introduce
\begin{equation}
p_{\overline\N}(N') = \sum_{N : R_{\partial N} = N'} p_\N(N).
\end{equation}
By definition, we have $\fail(\overline\N) = 0$.

For completeness, we reprove the following lemma.
\begin{lemma}[Lemma 1 in Ref.~\cite{Bombin2015}]
\label{lemma_failure}
For any two $\N,\M\in \mathbb P^X$ we have
\begin{eqnarray}
\label{eq_doublebar}
&\overline{\N\circ\M} = \overline{\overline\N\circ\overline\M},&\\
\label{eq_failfirst}
&\fail(\N\circ\M) \leq \fail(\N) + \fail(\M) + \fail(\overline{\N}\circ\overline{\M}),&\\
\label{eq_failsecond}
&\fail(\overline\N\circ\overline\M) \leq \fail(\N) + \fail(\M) + \fail(\N\circ\M).&
\end{eqnarray}
\end{lemma}

We remark that although we define the $\R_0$-dependent channel and the failure probability in terms of the ideal MWPM decoding channel, Lemma~\ref{lemma_failure} also holds for other decoding channels (with appropriately modified definitions).

\begin{proof}
To prove Eq.~\eqref{eq_doublebar}, we show that $p_{\overline{\N\circ\M}}(P) = p_{\overline{\overline\N\circ\overline\M}}(P)$ for any Pauli $X$ operator $P$.
Namely,
\begin{eqnarray}
p_{\overline{\N\circ\M}}(P) &=& \sum_{P': R_{\partial P'} = P} p_{\N\circ\M} (P')
= \sum_{\substack{N,M:\\ R_{\partial(NM)}= P}} p_{\N} (N)p_{\M} (M)\\
&=& \sum_{\substack{N',M':\\ R_{\partial(N'M')}= P}} 
\left(\sum_{N: R_{\partial N} = N'}p_{\N} (N)\right)\left(\sum_{M: R_{\partial M} = M'} p_{\M} (M)\right)
\label{eq_pauli_grouping}\\
&=& \sum_{\substack{N',M':\\ R_{\partial(N'M')}= P}} p_{\overline\N} (N') p_{\overline\M} (M')
= \sum_{P': R_{\partial P'} = P} p_{\overline\N\circ\overline\M} (P')
= p_{\overline{\overline\N\circ\overline\M}} (P),
\end{eqnarray}
where in Eq.~\eqref{eq_pauli_grouping} we group Pauli terms $N$'s and $M$'s according to their syndromes $\partial N'$ and $\partial M'$, respectively.
Then, we proceed by writing the failure probabilities for $\N\circ\M$ and $\overline\N\circ\overline\M$ as follows
\begin{eqnarray}
\fail(\N\circ\M) &=& \sum_{\substack{P: R_{\partial P} P\not\in\mathcal{G}}} p_{\N\circ\M} (P)
= \sum_{\substack{N,M:\\ R_{\partial(NM)}NM\not\in\mathcal{G}}} p_{\N} (N)p_{\M} (M),\\
\fail(\overline\N\circ\overline\M) 
&=& \sum_{\substack{P: R_{\partial P} P\not\in\mathcal{G}}} p_{\overline\N\circ\overline\M} (P)
= \sum_{\substack{N,M:\\ R_{\partial(NM)}R_{\partial N} R_{\partial M}\not\in\mathcal{G}}} p_{\N} (N)p_{\M} (M).
\end{eqnarray}
Since $R_{\partial(NM)} NM = (R_{\partial N} N) (R_{\partial M} M)(R_{\partial(NM)}R_{\partial M}R_{\partial M})$, then the condition $R_{\partial(NM)} NM\not\in\mathcal{G}$ implies that $R_{\partial N} N\not\in\mathcal{G}$ or $R_{\partial M} M\not\in\mathcal{G}$ or $R_{\partial(NM)}R_{\partial M}R_{\partial M}\not\in\mathcal{G}$, 
Thus,
\begin{eqnarray}
\fail(\N\circ\M) &=& \sum_{\substack{N,M:\\ R_{\partial(NM)}NM\not\in\mathcal{G}}} p_{\N} (N)p_{\M} (M)\\
&\leq&\sum_{\substack{N,M:\\ R_{\partial N}N\not\in\mathcal{G}}} p_{\N} (N)p_{\M} (M)
+ \sum_{\substack{N,M:\\ R_{\partial M }M\not\in\mathcal{G}}} p_{\N} (N)p_{\M} (M)
+ \mkern-18mu\sum_{\substack{N,M:\\ R_{\partial(NM)}R_{\partial N}R_{\partial M}\not\in\mathcal{G}}} 
\mkern-18mu p_{\N} (N)p_{\M} (M)\quad\quad\\
&=& \fail(\N) + \fail(\M) + \fail(\overline\N\circ \overline\M)
\end{eqnarray}
and we obtain the inequality in Eq.~\eqref{eq_failfirst}.
Similarly, the condition $R_{\partial(NM)} R_{\partial N}R_{\partial M}\not\in\mathcal{G}$ implies that $R_{\partial N} N\not\in\mathcal{G}$ or $R_{\partial M} M\not\in\mathcal{G}$ or $R_{\partial(NM)} NM\not\in\mathcal{G}$, and we can establish the inequality in Eq.~\eqref{eq_failsecond}.
\end{proof}

We say that two Pauli $X$ channels $\N$ an $\M$ are $\R_0$-equivalent and write $\N \sim \M$ iff their corresponding $\R_0$-dependent channels $\overline\N$ and $\overline\M$ are the same.
We emphasize that for any two $\R_0$-equivalent $\N,\M\in\mathbb{P}^X$ their resulting stabilizer syndrome distributions are the same, however terms in their Kraus representations might differ by some $X$-type gauge or logical Pauli operators.
Also note that $\N \sim \overline\N$, as $\overline{\overline{\N}} = {\overline{\N}}$.

Lastly, we introduce $\mathbb{N}_{\tau,\epsilon}$ to be a class comprising all the Pauli $X$ channels that are $\R_0$-equivalent to any $\tau$-bounded Pauli $X$ channel and whose failure probability is at most $\epsilon$, i.e.,
\begin{eqnarray}
\label{eq_localnoise}
\mathbb{N}_{\tau,\epsilon} = \{ \N \in \mathbb P^X
&\mathrel{|}& \textrm{$\fail(\N) \leq \epsilon$ and there exists a $\tau$-bounded $\M\in \mathbb P^X$ satisfying $\M\sim\N$}\}.
\end{eqnarray}

We finish this section with the following lemma about composing channels from the class $\mathbb{N}_{\tau,\epsilon}$.
\begin{lemma}[composition]
\label{lemma_composition}
For sufficiently small $\tau_1$ and $\tau_2$ there exist $\tau'$ and $\epsilon'$, such that the following inclusion holds
\begin{equation}
\mathbb{N}_{\tau_1,\epsilon_1}\circ\mathbb{N}_{\tau_2,\epsilon_2} \subseteq
\mathbb{N}_{\tau',\epsilon'}.
\end{equation}
In particular, we can have $\tau' = \tau_1+\tau_2$ and $\epsilon' = \epsilon_1+\epsilon_2+f(\tau_1)+ f(\tau_2)+ f(\tau_1+\tau_2)$, where $f(\cdot)$ is the upper bound specified in (decoding failure) Lemma~\ref{lemma_threshold}, provided that $\tau_1+\tau_2 < (2(\Delta_\mathrm{qub}-1))^{-2}$. 
\end{lemma}

\begin{proof}
Let $\N \in\mathbb N_{\tau_1,\epsilon_1}$ and $\M \in\mathbb N_{\tau_2,\epsilon_2}$.
Let $\N' = \{\sqrt{p_{\N'}(N')} N'\}_{N'}$ and $\M' = \{\sqrt{p_{\M'}(M')} M'\}_{M'}$ be two Pauli $X$ channels, which are $\R_0$-equivalent to $\N$ and $\M$ and which are $\tau_1$-bounded and $\tau_2$-bounded, respectively.
Using Eq.~\eqref{eq_doublebar} we have
\begin{equation}
\overline{\N'\circ\M'} = \overline{\overline{\N'}\circ\overline{\M'}} 
= \overline{\overline{\N}\circ\overline{\M}} = \overline{\N\circ\M},
\end{equation}
thus the composite channel $\N'\circ\M'$ is $\R_0$-equivalent to $\N\circ\M$.
We can straightforwardly show that $\N'\circ\M'$ is $(\tau_1+\tau_2)$-bounded, namely
\begin{eqnarray}
\sum_{P'\supseteq P} p_{\N'\circ \M'} (P')
&=& \sum_{\substack{N',M':\\ N'M' \supseteq P}} p_{\N'}(N') p_{\M'}(M')
= \sum_{P'\subseteq P} \sum_{\substack{N'\supseteq P', M' \supseteq P\setminus P':\\
N' \cap M' \cap P = \emptyset}} p_{\N'}(N') p_{\M'}(M')\\
&\leq& \sum_{P'\subseteq P} \left(\sum_{N'\supseteq P'} p_{\N'}(N')\right)
\left(\sum_{M' \supseteq P\setminus P'}  p_{\M'}(M')\right)\\
&\leq& \sum_{P'\subseteq P} \tau_1^{|P'|} \tau_2^{|P\setminus P'|} = (\tau_1 + \tau_2)^{|P|}.\quad\quad
\end{eqnarray}
This allows us to conclude that $\N\circ\M\in\mathbb{N}_{\tau',*}$.

Using the inequalities from Lemma~\ref{lemma_failure} and applying (decoding failure) Lemma~\ref{lemma_threshold} to $\N'$, $\M'$ and $\N'\circ\M'$ we obtain
\begin{eqnarray}
\fail(\N\circ\M) &\leq& \fail(\N) + \fail(\M) + \fail(\overline\N\circ\overline\M) 
= \epsilon_1 + \epsilon_2 + \fail(\overline{\N'}\circ\overline{\M'})\\
&\leq& \epsilon_1 + \epsilon_2 + \fail(\N') + \fail(\M') + \fail(\N'\circ\M')\\
&\leq& \epsilon_1 + \epsilon_2 + f(\tau_1) + f(\tau_2) + f(\tau_1+\tau_2),
\end{eqnarray}
which allows us to conclude that $\N\circ\M\in\mathbb{N}_{*,\epsilon'}$.
\end{proof}

\subsection{Single-shot QEC theorem}

\begin{theorem}
\label{thm_main}
For sufficiently small $\eta$ and $\tau$ there exist $\tau'$ and $\epsilon'$ satisfying
\begin{equation}
\lim_{\eta\rightarrow 0} \tau' = 0,\quad \lim_{L\rightarrow \infty} \epsilon' = \epsilon,
\end{equation}
such that the following inclusion holds
\begin{equation}
\mathbb{R}_\eta \circ \mathbb{N}_{\tau,\epsilon} \circ \Pi_0^{\mathcal S} 
\subseteq \mathbb{N}_{\tau',\epsilon'} \circ \Pi_0^{\mathcal S}.
\end{equation}
In particular, we can have 
$\tau' = 2^{2t-1}(\e\eta^{r})^t(\Delta_\mathcal{L}-1)^{t-1}$ and $\epsilon' = \epsilon+f(\tau)+f(\tau+\tau')$,
where $\e$ denotes Euler's number, $t= 1+c^{-1}_\mathrm{flu}$, $r = (2+2c_\mathrm{flu})^{-1}$ and $f(\cdot)$ is the upper bound specified in (decoding failure) Lemma~\ref{lemma_threshold}, provided that
$\eta < (\e - 1)^{1/r}(4\e^2(\Delta_\mathcal{L}- 1))^{-1/r}$ and $\tau + \tau' \leq (2(\Delta_\mathrm{qub}-1))^{-2}$.
\end{theorem}

Importantly, (single-shot QEC) Theorem~\ref{thm_main} says that the parameter $\tau'$ of the residual noise present in the system after performing one round of error correction with imperfect measurements can be made arbitrarily small only by reducing the parameter $\eta$ of the single-shot MWPM decoding channel, i.e.,
$\lim_{\eta\rightarrow 0} \tau' = 0$.
Moreover, the failure probability $\epsilon'$ for the residual noise increases by at most $f(\tau)+f(\tau+\tau')$ compared to $\epsilon$ and that increment can be made arbitrarily small just by increasing the linear size $L$ of the system, i.e.,
$\lim_{L\rightarrow \infty} \left(f(\tau)+f(\tau+\tau')\right)= 0$.
This establishes that the single-shot MWPM decoder is fault-tolerant and, subsequently, single-shot QEC is possible with the 3D~STC.

\begin{proof}
Let $\R = \{ \sqrt{p_\R(\mu)}R_{\partial(\phi+\fix\mu)} \Pi^\mathcal{G}_\phi\}_{\mu,\phi} \in \mathbb{R}_\eta$ and $\N = \{\sqrt{p_\N(N)} N\}_N \in \mathbb{N}_{\tau,\epsilon}$.
Let $\M\in\mathbb{P}^X$ be defined as follows
\begin{equation}
\M = \{ \sqrt{p_\R(\mu) p_\N(N)} R_{\partial N+\partial(\fix\mu)}  N\}_{\mu,N} = \{ \sqrt{p_\M(M) }M \}_{M},
\end{equation}
where we group the same Pauli terms and introduce
\begin{equation}
\label{eq_mchannel}
p_{\M}(M) = \sum_{\substack{\mu, N:\\ R_{\partial N + \partial(\fix\mu)} N = M}} p_\R (\mu) p_\N (N).
\end{equation}
Since for any stabilizer syndrome $\sigma$ and Pauli $X$ operator $N$ we have
$\Pi^\mathcal{S}_\sigma = \sum_{\phi : \partial\phi = \sigma} \Pi^\mathcal{G}_\phi$
and $\Pi^\mathcal{S}_\sigma N \Pi^\mathcal{S}_0 = \delta_{\sigma,\partial N} N\Pi^\mathcal{S}_0$, where
$\delta_{\sigma,\partial N}$ is the Kronecker delta, we obtain
\begin{eqnarray}
\R \circ \N \circ \Pi^\mathcal{S}_0 
&=&  \{ \sqrt{p_\R(\mu) p_\N(N)}
R_{\partial(\phi+\fix\mu)} \Pi^\mathcal{G}_\phi N \Pi^\mathcal{S}_0\}_{\mu,\phi,N}\\
&=&  \{ \sqrt{p_\R(\mu) p_\N(N)}
R_{\sigma+\partial(\fix\mu)} \Pi^\mathcal{S}_\sigma N \Pi^\mathcal{S}_0\}_{\mu,\sigma,N}
= \M \circ \Pi^\mathcal{S}_0.
\end{eqnarray}
Thus, proving the theorem is equivalent to showing that $\M \in \mathbb{N}_{\tau',\epsilon'}$.

Let us consider the $\R_0$-dependent channel for $\M$, i.e.,
\begin{equation}
\overline\M = \{ \sqrt{p_\M(M) }R_{\partial M} \}_{M} =
 \left\{ \sqrt{p_{\overline\M}(M') } M' \right\}_{M'},
\end{equation}
where we group the same Pauli terms and introduce
\begin{equation}
\label{eq_mbarchannel}
p_{\overline\M}(M') = \sum_{M: R_{\partial M} = M'} p_\M(M) 
= \sum_{\substack{\mu, N:\\ R_{\partial(\fix\mu)} = M'}} p_\R (\mu) p_\N (N)
= \sum_{\mu: R_{\partial(\fix\mu)} = M'} p_\R (\mu).
\end{equation}
In the second equality above we use Eq.~\eqref{eq_mchannel} and the fact that for $M = R_{\partial N + \partial(\fix\mu)} N$ we have $\partial M = \partial(\fix\mu)$.
By definition, $\overline\M$ is $\R_0$-equivalent to $\M$.
We now show that $\overline\M$ is $\tau'$-bounded by finding for any Pauli $X$ operator $M$ an upper bound on
\begin{equation}
\sum_{M'\supseteq M} p_{\overline\M}(M')
= \sum_{\mu: R_{\partial(\fix\mu)} \supseteq M} p_{\R}(\mu).
\end{equation}

We proceed by first proving for any Pauli $X$ operator $M$ the following inequality
\begin{eqnarray}
\label{eq_ineq2}
\sum_{\mu: R_{\partial(\fix\mu)} \supseteq M} p_{\R}(\mu)
\leq \sum_{s\geq t |M|} \sum_{W\in \cl (s,M)} \sideset{}{'}\sum_{\mu'\subseteq W: |\mu'| \geq rs} \sum_{\mu\supseteq \mu'} p_{\R}(\mu),
\end{eqnarray}
where $\cl(s,M)$ is the collection of connected covers of $M$ within the hypergraph $G_\mathcal{L}$ and $\sideset{}{'}{\textstyle\sum}$ denotes that we only sum over the measurement errors.
Let $\mu$ be a measurement error, for which $p_\R(\mu)$ appears on the left-hand side of the inequality.
We now show that the variables $s$, $W$ and $\mu'$ in Eq.~\eqref{eq_ineq2} can admit the following values
\begin{equation}
\label{eq_values}
W =\left((\fix\mu)\cup R_{\partial(\fix\mu)}\right)^{\cap M},\quad s = |W|,\quad \mu' = \mu \cap W.
\end{equation}
This, in turn, implies that the term $p_\R (\mu)$ also appears on the right-hand side and, subsequently, establishes the inequality, as all the terms in the inequality are non-negative.

Let $\hat\mu' = \hat\mu\cap W$ and $R' = R_{\partial(\fix\mu)} \cap W$.
We then have $R'\supseteq M$, $W = (\fix{\mu'})\cup R'$ and $|\fix{\mu'}| = |\mu'| + |\hat\mu'|$.
One can also verify that $\fix{\mu'}$ is a flux, as it satisfies the Gauss law, i.e., $\delta_R (\fix{\mu'}) = 0$, and that $\fix{\mu'}$ and $R'$ have the same stabilizer syndromes, i.e., $\partial(\fix{\mu'}) = \partial R'$.
Since $\hat\mu$ is the minimum-weight estimate of the measurement error $\mu$, we have $|\hat\mu'| \leq |\mu'|$; otherwise, $\mu'+\hat\mu'+\hat\mu$ would be an estimate of the measurement error $\mu$ of weight smaller than $\hat\mu$.
The assumption about confinement of the flux guarantees that for the flux $\fix{\mu'}$ there exists a Pauli $X$ operator $N$, such that $\partial N = \partial(\fix{\mu'})$ and $|N| \leq c_\text{flu}| \mu'+\hat\mu'|$.
Since $R_{\partial(\fix\mu)}$ is the minimum-weight recovery operator with the stabilizer syndrome $\partial(\fix\mu)$, we also have $|R'| \leq |N|$;
otherwise, $NR'R_{\partial(\fix\mu)}$ would be a recovery operator with the stabilizer syndrome $\partial(\fix\mu)$ of weight smaller than $R_{\partial(\fix\mu)}$.
By lower bounding the cardinality of $W$ as follows 
\begin{equation}
|W| = |\fix{\mu'}| + |R'| \geq c_\text{flu}^{-1}|N| + |R'| \geq (1+c_\text{flu}^{-1})|R'| \geq (1+c_\text{flu}^{-1})|M|
\end{equation}
we obtain that $s \geq t |M|$.
Thus, $s$ specified in Eq.~\eqref{eq_values} appears in~Eq.~\eqref{eq_ineq2}. 
Since $R_{\partial(\fix\mu)} \supseteq M$, we conclude that $W$ specified in Eq.~\eqref{eq_values} is a connected cover of $M$ within $G_\mathcal{L}$, i.e., $W\in \cl(s,M)$, and appears in~Eq.~\eqref{eq_ineq2}.
Also, by upper bounding the cardinality of $W$ as follows
\begin{equation}
|W| = |\mu'| + |\hat\mu'| + |R'| \leq 2|\mu'| + |N| \leq 2|\mu'| + c_\text{flu}|\fix{\mu'}| \leq 2(1+c_\text{flu})|\mu'|\\
\end{equation}
we obtain that $|\mu'| \geq rs$.
Thus, $\mu'$ specified in Eq.~\eqref{eq_values} appears in~Eq.~\eqref{eq_ineq2}.

Using the fact that $p_\R$ is $\eta$-bounded we obtain
\begin{eqnarray}
\sum_{s\geq t |M|} \sum_{W\in \cl (s,M)} \sideset{}{'}\sum_{\mu'\subseteq W: |\mu'| \geq rs} \sum_{\mu\supseteq \mu'} p_{\R}(\mu)
&\leq& \sum_{s\geq t |M|} \sum_{W\in \cl (s,M)} \sideset{}{'}\sum_{\mu'\subseteq W: |\mu'| \geq rs} \eta^{|\mu'|}\\
&\leq& \sum_{s\geq t|M|} \sum_{W\in \cl (s,M)} 2^{|W|} \eta^{rs}
\leq \sum_{s\geq t|M|} (2\eta^r)^s |\cl (s,M)|.
\end{eqnarray}
We can then use (connected covers)~Lemma~\ref{lemma_connected_cover} to bound the size of $\cl(s,M)$ and obtain
\begin{eqnarray}
\sum_{s\geq t|M|} (2\eta^r)^s |\cl (s,M)|
&\leq& \sum_{s\geq t|M|} (2\eta^r)^s \frac{(\e z)^s}{\e z^{|M|}} 
= \frac{((2\e\eta^{r})^t z^{t-1})^{|M|}}{\e(1-2\e z\eta^{r})}\\
&\leq& \frac{(2^{2t-1}(\e\eta^{r})^t (\Delta_\mathcal{L} - 1)^{t-1})^{|M|}}{\e-4\e^2 (\Delta_\mathcal{L}-1)\eta^{r}}
\leq (2^{2t-1}(\e\eta^{r})^t (\Delta_\mathcal{L} - 1)^{t-1})^{|M|},
\end{eqnarray}
where we use the upper bound on $\eta$ and $z = \max_{e\in E_{\mathcal{L}}} \sum_{v\in e}(\deg v - 1) \leq 2(\Delta_\mathcal{L} - 1)$, as any hyperedge of the hypergraph $G_\mathcal{L}$ contains at most two vertices.
Thus, $\overline\M$ is $\tau'$-bounded and, subsequently, $\M\in\mathbb{N}_{\tau',*}$.

To upper bound the failure probability for $\M$ we first observe that
\begin{eqnarray}
\fail(\M) &=& \sum_{M: R_{\partial M} M \not\in\mathcal G} p_\M (M) 
= \sum_{\substack{\mu, N:\\ R_{\partial(\fix\mu)} R_{\partial N+\partial(\fix\mu)} N \not\in\mathcal G}} p_\R (\mu) p_\N (N)\\
&=& \fail\left(\left\{\sqrt{p_\R (\mu)} R_{\partial(\fix\mu)} \right\}_\mu \circ \N\right)
= \fail(\overline\M \circ \N ).
\end{eqnarray}
Let $\N' = \{\sqrt{p_{\N'}(N')} N'\}_{N'}$ be a $\tau$-bounded Pauli $X$ channel, which is $\R_0$-equivalent to $\N$.
Then, (composition) Lemma~\ref{lemma_composition} implies that the composite channel $\overline\M \circ \N'$ is $(\tau+\tau')$-bounded.
Using the inequalities from Lemma~\ref{lemma_failure} and applying (decoding failure) Lemma~\ref{lemma_threshold} to $\N'$ and $\overline\M \circ \N'$ we obtain
\begin{eqnarray}
\fail\left(\overline\M \circ \N\right) 
&\leq& \fail(\overline\M) + \fail(\N) + \fail\left(\overline{\overline\M} \circ \overline\N\right)
= \fail(\N) + \fail\left(\overline{\overline\M} \circ \overline{\N'}\right)\\
&\leq& \epsilon + \fail(\overline\M) + \fail(\N') + \fail\left({\overline\M} \circ \N'\right)
\leq \epsilon + f(\tau) + f(\tau+\tau').
\end{eqnarray}
Thus, $\fail(\M) \leq \epsilon + f(\tau) + f(\tau+\tau')$ and subsequently $\M\in\mathbb{N}_{*,\epsilon'}$.
\end{proof}

Using (single-shot QEC) Theorem~\ref{thm_main} and (composition) Lemma~\ref{lemma_composition} we also conclude that after $n$ repeated rounds of Pauli $X$ noise and single-shot MWPM decoding channel we have
\begin{equation}
(\mathbb{R}_\eta \circ \mathbb{N}_{\tau,\epsilon})^n \circ \Pi_0^{\mathcal S} 
\subseteq \mathbb{N}_{\tau',\epsilon(n)} \circ \Pi_0^{\mathcal S},
\end{equation}
where $\epsilon(n) = n\epsilon+nf(\tau)+(n-1)f(\tau')+(2n-1)f(\tau+\tau')+(n-1)f(\tau+2\tau')$, provided that
$\tau+2\tau' < \left(2(\Delta_\text{qub}-1)\right)^{-2}$.
Thus, the residual noise in the 3D~STC does not accumulate in an uncontrollable way after we perform multiple rounds of single-shot QEC, and the logical information encoded in the 3D~STC will be protected for a long time.

\section{Discussion}

In our work we:
(i) introduced a new topological quantum code, the 3D~STC, which is a subsystem version of the toric code,
(ii) developed a single-shot decoding algorithm for the 3D~STC and numerically estimated its performance, 
and (iii) proved that single-shot QEC is possible with the 3D~STC.
We believe that the 3D~STC provides the canonical example of a topological quantum code demonstrating single-shot QEC.
Although we focused our presentation on the 3D~STC, we can define a subsystem version of the toric code in any dimension $d\geq 3$; see Appendix~\ref{app_ddim} for details.
Such higher-dimensional constructions are particularly appealing from the perspective of realizing self-correction and fault-tolerant non-Clifford logical gates~\cite{Kubica2015,Jochym-OConnor2021,Vasmer2021}.

The 3D~STC can be used to implement a fault-tolerant universal gate set $\{H, T, \mathrm{CNOT}\}$ without state distillation.
Namely, we can consider a version of the 3D~STC supported within a tetrahedral region, which has different transversal logical operations.
Since the 3D~STC is a CSS code, it has the transversal CNOT gate.
Moreover, depending on the state of the gauge qubits, it has either the transversal Hadamard gate $H$ or the transversal $T=\mathrm{diag}(1,e^{i\pi/4})$ gate.
We provide the details in the accompanying paper~\cite{Iverson2021}.

Lastly, we remark that the concepts of self-correction and single-shot QEC are closely connected.
Typically, the former implies the latter, as exemplified by the 4D toric code, however the converse is not immediate.
Thus, one fundamental problem worth exploring is whether the 3D~STC can give rise to a self-correcting topological phase.
We expect that by taking some local gauge operators of the 3D~STC we can construct Hamiltonians, which possess symmetry-protected topological order and thermal stability in the presence of $1$-form symmetries, analogously to the Hamiltonians arising from the 3D gauge color code~\cite{Roberts2017,Kubica2018,Roberts2020}.

\acknowledgements{
A.K.\ thanks H\'ector Bomb\'in, Sergey Bravyi, Daniel Gottesman, Alexei Kitaev and John Preskill for helpful discussions.
M.V.\ thanks Dan Browne and Jacob Bridgeman for valuable discussions. 
A.K.\ acknowledges funding provided by the Simons Foundation through the ``It from Qubit'' Collaboration. Research at Perimeter Institute is supported in part by the Government of Canada through the Department of Innovation, Science and Economic Development Canada and by the Province of Ontario through the Ministry of Colleges and Universities.
This work was completed prior to A.K. joining AWS Center for Quantum Computing.}

\appendix
\section{STC in higher dimensions}
\label{app_ddim}

In this appendix, we present a generalization of the STC to $d \geq 3$ dimensions.
We start by returning to and formalizing the picture of the 3D~STC that we used in Sec.~\ref{sec_glance}.
Next, we discuss the 4D~STC and show that it is related to the 4D stabilizer toric codes with a transversal four-qubit control-$Z$ gate described in Ref.~\cite{Jochym-OConnor2021}.
Finally, we sketch a proof that the STC defined on a $d$-dimensional hypercubic lattice with periodic boundary conditions has zero logical qubits.

\subsection{Recasting the STC}
\label{app_recast}

Let $\mathcal L$ denote a $d$-dimensional hypercubic lattice with periodic boundary conditions and linear length $L$, where $L$ is even. 
As in Sec.~\ref{sec_octahedral} we use $\mathcal L_i$ to denote the $i$-hypercubes of $\mathcal L$.
We can color the $d$-hypercubes of $\mathcal L$ in red and blue such that no two $d$-hypercubes sharing a $(d-1)$-hypercube have the same color.
We denote the corresponding sets of $d$-hypercubes as $\mathcal L_d^R$ and $\mathcal L_d^B$, respectively.
We place qubits on the $(d-2)$-hypercubes of $\mathcal L$.
Anticipating the definition of the STC, for any $i$-hypercube $\mu\in\mathcal L_i$, we define $\mathcal Q (\mu)$ to be the set of qubits on all the $(d-2)$-hypercubes that either contain $\mu$ if $i<d-2$ or are contained in $\mu$ if $i\geq d-2$.
Namely,
\begin{equation}
\mathcal Q (\mu) = 
\begin{cases}
\{ \omega \in \mathcal L_{d-2} | \omega \supset \mu \} & \text{if $i < d-2$,}\\
\{ \omega \in \mathcal L_{d-2} | \omega \subseteq \mu \} & \text{otherwise}. 
\end{cases}
\end{equation}
As in Sec.~\ref{sec_dual_lattice}, when we say that an operator is associated with $\delta$ we mean that it is supported on the set of qubits $\mathcal Q(\delta)$ and write $X (\delta) = \prod_{\omega \in \mathcal Q(\delta)} X_\omega$, where $X_\omega$ denotes a Pauli $X$ operator acting on the qubit on the $(d-2)$-hypercube $\omega$.
In the next section, we abuse the notation somewhat and talk about operators being supported on $i$-cells of various lattices, but it should be clear from the context what we mean.

The gauge group of the $d$-dimensional STC is generated by operators associated with pairs comprising $d$-hypercubes and $(d-3)$-hypercubes. 
Namely, for each $d$-hypercube $\mu \in \mathcal L_d^R$ and $(d-3)$-hypercube $\lambda \subset \mu$, we have the gauge operator 
\begin{equation}
\label{eq_gauge_ddim_X}
X (\mu, \lambda) = \prod_{\omega \in \mathcal Q (\mu) \cap \mathcal Q (\lambda)} X_\omega.
\end{equation}
Likewise for each $d$-hypercube $\nu \in \mathcal L_d^B$ and $(d-3)$-hypercube $\kappa \subset \nu$, we have the gauge operator 
\begin{equation}
\label{eq_gauge_ddim_Z}
Z (\nu, \kappa) = \prod_{\omega \in \mathcal Q (\nu) \cap \mathcal Q (\kappa)} Z_\omega. 
\end{equation}
Thus, the gauge group of the $d$-dimensional STC is
\begin{equation}
\mathcal G = \langle X (\mu, \lambda), Z (\nu, \kappa) \mathrel{|}
\mathcal L^R_{d}\ni \mu \supset \lambda \in \mathcal L_{d-3}, \mathcal L^B_{d}\ni \nu \supset \kappa \in \mathcal L_{d-3}\rangle. 
\end{equation}
Since the number of $i$-hypercubes contained in a $d$-hypercube is $2^{d-i}{d \choose d-i}$, there are $2^3 {d \choose 3} |\mathcal L_d|$ gauge generators of $\mathcal G$.
Note that not all of them are independent.
Also, one can check that the above definitions are equivalent to the pictorial representation shown in Eqs.~\eqref{eq_3DSTC_gauge}~and~\eqref{eq_3DSTC_stab} for $d=3$.
We illustrate the construction in 4D in Fig.~\ref{fig_4d_example}.

The stabilizer group of the $d$-dimensional STC is generated by operators associated with $d$-hypercubes and $(d-1)$-dimensional hyperplanes, namely
\begin{equation}
\label{eq_stabgroup_ddim}
\mathcal S = \langle X (\mu), Z (\nu), X (\pi_{i}), Z (\pi_{i}) \mathrel{|} \mu \in \mathcal L_d^R, \nu \in \mathcal L_d^B, i \in \{1,\ldots,d\} \rangle,
\end{equation}
where $X (\pi_{i})$ denotes a product of Pauli $X$ operators acting on all qubits contained in the $(d-1)$-hyperplane $\pi_{i}$ perpendicular to the Cartesian axis $\hat x_i$.
Using an inductive argument over the dimension $d$ one can show that the stabilizer generators specified in Eq.~\eqref{eq_stabgroup_ddim} can be constructed from the gauge generators specified in Eqs.~\eqref{eq_gauge_ddim_X}~and~\eqref{eq_gauge_ddim_Z}.

\begin{figure}
(a)\hspace{0.02\textwidth}\includegraphics[width=0.25\textwidth]{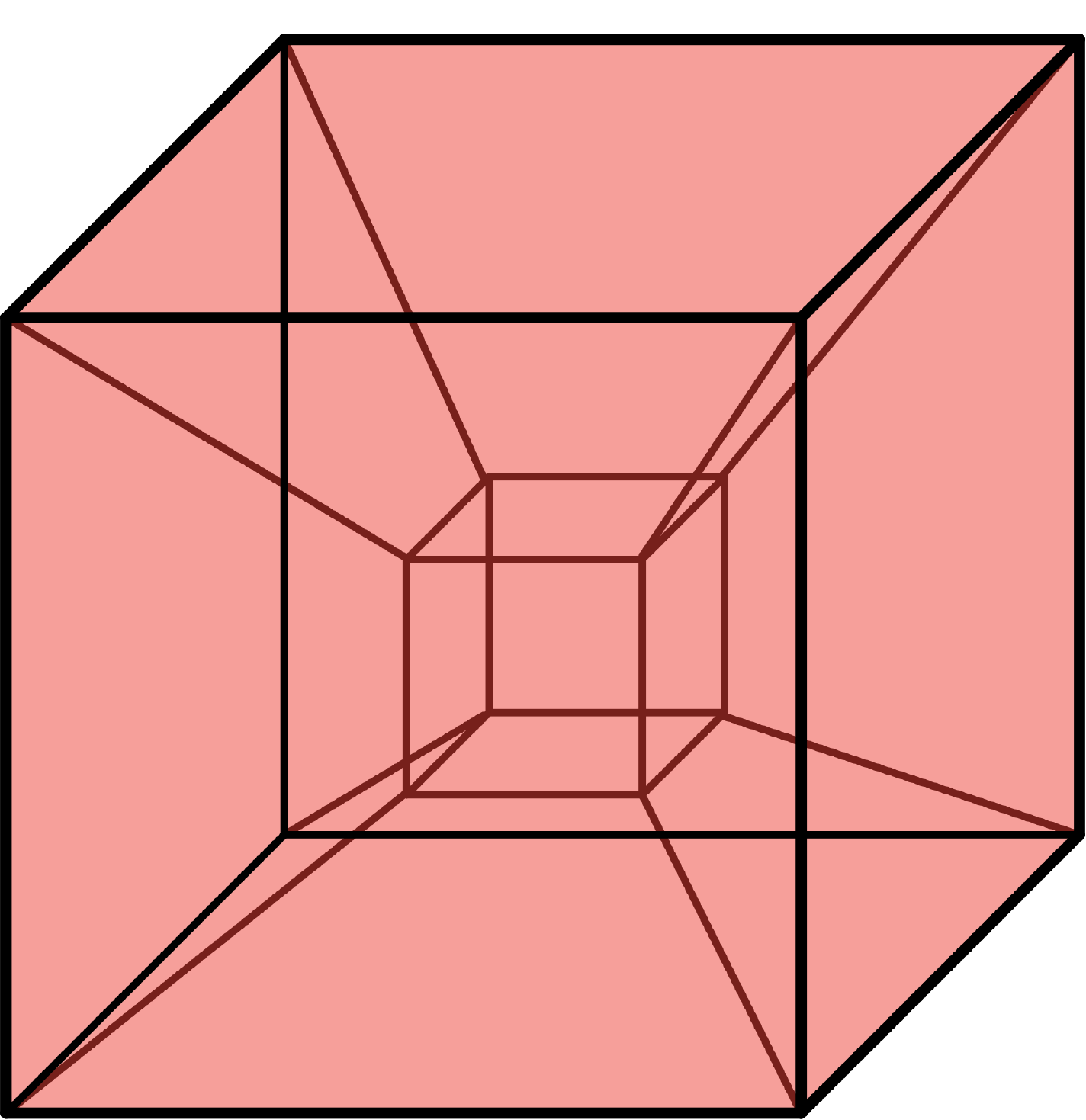}
\quad\quad\quad
(b)\hspace{0.02\textwidth}\includegraphics[width=0.25\textwidth]{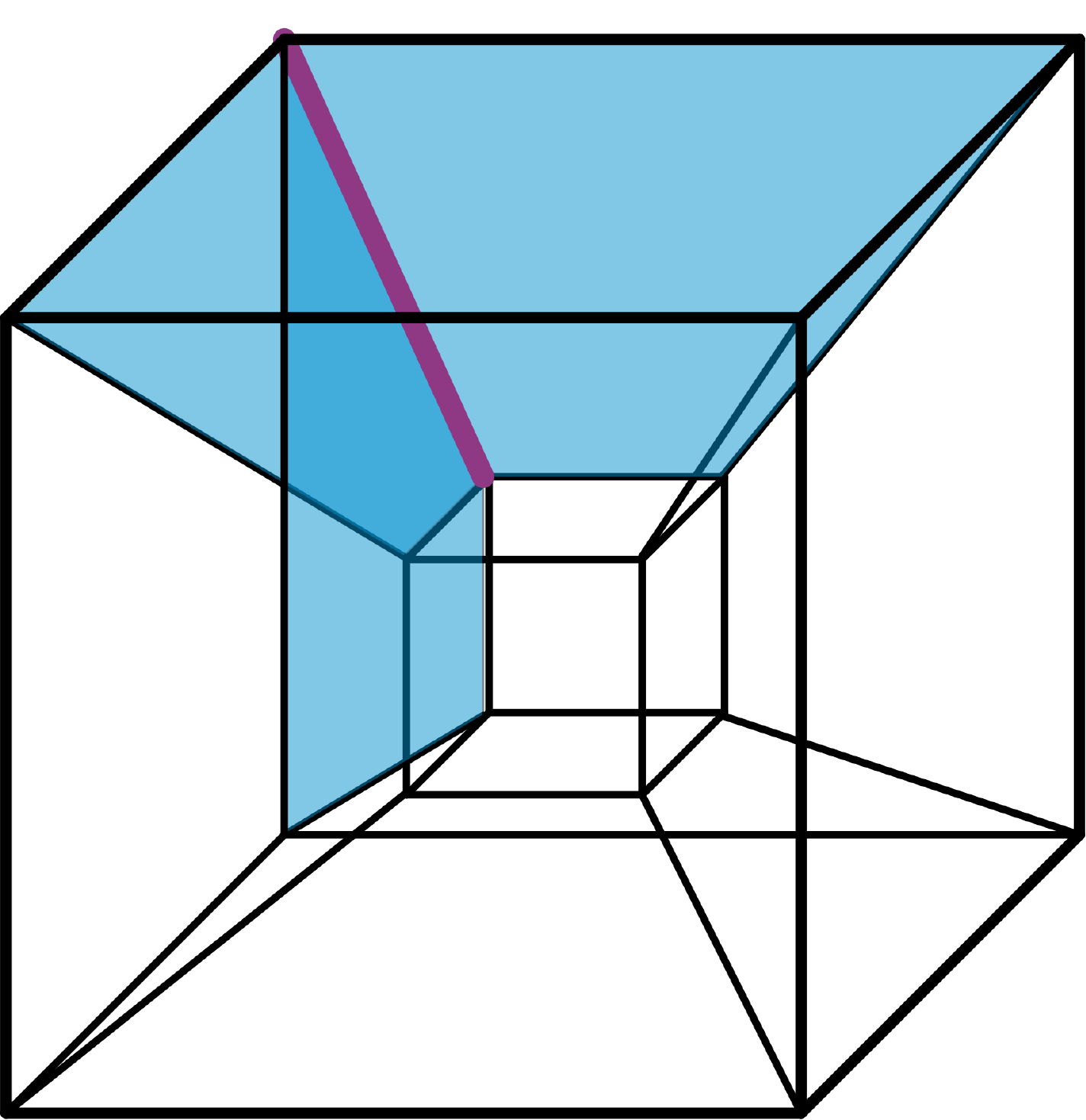}
\caption{
The support of (a) an $X$-type stabilizer and (b) $Z$-type gauge operator in the 4D~STC with qubits placed on faces.
The $X$-type stabilizer acts on the 24 qubits  contained in a red $4$-hypercube.
The $Z$-type gauge operator associated with a blue $4$-hypercube and an edge contained in this hypercube (purple) acts on the three qubits (shaded in blue) that are contained in the hypercube and also contain the edge.
}
\label{fig_4d_example}
\end{figure}

We now show that the gauge and stabilizer operators defined above commute.
Suppose that the edges of $\mathcal L$ have unit length and that one of the vertices of $\mathcal L$ is at the origin\footnote{We only need to consider the neighborhood around the origin as $\mathcal L$ is translationally invariant.}.
Every hypercube containing the origin as one of its vertices can be uniquely specified by a string ${\bf c} \in \{-1,0,1\}^{\times d}$; the dimensionality of this hypercube is $|{\bf c}| = \sum_{i}|c_i|$.
We define the support of a hypercube to be the positions of its non-zero coordinates, i.e., $\supp {\bf c} = \{ i | c_i \neq 0 \}$.
We write $\bf a \subseteq b$ whenever the non-zero entries of $\bf a$ match the non-zero entries of $\bf b$, i.e., $a_i = b_i$ for all $i\in \supp \bf a$.

Without loss of generality, consider the $X$-type gauge operator associated with the $d$-hypercube specified by ${\bf a} = (1,\ldots,1)$ and the $(d-3)$-hypercube ${\bf c} = ( 0,0,0,1,\ldots,1 )$. 
This operator has weight three, as it is supported on the following three qubits: $(1,0,0,1,\ldots,1)$, $(0,1,0,1,\ldots, 1)$ and $(0,0,1,\ldots,1)$.
The only stabilizer generators that have non-trivial overlap with $X ({\bf a}, {\bf c})$ are those associated with neighboring $d$-hypercubes and those associated with the hyperplanes $\pi_1$, $\pi_2$ and $\pi_3$ that contain the origin.
Without loss of generality, consider the $Z$-type stabilizer associated with $d$-hypercube ${\bf b} = ( -1, 1, \ldots, 1)$.
The operators $Z ({\bf b})$ and $X ({\bf a}, {\bf c})$ commute, as they overlap on the following two qubits: $(0,1,0,1,\ldots,1)$ and $(0,0,1,\ldots,1)$.
In addition, the support of $X ({\bf a}, {\bf c})$ contains exactly two $(d-2)$-hypercubes in each of the relevant hyperplanes, so the hyperplane stabilizers commute with $X ({\bf a}, {\bf c})$.

\subsection{Relation to the toric code in $d=4$ dimensions}

Let us now consider a special case of the STC on the four-dimensional hypercubic lattice $\mathcal L$.
There, the qubits are on faces, local $X$- and $Z$-type stabilizer generators are associated with red and blue 4-hypercubes, and gauge generators are associated with pairs comprising a 4-hypercube and an edge contained in the 4-hypercube.
The local stabilizer generators have weight 24 and the gauge generators have weight three.
We note that the vertices of the 4D hypercubic lattice $\mathcal L$ can be colored in green and yellow such that no two vertices sharing an edge have the same color.

We now explain how to obtain the lattice and codes considered in Ref.~\cite{Jochym-OConnor2021} from $\mathcal L$ (the four codes are the same so we need only show that one of them is a gauge-fixing of the 4D~STC).
The first step is to take the dual of $\mathcal L$.
This gives a lattice $\mathcal L^*$ formed from $\mathcal L$ by exchanging vertices, edges, faces, 3-cells and 4-cells respectively with 4-cells, 3-cells, faces, edges, and vertices.
It transpires that $\mathcal L^*$ is also the 4D hypercubic lattice. 
Next, we apply a further geometric transformation to $\mathcal L^*$ called alternation. 
We select the blue vertices of $\mathcal L^*$ and connect them with new edges if they are part of the same face in $\mathcal L^*$. 
This creates a new tessellation $\mathcal L^*_{\rm alt}$ whose vertices correspond to the blue vertices of $\mathcal L^*$. 
The red vertices of $\mathcal L^*$ are mapped to 4-cells in $\mathcal L^*_{\rm alt}$, the faces of $\mathcal L^*$ are mapped to edges, the cubic volumes of $\mathcal L^*$ are mapped to alternated cubes (tetrahedra) and the 4-hypercubes of $\mathcal L^*$ are mapped to 4-hyperoctahedra.
The lattice $\mathcal L^*_{\rm alt}$ is a regular tessellation of 4D Euclidean space by 4-hyperoctahedra\footnote{
This tessellation is often called the 16-cell honeycomb.}.
The final step is to take the dual of $\mathcal L^*_{\rm alt}$, which gives a tessellation of 4D Euclidean space by octaplexes\footnote{This tessellation is often called the 24-cell honeycomb.}.
We denote the dual of $\mathcal L^*_{\rm alt}$ by $\mathcal L_{\rm opx}$.
We illustrate the evolution of a gauge operator through the lattice transformations we have just described in Fig.~\ref{fig_4d_gauge_transform}.

\begin{figure}
(a)\hspace{0.01\textwidth}\includegraphics[width=0.19\textwidth]{figures/4d_gauge_primal.pdf}
\hspace{0.01\textwidth}
(b)\hspace{0.01\textwidth}\includegraphics[width=0.19\textwidth]{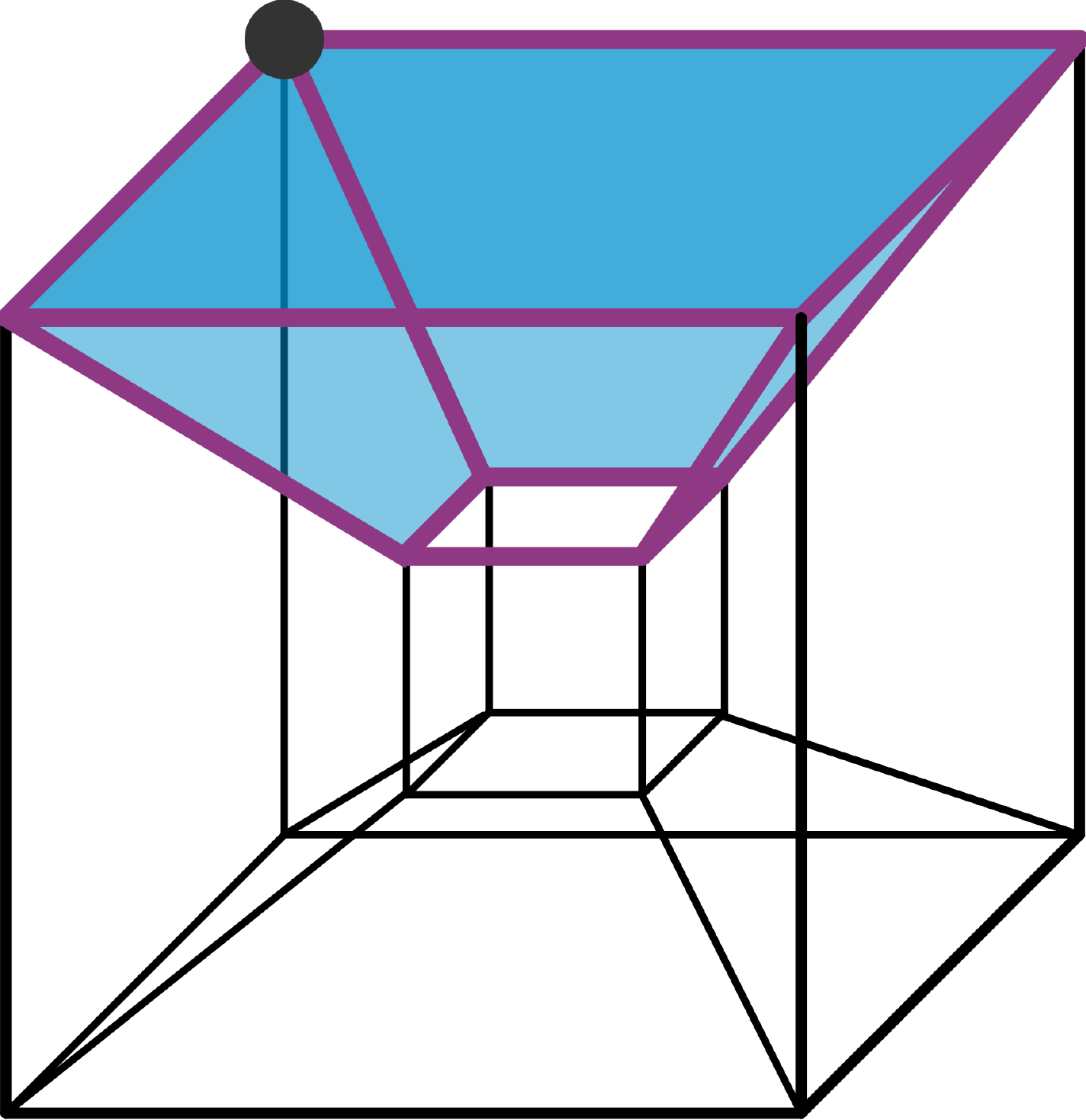}
\hspace{0.01\textwidth}
(c)\hspace{0.01\textwidth}\includegraphics[width=0.19\textwidth]{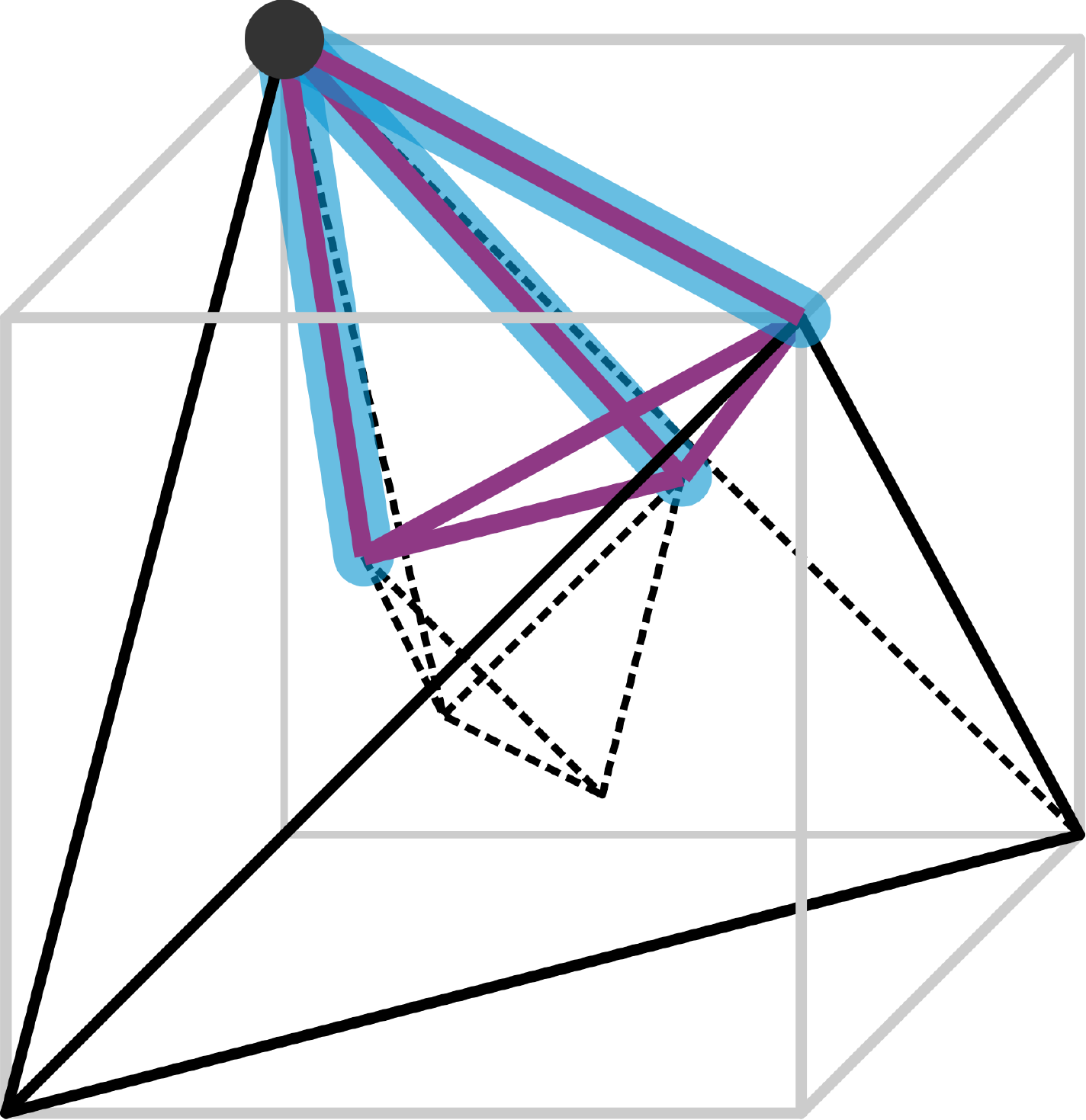}
\hspace{0.01\textwidth}
(d)\hspace{0.01\textwidth}\includegraphics[width=0.19\textwidth]{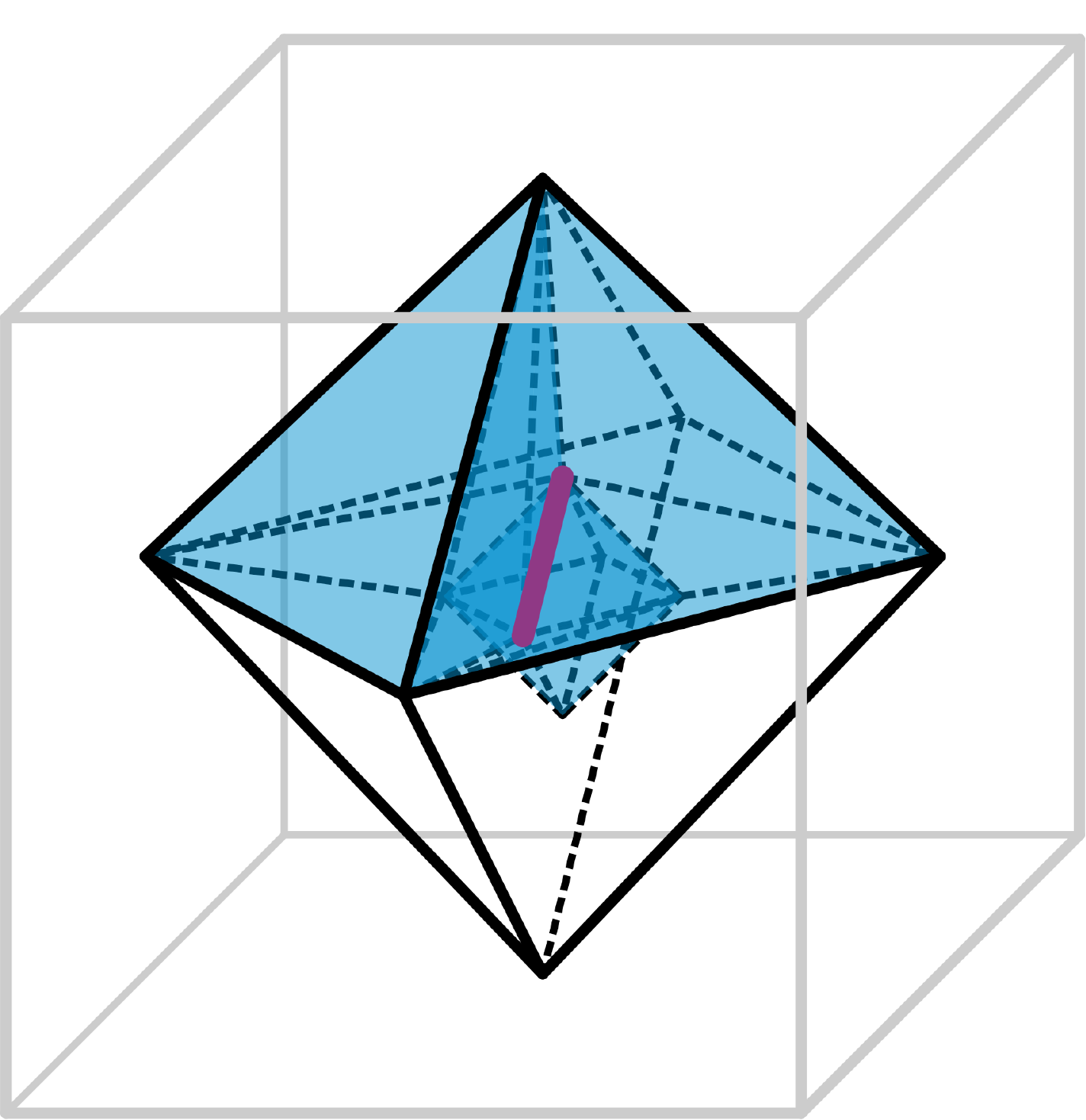}
\caption{
The evolution of $Z$-type gauge generator through three lattice transformations.
(a) A 4-hypercube of $\mathcal L$, the original 4D~STC lattice,
with qubits on faces. 
We illustrate the support (shaded blue) of the gauge operator associated with the 4-hypercube and an edge contained in this hypercube (purple).
(b) A 4-hypercube of $\mathcal L^*$, the dual of $\mathcal L$ (qubits are still on faces).
The gauge operator is now associated with a vertex (dark gray circle)
and a 3-hypercube containing this vertex (outlined in purple).
The faces in its support are shaded in blue.
(c) A 4-cell of $\mathcal L^*_{\rm alt}$, the alternation of $\mathcal L^*$, with qubits on edges.
We show the outer 3-cell of (b) as a guide to the eye.
Here, the gauge operator is associated with a vertex (dark gray circle) and a 3-cell containing this vertex (outlined in purple).
The edges in its support are highlighted in blue.
(d) A 4-cell of $\mathcal L_{\rm opx}$, the dual of $\mathcal L^*_{\rm alt}$, with qubits on 3-cells. 
Here, the gauge operator is associated with the 4-cell and an edge contained in this 4-cell (purple).
The 3-cells in its support are shaded in blue.
To avoid clutter, in (c) and (d) we draw only a subset of the edges of the 4-cells.
}
\label{fig_4d_gauge_transform}
\end{figure}

Tracking the faces, 4-hypercubes and edges of $\mathcal L$ through the three transformation steps, we find that in $\mathcal L_{\rm opx}$, qubits are placed on 3-cells, $X$-type stabilizer generators are associated with red vertices, and $Z$-type gauge operators are associated with pairs comprising an 4-cell and an edge contained in the 4-cell.
We note that the vertices of $\mathcal L_{\rm opx}$ can be colored in red, green and yellow such that no two vertices sharing an edge have the same color.
In particular, these vertex colors are inherited from the red 4-hypercubes, green vertices and yellow vertices of $\mathcal L$.
Consequently every $Z$-type gauge operator in $\mathcal L_{\rm opx}$ is associated with a 4-cell and a GY-edge.
In Ref.~\cite{Jochym-OConnor2021} the authors consider an octaplex tessellation with qubits on 3-cells and vertices 3-colored in, e.g.\ red, green and yellow.
One of their 4D stabilizer toric codes has $X$-type stabilizer generators associated with red vertices, and $Z$-type stabilizer generators associated with pairs comprising a 4-cell and a $GY$ edge. 
This is exactly the construction we just obtained starting from $\mathcal L$.
Furthermore, we have the inclusions
\begin{equation}
\mathcal S \leq \mathcal S'_{\rm 4DST} \leq \mathcal G,
\end{equation}
where $\mathcal S$ and $\mathcal G$ are the stabilizer and gauge groups of the 4D~STC, and $\mathcal S'_{\rm 4DST}$ is a group generated by stabilizer operators and $X$-type logical operators of the 4D stabilizer toric code.
Therefore, the logical $\ket{\overline +}$ state of the 4D stabilizer toric code is also in the code space of the 4D~STC, with gauge qubits in a certain state.
Moreover, using the procedure of gauge fixing we can map a state in the code space of the 4D~STC to the logical $\ket{\overline +}$ state of the 4D stabilizer toric code.

We have shown that the 4D~STC can be gauge fixed to a code with a transversal 4-qubit control-$Z$ gate. 
We also expect that the 4D~STC will exhibit single-shot QEC and therefore may be an attractive candidate for realizing universal fault-tolerant quantum computation, especially in architectures where non-local connections of qubits are available.

\subsection{Calculating the number of logical qubits}

For simplicity, let us consider the STC on the $d$-dimensional hypercubic lattice $\mathcal L$ with periodic boundary conditions and the linear size $L$, where $L$ is even.
Then, we have
\begin{equation}
\label{eq_hypercubes}
|\mathcal L_i | = L^d {d\choose i}.
\end{equation}
In particular, the number of physical qubits in the $d$-dimensional STC is
\begin{equation}
\label{eq_qubitsddim}
N=|\mathcal L_{d-2}| = L^d {d\choose 2}.
\end{equation}
Note that stabilizer generators can be associated with the $d$-hypercubes and non-contractible $(d-1)$-dimensional hyperplanes within $\mathcal L$.
Since there are two relations between them, i.e.,
\begin{equation}
\prod_{\delta \in \mathcal{L}_d^R} X(\delta) = \prod_{\delta \in \mathcal{L}_d^B} Z(\delta) = I,
\end{equation}
we thus obtain that the number of independent generators of the stabilizer group is
\begin{equation}
\label{eq_stabddim}
\log_2 |\mathcal S| = |\mathcal L_d| + d - 2 = L^d + d - 2.
\end{equation}
Note that gauge generators, which are associated with $d$-hypercubes and $(d-3)$-hypercubes, are not independent.
Rather, they have to satisfy three types of relations.
Relations of the first type, which we call $\mathcal R_3$, arise for every $(d-3)$-hypercube $\mu\in\mathcal L_{d-3}$ from the following identities
\begin{equation}
\forall \mu\in\mathcal L_{d-3}: \prod_{\delta\in\mathcal L_d^R: \delta \supset \mu} X(\delta, \mu)
= \prod_{\delta\in\mathcal L_d^B: \delta \supset \mu} Z(\delta, \mu) = I.
\end{equation}
Note that not all relations in $\mathcal R_3$ are independent.
They, as well, have to satisfy certain relations, which we call $\mathcal R_4$.
They arise for every $(d-4)$-hypercube $\nu\in\mathcal L_{d-4}$ from the following identities
\begin{equation}
\forall \nu\in\mathcal L_{d-4}: 
\prod_{\mu\in\mathcal L_{d-3}: \mu \supset\nu}\prod_{\delta\in\mathcal L_d^R: \delta \supset \mu} X(\delta, \mu)
= \prod_{\mu\in\mathcal L_{d-3}: \mu \supset\nu} \prod_{\delta\in\mathcal L_d^B: \delta \supset \mu} Z(\delta, \mu) = I.
\end{equation}
But relations $\mathcal R_4$ are not independent, and so on.
In general, we have
\begin{equation}
|\mathcal R_i| = 2 |\mathcal L_{d-i}| = 2L^d {d\choose i}.
\end{equation}
Proper counting of independent relations of the first type gives the following alternating sum
\begin{eqnarray}
|\mathcal R_3| - |\mathcal R_4| +\ldots + (-1)^{d+1} |\mathcal R_d|
&=& 2L^d \sum_{i=3}^d (-1)^{i+1}{d\choose d-i}\\
&=& 2L^d \left( {d\choose d} - {d\choose d-1} +{d\choose d-2} \right),
\end{eqnarray}
where we use the identity $\sum_{i=0}^d(-1)^i{d\choose d-i} = 0$.
Relations of the second type, which we call $\mathcal R'_4$, arise for every $d$-hypercube $\delta'\in\mathcal L_d$ and every $(d-4)$-hypercube $\mu'\in\mathcal L_{d-4}$ contained in $\delta'$ from the following identities
\begin{eqnarray}
\forall \mathcal L_{d-4}\ni\mu'\subset \delta' \in \mathcal L^R_d &:& \prod_{\mu\in\mathcal L_{d-3}: \mu'\subset\mu\subset\delta' } X(\delta', \mu) = I,\\
\forall \mathcal L_{d-4}\ni\mu'\subset \delta' \in \mathcal L^B_d &:& \prod_{\mu\in\mathcal L_{d-3}: \mu'\subset\mu\subset\delta' } Z(\delta', \mu) = I.
\end{eqnarray}
Note that not all relations in $\mathcal R'_4$ are independent.
They, as well, have to satisfy certain relations, which we call $\mathcal R'_5$.
They arise for every $d$-hypercube $\delta'\in\mathcal L_d$ and every $(d-5)$-hypercube $\nu'\in\mathcal L_{d-5}$ contained in $\delta'$ from the following identities
\begin{eqnarray}
\forall \mathcal L_{d-5}\ni\nu'\subset \delta' \in \mathcal L^R_d &:& \prod_{\mu'\in\mathcal{L}_{d-4}: \mu'\supset \nu'}
\prod_{\mu\in\mathcal L_{d-3}: \mu'\subset\mu\subset\delta' } X(\delta', \mu) = I,\\
\forall \mathcal L_{d-5}\ni\nu'\subset \delta' \in \mathcal L^B_d &:& \prod_{\mu'\in\mathcal{L}_{d-4}: \mu'\supset \nu'}
\prod_{\mu\in\mathcal L_{d-3}: \mu'\subset\mu\subset\delta' } Z(\delta', \mu) = I.
\end{eqnarray}
But relations $\mathcal R'_5$ are not independent, and so on.
In general, we have
\begin{equation}
|\mathcal R'_i| = 2^i |\mathcal L_{d}| {d\choose d-i} = 2^iL^d {d\choose d-i}.
\end{equation}
Finally, we arrive at the fact that not all relations $\mathcal R'_d$ are independent. In fact, they satisfy one relation for every $d$-hypercube.
Proper counting of independent relations of the second type gives the following alternating sum
\begin{eqnarray}
|\mathcal R'_4| - |\mathcal R'_5| +\ldots + (-1)^{d} |\mathcal R'_d| + (-1)^{d+1} |\mathcal L_d|
&=& L_{d} \left(\sum_{i=4}^d (-2)^i{d \choose d-i} + (-1)^{d+1}\right)\\
&=& L^{d} \left( - 2^0 {d \choose 0} + 2^1 {d \choose 1} - 2^2 {d \choose 2} + 2^3 {d \choose 3} \right),\quad\quad\quad
\end{eqnarray}
where we use the identity $\sum_{i=0}^d(-2)^i{d\choose d-i} = (-1)^d$.
Relations of the third type arise for non-contractible $(d-1)$-hyperplanes within $\mathcal L$, and thus there are $d$ independent ones.
Lastly, relations of the first and second type are not independent, as there are two relations between them---one for all $X$-type relations and one for all $Z$-type relations.
Once all the independent relations have been properly accounted for, the number of independent generators of the gauge group is
\begin{eqnarray}
\label{eq_gaugeddim1}
\log_2 |\mathcal G| &=& 2^3 |\mathcal L_{d}| {d\choose d-3} - 2L^d \left( {d\choose d} - {d\choose d-1} +{d\choose d-2} \right)\\
&& -L^{d} \left(- 2^0 {d \choose 0} + 2^1 {d \choose 1} - 2^2 {d \choose 2} + 2^3 {d \choose 3} \right) - d + 2\\
&=& L^d \left(2{d\choose 2} -1 \right) - d + 2
\label{eq_gaugeddim2}
\end{eqnarray}
By combining Eqs.~\eqref{eq_qubitsddim},~\eqref{eq_stabddim}~and~\eqref{eq_gaugeddim1}-\eqref{eq_gaugeddim2} we obtain that there are no logical qubits encoded into the STC on the $d$-dimensional hypercubic lattice with periodic boundary conditions, i.e.,
\begin{equation}
K = N -\frac{1}{2}\left( \log_2 |\mathcal G| + \log_2 |\mathcal S| \right) = 0.
\end{equation}

\bibliography{biblio_STC,biblio_misc}

\end{document}